 \documentclass[journal]{IEEEtran}                                                        

\IEEEoverridecommandlockouts                              
\usepackage{amsmath}
\usepackage{xr}
\usepackage{graphics}
\usepackage{graphicx}
\usepackage{setspace}
\usepackage{amssymb}
\usepackage{amsthm}
\usepackage[usenames,dvipsnames,svgnames,table]{xcolor}
\usepackage{tikz}
\usepackage{multirow}
\usepackage{subcaption}

\def \bi {\begin{itemize}\item}
\def \ei {\end{itemize}}
\def \be {\begin{equation}}
\def \ee {\end{equation}}
\def \ba {\begin{aligned}}
\def \ea {\end{aligned}}

\def \l {\bigg{(}}
\def \r {\bigg{)}}

\newcommand{\LyX}{L\kern-.1667em\lower.25em\hbox{Y}\kern-.125emX\spacefactor1000}

\title{Is Non-Neutrality Profitable for the Stakeholders of the Internet Market? - Part II}

\author{Mohammad Hassan Lotfi, Saswati Sarkar, and George Kesidis\thanks{ M. H. Lotfi and Saswati Sarkar are with the Department of ESE at University of Pennsylvania, Philadelphia,
		PA, U.S.A. Their email addresses are lotfm@seas.upenn.edu and
		swati@seas.upenn.edu respectively. George Kesidis is with the school of EECS of Pennsylvania State University, University Park, PA, U.S.A. His Email Address is gik2@psu.edu.}}

\begin{document}
\maketitle\newtheorem{lemma}{Lemma}
\newtheorem{note}{Note}
\newtheorem{property}{Property}
\newtheorem{theorem}{Theorem}
\newtheorem{definition}{Definition}
\newtheorem{corollary}{Corollary}
\newtheorem{remark}{Remark}
\newtheorem{assumption}{Assumption}

\begin{abstract}
In this part of the paper, we obtain analytical results for the case that transport costs are not small (complement of Part I), and combine them with the results in Part I of the paper to provide general results for all values of transport costs. We show that, in general, if an SPNE exists, it would be one of the five possible strategies each of which we explicitly characterize. We also prove that when EUs have sufficiently high inertia for at least one of the ISPs, there exists a unique SPNE with a non-neutral outcome in which both of the ISPs are active, and the CP offers her content with free quality on the neutral ISP and with premium quality on the non-neutral ISP. Moreover, we show that an SPNE does not always exist. We also analyze a benchmark case in which both ISPs are neutral, and prove that there exists a unique SPNE in which the CP offers her content with free quality on both ISPs, and both ISPs are active. We also provide extensive numerical results and discussions for all ranges of transport costs. Simulation results suggest that if the SPNE exists, it would be unique. In addition, results reveal that the neutral ISP receives a lower payoff and the non-neutral ISP receives a higher payoff (most of the time) in a non-neutral scenario. However, we also identify scenarios in which the non-neutral ISP loses payoff by adopting non-neutrality. In addition, we show that a non-neutral regime yields a higher welfare for EUs than a neutral one if the market power of the non-neutral ISP is small, the sensitivity of EUs (respectively, the CP) to the quality is low (respectively, high), or a combinations of these factors.
\end{abstract}

\section{Introduction}

In the first part of the paper, we focused on the case that the transport costs are lower than a threshold, i.e. the case that EUs can switch between ISPs relatively easy. We proved that when transport costs are small and if the SPNE exists, then the equilibrium strategy of the CP would be to offer her content on only the non-neutral ISP. Depending on the values of the transport costs, at the SPNE, either both ISPs would be active (have positive share of EUs), or only the non-neutral ISP would be active.  We also proved that if the transport costs are sufficiently small (smaller than another threshold), then the SPNE exists and is unique. In this case, at the SPNE outcome of the market, only the non-neutral ISP would be active and the neutral one would be driven out of the market.

In this part, we use the same system model as in the previous part of the paper. However, we  consider the complement of the previous part, where  transport costs are not small and EUs cannot easily switch between ISPs.  This case often happens in practice in the Internet market, e.g. when ISPs bundle Internet access with other services (e.g. cable, phone). In this case, an EU may incur additional expenses for other services if she buys  Internet access from another ISP. Another example of high inertia of EUs is the case in which EUs require different devices to access the Internet through different ISPs (e.g. different devices for cable and DSL services), i.e. high set up costs.

\subsection*{Analytical Results}
First, we obtain analytical results for the case that transport costs are not small (complement of the previous part), and combine them with the results in the first part of the paper to provide general results for all values of transport costs. We show that, in general, if an SPNE exists, it would be one of the five possible strategies each of which we explicitly characterize. In some of these strategies, the CP offers her content on only the non-neutral ISP, and in the rest she offers her content on both ISPs. In addition, in one of the outcomes, all EUs pay the Internet access fee to the non-neutral ISP, i.e. the neutral ISP is driven out of the market. However, in the rest, both ISPs receive a positive share of EUs, i.e. both ISPs are active.  In addition, we provide specific instances with no SPNE. This shows that an SPNE does not always exist. Our analysis reveals that the market would have a more diverse set of equilibrium outcomes in comparison to Part I of the paper \cite{PartI}. 

Second, we focus on a special case in which EUs have sufficiently high inertia for at least one of the ISPs. We prove that in this case, there exists a unique SPNE with a non-neutral outcome, and we explicitly characterize the SPNE. In the unique SPNE,  both of the ISPs are active, and the CP offers her content with free quality on the neutral ISP and with premium quality on the non-neutral ISP. 

Third, we consider a benchmark case in which both ISPs are neutral. In this case, we prove that regardless of the transport costs, there exists a unique SPNE, in which  the CP offers her content over both ISPs with free quality, and both ISPs would be  active. We use the results of this case, as a benchmark for comparing with the results of our model.


\subsection*{Numerical Results}
We provide numerical results for all ranges of transport costs.  Numerical results confirm our theoretical results that when the transport costs are small  enough, then the SPNE exists and is unique. They also confirm that when the transport costs are large enough, then the SPNE with a non-neutral outcome  uniquely exists. Numerical results reveal that if the inertia of EUs for ISPs is between these two extreme cases (high and low inertias) but still small enough, the game has an SPNE outcome in which both of the ISPs are active, but the CP offers her content with premium quality and only on the non-neutral ISP. Numerical results also reveal that if the inertia of EUs for ISPs is between the two extreme cases and large enough, then the game has no SPNE. Results of simulation over large sets of parameters also suggest that in all scenarios, the SPNE is unique if it were to exist.  

Numerical results reveal that  the neutral ISP loses payoff in all of the SPNE outcomes in comparison to the benchamrk case.  
In addition,  for a wide range of parameters, the non-neutral ISP receives a better payoff under a non-neutral scenario. This implies that it is beneficial for ISPs to operate as non-neutral, if they have the choice.   However, switching to a non-neutral regime is \emph{not} always profitable for ISPs. If EUs or the CP are not sensitive to the quality of the content delivered and the market power of the non-neutral ISP is small, then ISPs are better off staying neutral. 

Results also reveal that the sensitivity of the EUs and the CP, and the market power of ISPs substantially influences the welfare of EUs (EUW) in neutral and non-neutral scenarios. EUW would be higher in a non-neutral setting if (i) the market power of the non-neutral ISP is low, (ii) the sensitivity of the CP to the quality is high, or (iii) EUs are not very sensitive to the quality, or a combination of these conditions. In these cases a cheaper Internet access fee would be charged to the EUs by the non-neutral ISP which yields a higher EUW.   


\subsection*{Organization of the Paper}
Note that the model herein is similar to the one presented in Part I. The rest of the paper is organized as follows. We present the results under the assumption that the transport costs are not small for the case of one neutral and one non-neutral ISPs in Section \ref{section:theory}. In Section     \ref{section:proofsBenchamrk}, we present the results for the benchmark case, i.e. both ISPs neutral. All proofs are presented in the Appendices. In Section~\ref{section:discussion}, we summarize and discuss about the key results of the paper. We  provide numerical examples  spanning Part I and II in Section~\ref{section:numericalresults}. Finally, we comment on some of the assumptions of the model and their generalizations in Section~\ref{section:implicationAssum}.

\begin{remark}
We refer to the equations, theorems, lemmas, etc of the previous part using ``I" suffix, e.g. Theorem 1-I and Equation (1)-I refers to the Theorem 1 and Equation (1) of Part I, respectively. 
\end{remark}

\section{The Sub-Game Perfect Nash Equilibrium}\label{section:theory}

In Part I, we characterized the SPNE(s) for the case that the transport costs are smaller than a threshold, specifically $t_N+t_{NoN}\leq\kappa_u \tilde{q}_p$. Now, in this section, we characterize the SPNE(s) for the complement of this case, i.e. $t_N+t_{NoN}>\kappa_u \tilde{q}_p$. Note that, in Part I, the analysis for Stages 2 to 4 was done regardless of the assumption on the transport costs. Thus, we can use the equilibrium strategies of these stages, and only need to characterize the equilibrium strategies of Stage 1, under the assumption that $t_N+t_{NoN}>\kappa_u \tilde{q}_p$.   


\subsection*{Stage 1: ISPs determine $p^{eq}_N$ and $p^{eq}_{NoN}$:} \label{section:stage1}

In this section, we characterize the NE strategies $p^{eq}_N$ and $p^{eq}_{NoN}$ using \eqref{l-equ:payoffISPsGeneral_new}-I for the case that $t_N+t_{NoN}>\kappa_u \tilde{q}_p$. First, in Theorem~\ref{theorem:NE_stage1_new_q<}, we characterize all possible NE strategies by which $z^{eq}=1$. In Theorem~\ref{theorem:bigt}, we prove that when one of the inertias is large, the only NE strategy by which $z^{eq}=1$ is the third candidate strategy of Theorem~\ref{theorem:NE_stage1_new_q<}. Then, we focus on characterizing NE strategies by which $z^{eq}=0$. In Theorem~\ref{theorem:neutralz=0_q<}, we characterize the only candidate NE strategy by which $z^{eq}=0$.

Similar to the previous part of the paper, and  without loss of generality, in the equilibrium, $p^{eq}_N\geq c$.\footnote{Note that if $n_N>0$ then $p_N<c$ yields a negative payoff for the neutral ISP. Thus, no $p_N<c$ can be an equilibrium payoff. If $n_N=0$, the value of $p_N$ is of no importance. Therefore, without loss of generality we can consider $p_N\geq c$.} In addition, if $z=0$, $p^{eq}_{NoN}\geq c$. Similar to the previous part, if $0\leq x_n\leq 1$, i.e. $(q^{eq}_N,q^{eq}_{NoN})\in F^I$, from \eqref{l-equ:EUs_linear}-I,  the payoff of neutral and non-neutral ISPs are as follow: 
\footnotesize
\be \label{equ:UN_new}
\pi_N(p_N)=(p_N-c)\frac{t_{NoN}+\kappa_u (q_N-q_{NoN})+p_{NoN}-p_N}{t_N+t_{NoN}} 
\ee

\be\label{equ:UNoN_new}
\ba 
\pi_{NoN}(p_{NoN},\tilde{p})&=(p_{NoN}-c)\frac{t_N+\kappa_u (q_{NoN}-q_N)+p_N-p_{NoN}}{t_N+t_{NoN}}\\
&\qquad \qquad +zq_{NoN}\tilde{p}
\ea
\ee
\normalsize

Given the strategies of the CP and EUs described in the previous part of the paper, in the following theorem we characterize the NE strategies by which $z^{eq}=1$ when $\tilde{q}_{p}<\frac{t_N+t_{NoN}}{\kappa_u}$:

\begin{theorem}\label{theorem:NE_stage1_new_q<}
If $\tilde{q}_{p}<\frac{t_N+t_{NoN}}{\kappa_u}$, then the only possible  NE strategies by which $(q^{eq}_N,q^{eq}_{NoN})\in F_1$, i.e. $z^{eq}=1$, are:
\begin{enumerate}
\item If $\Delta p\leq \kappa_u \tilde{q}_p-t_{NoN}$, then $p^{eq}_{NoN}=c+\kappa_u \tilde{q}_{p}-t_{NoN}$ and $p^{eq}_N=c$.
\item  If  (i)  $\kappa_u \tilde{q}_{p}-t_{NoN}<\Delta p^{eq}<\kappa_u (2\tilde{q}_{p}-\tilde{q}_f)-t_{NoN}$ or $t_N+\kappa_u(\tilde{q}_{p}-\tilde{q}_f)<\Delta p^{eq}<t_N+\kappa_u \tilde{q}_{p}$, then $p^{eq}_{NoN}=c+\frac{t_{NoN}+2t_N+\tilde{q}_{p}(\kappa_u -2\kappa_{ad})}{3}$ and $p^{eq}_{N}=c+\frac{2t_{NoN}+t_N-\tilde{q}_{p}(\kappa_u +\kappa_{ad})}{3}$, and $\pi^{eq}_{NoN}=\pi_{NoN}(\tilde{p}^{eq}_{NoN},\tilde{p}_{t,2})$. The necessary conditions: (ii) $\tilde{q}_{p}\leq \frac{2t_{NoN}+t_N}{\kappa_u+\kappa_{ad}}$, and  (iii)  $\pi^{eq}_{NoN}>\pi_{NoN,z=0}(\tilde{p}^{eq}_{NoN},\tilde{p})$.
\item If (i)  $\kappa_u (2\tilde{q}_{p}-\tilde{q}_f)-t_{NoN}< \Delta p^{eq}<t_N+\kappa_u (\tilde{q}_{p}-\tilde{q}_f)$, then $p^{eq}_{NoN}=c+\frac{t_{NoN}+2t_N+  (\tilde{q}_{p}-\tilde{q}_f)(\kappa_u -2\kappa_{ad})}{3}$ and $p^{eq}_{N}=c+\frac{2t_{NoN}+t_N-(\tilde{q}_{p}-\tilde{q}_f)(\kappa_u +\kappa_{ad})}{3}$, and $\pi^{eq}_{NoN}=\pi_{NoN}(\tilde{p}^{eq}_{NoN},\tilde{p}_{t,3})$. The necessary conditions:  (ii) $\tilde{q}_{p}-\tilde{q}_f\leq \frac{2t_{NoN}+t_N}{\kappa_u+\kappa_{ad}}$,  and (iii)  $\pi^{eq}_{NoN}>\pi_{NoN,z=0}(\tilde{p}^{eq}_{NoN},\tilde{p})$.
\item  $p^{eq}_{NoN}=c$ and $p^{eq}_N=c-\kappa_u(2\tilde{q}_{p}-\tilde{q}_f)+t_{NoN}$, and $\pi^{eq}_{NoN}=\pi_{NoN}(\tilde{p}^{eq}_{NoN},\tilde{p}_{t,3})$. The necessary conditions: (i) $2\tilde{q}_{p}-\tilde{q}_f\leq \frac{t_{NoN}}{\kappa_{u}}$ and (ii)  $\pi^{eq}_{NoN}>\pi_{NoN,z=0}(\tilde{p}^{eq}_{NoN},\tilde{p})$.
\end{enumerate}
\end{theorem}
\begin{remark}
Note that these candidate strategies are NE if and only if the conditions listed in the theorem hold and no unilateral deviation is profitable for each of the ISPs.
\end{remark}

We will show in the next corollary that the first two sets of strategies are associated with the case that the CP offers with premium quality on ISP NoN and with zero quality on ISP N. With the first strategies, ISP N would be driven out of the market, while with the second strategy, ISP N would be active. The third and fourth set of strategies are associated with the case that both ISPs are active and the CP offers her content with premium quality on ISP NoN and with free quality on ISP N.

\begin{corollary} \label{corollary:outcome_q<} If Strategy 1 of Theorem~\ref{theorem:NE_stage1_new_q<} is an NE, it yields $\tilde{p}^{eq}=\tilde{p}_{t,1}=\kappa_{ad}(1-\frac{\tilde{q}_f}{\tilde{q}_p})$, $(q^{eq}_N,q^{eq}_{NoN})=(0,\tilde{q}_p)\in F^L_1$, $n^{eq}_N=0$, and $n^{eq}_{NoN}=1$. If  Strategy 2 of Theorem~\ref{theorem:NE_stage1_new_q<} is an NE, it yields $\tilde{p}^{eq}=\tilde{p}_{t,2}=\kappa_{ad}(n^{eq}_{NoN}-\frac{\tilde{q}_f}{\tilde{q}_p})$, $(q^{eq}_N,q^{eq}_{NoN})=(0,\tilde{q}_p)\in F^I_1$, $n^{eq}_N=\frac{t_N+2t_{NoN}-\tilde{q}_p(\kappa_u+\kappa_{ad})}{3(t_N+t_{NoN})}$, and $n^{eq}_{NoN}=\frac{2t_N+t_{NoN}+\tilde{q}_p(\kappa_u+\kappa_{ad})}{3(t_N+t_{NoN})}$. If Strategy 3 of Theorem~\ref{theorem:NE_stage1_new_q<} is an NE, it yields $\tilde{p}^{eq}=\tilde{p}_{t,3}=\kappa_{ad}n^{eq}_{NoN}(1-\frac{\tilde{q}_f}{\tilde{q}_p})$, $(q^{eq}_N,q^{eq}_{NoN})=(\tilde{q}_f,\tilde{q}_p)\in F^I_1$, $n^{eq}_N=\frac{t_N+2t_{NoN}-(\tilde{q}_p-\tilde{q}_f)(\kappa_u+\kappa_{ad})}{3(t_N+t_{NoN})}$, and $n^{eq}_{NoN}=\frac{2t_N+t_{NoN}+(\tilde{q}_p-\tilde{q}_f)(\kappa_u+\kappa_{ad})}{3(t_N+t_{NoN})}$. If Strategy 4 of Theorem~\ref{theorem:NE_stage1_new_q<} is an NE, it yields $\tilde{p}^{eq}=\tilde{p}_{t,3}=\kappa_{ad}n^{eq}_{NoN}(1-\frac{\tilde{q}_f}{\tilde{q}_p})$, $(q^{eq}_N,q^{eq}_{NoN})=(\tilde{q}_f,\tilde{q}_p)\in F^I_1$,  $n^{eq}_N=\frac{\kappa_u \tilde{q}_p}{t_N+t_{NoN}}$, and $n^{eq}_{NoN}=\frac{t_N+t_{NoN}-\kappa_u \tilde{q}_p}{t_N+t_{NoN}}$.
\end{corollary}

Next, we prove that when either of the transport costs is large enough, then the NE strategy by which $z^{eq}=1$ is the third candidate strategy of the previous theorem. 

\begin{theorem}\label{theorem:bigt}
	When either $t_N$ or $t_{NoN}$ is large enough, for the case that $\tilde{q}_p<\frac{t_N+t_{NoN}}{\kappa_u}$, the only NE strategy by which $z^{eq}=1$ is $p^{eq}_{NoN}=c+\frac{t_{NoN}+2t_N+  (\tilde{q}_{p}-\tilde{q}_f)(\kappa_u -2\kappa_{ad})}{3}$ and $p^{eq}_{N}=c+\frac{2t_{NoN}+t_N-(\tilde{q}_{p}-\tilde{q}_f)(\kappa_u +\kappa_{ad})}{3}$. 
\end{theorem}

Proof of Theorem \ref{theorem:bigt} is presented in Appendix~\ref{appendix:theorem:bigt}. 
\begin{remark}
	Note that when at least one of $t_N$ and $t_{NoN}$ is large, then the effect of $\tilde{q}_p$ can be ignored. Thus, this scenario can be considered to be similar to the case that both ISPs are neutral, i.e. the benchmark case.  Later, in Theorem~\ref{lemma:NEz=0}, we prove that a unique SPNE exists in this case, and it is similar to the NE strategies characterized in Theorem~\ref{theorem:bigt} with $\tilde{q}_p=\tilde{q}_f$. 
\end{remark}

Now, we characterize the equilibrium strategy by which $z^{eq}=0$. In Theorem~\ref{theorem:neutralz=0_q<}, we characterize the only candidate NE strategy by which $z^{eq}=0$, when inertias are not small.

\begin{theorem}\label{theorem:neutralz=0_q<}
	If $\tilde{q}_p<\frac{t_N+t_{NoN}}{\kappa_u}$, the only possible NE strategy by which  $(q^{eq}_N,q^{eq}_{NoN})\in F_0$, i.e. $z^{eq}=0$ is 	$p^{eq}_N=c+\frac{1}{3}(2t_{NoN}+t_N)$ and $p^{eq}_{NoN}=c+\frac{1}{3}(2t_N+t_{NoN})$.  The necessary condition for this strategy to be a candidate NE strategy is $\pi_{NoN}(p^{eq}_{NoN},z=0)\geq \pi_{NoN}(p^{eq}_{NoN},\tilde{p}_t) $.  
\end{theorem}

\begin{remark}
	This strategy is an NE strategy if and only if there is no  unilateral profitable deviation for any of the ISPs. 
\end{remark}

\begin{corollary} \label{corollary:outcome_z=0}
If the strategy  of Theorem~\ref{theorem:neutralz=0_q<} is an NE, it yields  $(q^{eq}_N,q^{eq}_{NoN})=(\tilde{q}_f,\tilde{q}_f)\in F^L_0$, $n^{eq}_N=\frac{2t_{NoN}+t_N}{3(t_{NoN}+t_N)}$, and $n^{eq}_{NoN}=\frac{2t_N+t_{NoN}}{3(t_N+t_{NoN})}$. Since $z^{eq}=0$, $\tilde{p}^{eq}$ is of no importance.
\end{corollary}


 \section{Benchmark Case: A Neutral Regime}\label{section:proofsBenchamrk}
In this section, we consider a benchmark case in which both ISPs are forced to be neutral. Our goal is to find the SPNE when both ISPs are neutral. We compare the results of the benchmark case with the results we found in the previous section. Note that we do not restrict the analysis of this section to any particular range of transport costs, and the analysis is done for a general case.

 First, we find the equilibrium strategies for this benchmark case. The main result of this section is Theorem~\ref{lemma:NEz=0}. Then, in Corollary~\ref{corollary:outcome_neutral}, we characterize the equilibrium outcome of the game given the equilibrium strategies. 
 
 In order to characterize the equilibrium in this case, we can consider a simple change to our previous model and use some of the previous results. We assume that in this case, the CP chooses $z^{eq}=0$, regardless of the strategy of ISPs. Thus, $(q^{eq}_N,q^{eq}_{NoN})\in F_0$, and as a result both ISPs are neutral. 
 
 Note that in this case, the equilibrium strategy of Stage 4 is similar as before, and the equilibrium strategy of Stage 3 is characterized in Theorem~\ref{l-lemma:CP_z=0_new}-I. Recall that in Theorem~\ref{l-lemma:CP_z=0_new}-I, we characterize the equilibrium strategies within $F_0$ without  considering the strategies in $F_1$. In addition, note that the strategy of Stage 2 is of no importance  since with $z^{eq}=0$, the effect of $\tilde{p}$ would be eliminated in all analyses. Thus, we only need to find the new equilibrium strategies in Stage 1 of the game:

 \begin{theorem}\label{lemma:NEz=0}
 	The unique NE strategies chosen by the ISP are $p^{eq}_N=c+\frac{1}{3}(2t_{NoN}+t_N)$ and $p^{eq}_{NoN}=c+\frac{1}{3}(2t_N+t_{NoN})$. 
 \end{theorem}

 \begin{remark}\label{remark:pNpNoNmulti}
 	Note that if $t_N$ (respectively, $t_{NoN}$) increases, $p_{NoN}$ (respectively, $p_N$) increases with a rate $\frac{2}{3}$rd the rate of the growth of this transport cost. In addition, counter-intuitively, $p_N$ (respectively, $p_{NoN}$) also increases with a rate $\frac{1}{3}$rd of the rate of growth of $t_N$ (respectively, $t_{NoN}$). This counter-intuitive result (Internet access fee of an ISP being an increasing function of the transport cost of this ISP) is because of competition between ISPs: For example, with the  increase of $t_{NoN}$, the neutral ISP can set a higher price. This allows her competitor, i.e. ISP NoN, to increases her price. 
 \end{remark}

Now, using the equilibrium strategies characterized previously, we characterize the equilibrium outcome of the game in a neutral setting in the following corollary:
 
 \begin{corollary}\label{corollary:outcome_neutral}
 	If both ISPs are neutral, then in the equilibrium, ISPs chooses  $p^{eq}_N=c+\frac{1}{3}(2t_{NoN}+t_N)$ and $p^{eq}_{NoN}=c+\frac{1}{3}(2t_N+t_{NoN})$ 	as  the Internet access fees. The CP chooses the vector of qualities $(q^{eq}_N,q^{eq}_N)=(\tilde{q}_f,\tilde{q}_f)$. The fraction of EUs with each ISP is $n^{eq}_{N}=\frac{2t_{NoN}+t_N}{3(t_{NoN}+t_N)}$ and $n^{eq}_{NoN}=\frac{2t_N+t_{NoN}}{3(t_{NoN}+t_N)}$.
 \end{corollary}
 Results follow from Theorem~\ref{l-lemma:CP_z=0_new}-I (note that  $-t_{NoN}<\Delta p^{eq}<t_N$), and  \eqref{l-equ:EUs_linear}-I.

\section{The Outcome of the Game and Discussions}\label{section:discussion}

First, in Section~\ref{section:summaryof resutls}, we summarize, discuss, and interpret the possible outcomes of the model characterized in Part I and II. Then, in Section~\ref{section:summary_benchmark}, we summarize  and discuss the results for a benchmark case in which both ISPs are neutral. 

\subsection{Possible Outcomes of the Market}  \label{section:summaryof resutls}
Recall that in Section~\ref{section:theory}, we have characterized all the possible SPNE strategies and their outcomes in Theorems  \ref{l-theorem:NE_stage1_new_q>}-I, \ref{theorem:NE_stage1_new_q<}, and  \ref{theorem:neutralz=0_q<}, and Corollaries~\ref{l-corollary:outcome_q>}-I, \ref{corollary:outcome_q<}, and \ref{corollary:outcome_z=0}. In this section, we summarize, discuss, and interpret these strategies.

\textbf{Candidate Strategy (a):}   $p^{eq}_{NoN}=c+\kappa_u \tilde{q}_{p}-t_{NoN}$, $p^{eq}_N=c$, $z^{eq}=1$, i.e. the CP pays for the premium quality, $\tilde{p}^{eq}=\tilde{p}_{t,1}=\kappa_{ad}(1-\frac{\tilde{q}_f}{\tilde{q}_p})$, $(q^{eq}_N,q^{eq}_{NoN})=(0,\tilde{q}_p)\in F^L_1$, $n^{eq}_N=0$, and $n^{eq}_{NoN}=1$ ({Theorem~\ref{l-theorem:NE_stage1_new_q>}-I-1} and {Theorem~\ref{theorem:NE_stage1_new_q<}-1}). Also, recall that in Theorems~\ref{l-theorem:neutralnotexists_q>}-I and \ref{l-theorem:NE_stage1_new_q>}-I, we prove that when $t_N$ and $t_{NoN}$ are sufficiently small, i.e. EUs are not locked-in with ISPs, then the unique SPNE of  the game is (a). 

This strategy represents the outcome of the market in which the CP offers the content with premium quality and pays the side-payment to the non-neutral ISP. Note that EUs can receive a better quality of content on the non-neutral ISP, and that yields a better advertisement revenue for the CP. Thus, in the equilibrium, the CP offers her content only on the non-neutral ISP to increase the number of EUs connecting to the Internet via the non-neutral ISP. By doing so, given the conditions of this candidate strategy, all  EUs would join the non-neutral ISP and the neutral ISP would be driven out of the market. 

In addition, note that the Internet access fee chosen by ISP NoN ($p^{eq}_{NoN}$) increases with (i) increasing the sensitivity of end-users to the quality ($\kappa_u$), (ii) increasing the value of the premium quality ($\tilde{q}_p$), and (iii) decreasing the transport cost of ISP NoN ($t_{NoN}$)\footnote{Note that $t_{NoN}$ has an inverse relationship with the market power of ISP NoN if $t_N$ is fixed.}. 

Moreover, note that the side-payment charged for the premium quality ($\tilde{p}^{eq}\tilde{q}_p$) is positive, and is dependent on (i) the sensitivity of the payoff of the CP to the quality of the advertisement, i.e. $\kappa_{ad}$, and (ii) the difference between the premium and free quality, i.e. $\tilde{q}_p-\tilde{q}_f$. The latter implies that ISP NoN chooses the side-payment in proportion to the additional value created for the CP.

\textbf{Candidate Strategy (b):} $p^{eq}_{NoN}=c+\frac{t_{NoN}+2t_N+\tilde{q}_{p}(\kappa_u -2\kappa_{ad})}{3}$, $p^{eq}_{N}=c+\frac{2t_{NoN}+t_N-\tilde{q}_{p}(\kappa_u +\kappa_{ad})}{3}$, $z^{eq}=1$, $\tilde{p}^{eq}=\tilde{p}_{t,2}=\kappa_{ad}(n^{eq}_{NoN}-\frac{\tilde{q}_f}{\tilde{q}_p})$, $(q^{eq}_N,q^{eq}_{NoN})=(0,\tilde{q}_p)\in F^I_1$, $n^{eq}_N=\frac{t_N+2t_{NoN}-\tilde{q}_p(\kappa_u+\kappa_{ad})}{3(t_N+t_{NoN})}$, and $n^{eq}_{NoN}=\frac{2t_N+t_{NoN}+\tilde{q}_p(\kappa_u+\kappa_{ad})}{3(t_N+t_{NoN})}$ ({Theorem~\ref{l-theorem:NE_stage1_new_q>}-I-2} and {Theorem~\ref{theorem:NE_stage1_new_q<}-2}). 

Candidate strategy (b) represents the outcome of the market in which both ISPs are active. However, similar to (a), with this outcome, the CP does not offer her content over the neutral ISP, and  offers her content only over the non-neutral ISP  with premium quality. Thus, although the CP stops offering her content on the neutral ISP, she cannot move all EUs to ISP NoN. The loss in the number of EUs would be compensated by receiving higher advertisement revenue (due to the premium quality) and paying  a lower side payment to ISP NoN (will be explained the paragraph discussing the side payment).

It is noteworthy to observe that if $t_N$ (respectively, $t_{NoN}$) increases, $p^{eq}_{NoN}$ (respectively, $p^{eq}_N$) increases with a rate $\frac{2}{3}$rd the rate of the growth of this transport cost. This is intuitive. The higher $t_N$ (while $t_{NoN}$ fixed), the higher would be the market power of ISP NoN, and subsequently the higher would be $p^{eq}_{NoN}$. In addition, counter-intuitively, $p_N$ (respectively, $p_{NoN}$) also increases with a rate $\frac{1}{3}$rd of the rate of growth of $t_N$ (respectively, $t_{NoN}$). This counter-intuitive result (Internet access fee of an ISP being an increasing function of the transport cost of this ISP) is because of competition between ISPs\footnote{For example, with the increase of $t_{NoN}$, EUs have more incentive to join the neutral ISP and less incentive to switch to the non-neutral ISP. Thus the neutral ISP can set a higher price for EUs. This allows her competitor, i.e. ISP NoN, to increases her price, but with a rate lower than the rate by which the price of ISP N increases.}. 

In addition, note that $p^{eq}_N$ is a decreasing function of $\tilde{q}_p$, $\kappa_u$, and $\kappa_{ad}$: The higher the premium quality or the sensitivity of EUs and the CP to the quality, the lower would be the Internet access fee of ISP N. On the other hand, the relationship between these parameters and $p^{eq}_{NoN}$ is more complicated. The Internet access fee of ISP NoN is increasing with respect to the sensitivity of EUs to the quality, and is decreasing with respect to the sensitivity of the CP to the quality. Thus, the more the CP is sensitive to the quality, the more the ISP NoN provides subsidies for EUs (cheaper Internet access fees), and compensates the payoff through charging the CP. In addition, note that $p^{eq}_{NoN}$ is decreasing or increasing with respect to the amount of premium quality ($\tilde{q}_p$) depending on the sensitivity of EUs and the CP to the quality: If the sensitivity of EUs to the quality dominates the sensitivity of the CP ($\kappa_u>2\kappa_{ad}$), then $p^{eq}_{NoN}$ is increasing with respect to $\tilde{q}_p$. If not, then ISP NoN can generate more revenue from the CP, and thus provide a cheaper Internet access fee for EUs. The higher this sensitivity, the higher would be the side payment from the CP (as we can see from the expression of $\tilde{p}^{eq}$), and the higher would be the discount on the Internet access fees for EUs. 

Moreover, note that the side-payment charged for the premium quality ($\tilde{p}^{eq}\tilde{q}_p$) is increasing with respect to (i) $\kappa_{ad}$ (the sensitivity of the CP to the quality), (ii) the premium quality ($\tilde{q}_p$), (iii) number of EUs with the non-neutral ISP ($n^{eq}_{NoN}$), and it is decreasing with respect to the free quality ($\tilde{q}_f$). Note that since in this case $n_{NoN}<1$, the side payment would be lower than the side payment in candidate strategy (a). This side-payment can be positive or negative.  However, as we explain later, the numerical results reveal that the side-payment is positive whenever this candidate strategy is an NE. 

In addition, note that $n_{NoN}$ is increasing with respect to the premium quality, i.e. $\tilde{q}_p$, and the sensitivity of the CP and EUs to the quality, i.e. $\kappa_u$ and $\kappa_{ad}$. The relationship between $n_{NoN}$ and the transport costs, i.e. $t_N$ and $t_{NoN}$ is more complex and is discussed in Section~\ref{section:simulation_findNE}.\footnote{Recall that $n_N=1-n_{NoN}$. Thus, the relationship between $n_N$ and the parameters are the inverse of that of $n_{NoN}$.}

\textbf{Candidate Strategy (c):}  $p^{eq}_{NoN}=c+\frac{t_{NoN}+2t_N+  (\tilde{q}_{p}-\tilde{q}_f)(\kappa_u -2\kappa_{ad})}{3}$, $p^{eq}_{N}=c+\frac{2t_{NoN}+t_N-(\tilde{q}_{p}-\tilde{q}_f)(\kappa_u +\kappa_{ad})}{3}$, $z^{eq}=1$, $\tilde{p}^{eq}=\tilde{p}_{t,3}=\kappa_{ad}n^{eq}_{NoN}(1-\frac{\tilde{q}_f}{\tilde{q}_p})$, $(q^{eq}_N,q^{eq}_{NoN})=(\tilde{q}_f,\tilde{q}_p)\in F^I_1$, $n^{eq}_N=\frac{t_N+2t_{NoN}-(\tilde{q}_p-\tilde{q}_f)(\kappa_u+\kappa_{ad})}{3(t_N+t_{NoN})}$, and $n^{eq}_{NoN}=\frac{2t_N+t_{NoN}+(\tilde{q}_p-\tilde{q}_f)(\kappa_u+\kappa_{ad})}{3(t_N+t_{NoN})}$  ({Theorem~\ref{theorem:NE_stage1_new_q<}-3}).  In Theorem~\ref{theorem:bigt}, we prove that when either of $t_N$ or $t_{NoN}$ is large, then the only candidate strategy by which $z^{eq}=1$ is (c).

 Candidate strategy (c) represents the outcome of the market in which both ISPs are active, and the CP  offers her content over both ISPs, with free quality on the  neutral ISP and with premium quality on the non-neutral one.  
 
 The dependencies of the access fees ($p^{eq}_N$ and $p^{eq}_{NoN}$) to $t_N$, $t_{NoN}$, $\kappa_u$, and $\kappa_{ad}$ are the same as what described for candidate strategy (b).  In addition, note that $p^{eq}_N$ is decreasing with the difference between the premium and free qualities, i.e. $\tilde{q}_p-\tilde{q}_f$, and $p^{eq}_{NoN}$ is decreasing or increasing with respect to the difference in the qualities depending on the sensitivity of EUs and the CP to the quality (similar to the description  for the candidate strategy (b)). 
 
 Moreover, note that the side-payment charged for the premium quality ($\tilde{p}^{eq}\tilde{q}_p$) is increasing with respect to (i) $\kappa_{ad}$ (the sensitivity of the CP to the quality), (ii) the difference between the premium  and free qualities ($\tilde{q}_p-\tilde{q}_f$), (iii) number of EUs with the non-neutral ISP ($n^{eq}_{NoN}$). This side-payment is always positive. The dependencies of $n_{NoN}$ to the parameters is similar to what described for candidate strategy (b), with the difference that $n_{NoN}$ depends on the difference in the qualities, i.e. $\tilde{q}_p-\tilde{q}_f$, instead of only $\tilde{q}_p$.

 Note that when either of $t_N$ or $t_{NoN}$ is large, then (c) is the only candidate strategy by which $z^{eq}=1$. First, recall that the payoff that an ISP receives depends on both the number of EUs and the Internet connection fee of that ISP. In addition, we discussed that when either of $t_N$ or $t_{NoN}$ is large, then both of  the Internet connection fees would be large in candidate strategies (b) and (c). It turns out that when $t_N$ or $t_{NoN}$ is large, ISPs prefer candidate strategies (b) and (c) to the strategies by which they discount the price heavily to attract EUs (Candidate strategies (a) and (d)). 
 
 
 Moreover, when both ISPs are active, large $t_{NoN}$ or $t_N$ decreases the effect of quality of the content on the decision of EUs (both through the high transport costs and increase in the Internet access fees).  Thus, the CP cannot control the number of EUs with each ISP by strategically controlling her quality. Therefore the CP simply chooses to provide with maximum possible quality on both ISPs instead of choosing strategic qualities to control EUs with each ISP. Note that the only candidate strategy that yields this outcome is (c).
 
\textbf{Candidate Strategy (d):} $p^{eq}_{NoN}=c$, $p^{eq}_N=c-\kappa_u(2\tilde{q}_{p}-\tilde{q}_f)+t_{NoN}$, $z^{eq}=1$, $\tilde{p}^{eq}=\tilde{p}_{t,3}=\kappa_{ad}n^{eq}_{NoN}(1-\frac{\tilde{q}_f}{\tilde{q}_p})$, $(q^{eq}_N,q^{eq}_{NoN})=(\tilde{q}_f,\tilde{q}_p)\in F^I_1$,  $n^{eq}_N=\frac{\kappa_u \tilde{q}_p}{t_N+t_{NoN}}$, and $n^{eq}_{NoN}=\frac{t_N+t_{NoN}-\kappa_u \tilde{q}_p}{t_N+t_{NoN}}$  ({Theorem~\ref{theorem:NE_stage1_new_q<}-4}).

Candidate strategy (d) and its outcome represent the scenario in which the non-neutral ISP is forced to provide  a low Internet access fee for EUs. This strategy can only be valid when $t_{NoN}$ is large (so that $p^{eq}_N\geq c$). In other words, the only scenario that this strategy is possible is when EUs are reluctant joining the non-neutral ISP. Thus, this ISP should provide large discounts for EUs. Note that in this case, both ISPs are active, and the CP  offers her content over both ISPs, with free quality on the  neutral ISP and with premium quality on the non-neutral one. 

Note that in this case, $p^{eq}_N$ is decreasing with respect to $\kappa_u$ and $\tilde{q}_p$, and increasing with respect to $\tilde{q}_f$ and $t_{NoN}$. In addition, the side payment is similar to the candidate strategy (c).  

In this candidate strategy, $p^{eq}_{NoN}$ is fixed, while $p^{eq}_N$ is decreasing with respect to $\tilde{q}_p$ and $\kappa_u$. In addition, the rate of decrease of $p^{eq}_N$ is twice of the rate of increase of utility of EUs from $\kappa_u$ and $\tilde{q}_p$ when connecting to ISP NoN. Thus, The rate of increase in the utility of EUs for ISP N is higher than that of ISP NoN, and as result confirms, $n^{eq}_N$ would be increasing with respect to the premium quality and the sensitivity of EUs to the quality. In addition,  $p^{eq}_N$ is increasing with $t_{NoN}$. Thus, as result confirms, $n^{eq}_N$ would be decreasing with respect to the transport cost of ISP NoN.\footnote{Note that the utility of EUs connecting to ISP NoN is also decreasing with $t_{NoN}$ \eqref{l-equation:CP_2}-I. However, the rate of decrease in the utility of EUs connecting to ISP NoN ($t_{NoN}$ is multiplies to a coefficient smaller than one) is lower than the rate of increase of the price of the neutral ISP (multiplied by one). Thus, overall, the number of EUs with the neutral (respectively, non-neutral) ISP is decreasing (respectively, increasing).} Finally, note that the Internet access fees are independent of $t_N$, but the utility of EUs connecting to neutral ISP is decreasing with $t_N$ \eqref{l-equation:CP_2}-I. Thus, as result confirms, the number of EUs with the neutral ISP, i.e. $n^{eq}_{N}$, is decreasing with respect to both $t_N$.

\textbf{Candidate Strategy (e):}  $p^{eq}_{NoN}=c+\frac{1}{3}(2t_N+t_{NoN})$, $p^{eq}_N=c+\frac{1}{3}(2t_{NoN}+t_N)$, $z^{eq}=0$, $(q^{eq}_N,q^{eq}_{NoN})=(\tilde{q}_f,\tilde{q}_f)\in F^L_0$, $n^{eq}_N=\frac{2t_{NoN}+t_N}{3(t_{NoN}+t_N)}$, $n^{eq}_{NoN}=\frac{2t_N+t_{NoN}}{3(t_N+t_{NoN})}$, and since $z^{eq}=0$, $\tilde{p}^{eq}$ is of no importance ({Theorem~\ref{theorem:neutralz=0_q<}}).  

This case characterizes the only possible NE strategy by which $z^{eq}=0$. This strategy and its outcome is similar to the benchmark case presented in 
Section~\ref{section:summary_benchmark}.  (c) would be reduced to (e), if $\tilde{q}_p=\tilde{q}_f$.

\textbf{Remark: } Note that candidate strategies in different theorems (defined for different regions of $t_N$ and $t_{NoN}$) can be similar, e.g. {Theorem~\ref{l-theorem:NE_stage1_new_q>}-I-1} and {Theorem~\ref{theorem:NE_stage1_new_q<}-1}. In addition, there is no outcome in which the CP offers her content  only on the neutral ISP. From the expression of payoff of the CP \eqref{l-equ:payoffCP_new}-I, the CP can get at most $\kappa_{ad}\tilde{q}_f$ by offering only on the neutral ISP. On the other hand, the CP can guarantee a payoff of this amount by offering on both ISPs and $z=0$. Assumption~\ref{l-assumption:tie_diversify}-I, i.e. the CP prefers to offer on both ISP whenever she is indifferent, yields that the CP never choose the strategy in which she offers only on the neutral ISP.  

\textbf{Interplay Between Sensitivities to Quality and the Outcome:} Intuitively, we expect that high sensitivity of EUs and the CP to the quality, i.e.  large $\kappa_u$ and $\kappa_{ad}$, respectively, yields high payoff for the non-neutral ISP, since this ISP can provide a premium quality and charge the EUs accordingly to increase her payoff. Thus, the payoff can be collected from EUs or the CP, or both. Results reveal that in all candidate strategies ISP NoN charges the CP in proportion to her sensitivity to the quality of the content. In addition, in candidate strategies (a) to (c), the payoff collected from EUs through the Internet connection fees is always increasing with respect to the sensitivity of the EUs to the quality. In candidate strategies (b) and (c), the Internet connection fees are decreasing with respect to the sensitivity of the CP to the quality. Thus, in these candidate strategies EUs receive a discount in proportion to the sensitivity of the CP to the quality. In candidate strategy (d), the Internet connection fee of ISP NoN does not depends on the qualities, but it is as low as the marginal cost. 

\textbf{Existence of NE:} An SPNE may not always exist. For example, for the parameters $\tilde{q}_f=1$, $\tilde{q}_p=1.5$, $c=1$, $\kappa_u=1$, $\kappa_{ad}=0.5$, $t_N=3$, and $t_{NoN}=2$, none of the candidate strategies listed above would be an SPNE. The reason is that there exists a profitable deviation for at least one of the ISPs for those candidate strategies that their conditions are satisfied given these parameters. Later in Section~\ref{section:simulation_findNE}, we identify the regions with no SPNE.

\subsection{Benchmark: A Neutral Scenario}\label{section:summary_benchmark}
In the benchmark case, i.e. when both ISPs are neutral, we proved that there exists a unique SPNE, and the unique equilibrium outcome of the game is\footnote{The subscript B in the results refers to ``Benchmark".}:
\begin{itemize}
	\item Stage 1 - Internet access Fees chosen by ISPs: $p^{eq}_{N,B}=c+\frac{1}{3}(2t_{NoN}+t_N)$ and $p^{eq}_{NoN,B}=c+\frac{1}{3}(2t_N+t_{NoN})$.
	\item Stage 2 -  Side Payment chosen by ISP NoN is of no importance. 
	\item Stage 3 -  Qualities chosen by the CP: $q^{eq}_{NoN,B}=\tilde{q}_f$ and ${q}^{eq}_{N,B}=\tilde{q}_f$.
	\item Stage 4 - Fractions of EUs with ISPs: $n^{eq}_{N,B}=\frac{2t_{NoN}+t_N}{3(t_{NoN}+t_N)}$ and $n^{eq}_{NoN,B}=\frac{2t_N+t_{NoN}}{3(t_{NoN}+t_N)}$ .
\end{itemize}

Note that this case is similar to candidate strategy (e), i.e. the only candidate strategy of our model by which $z^{eq}=0$. In this case, both ISPs are active and the CP offers the free quality on both ISPs. Note that in this case, the asymmetries of the model only arise from the asymmetry in  $t_N$ and $t_{NoN}$. Thus, EUs are divided between ISPs depending on $t_N$ and $t_{NoN}$, and the Internet access fees ($p_N$ and $p_{NoN}$) are a function of transport costs ($t_N$ and $t_{NoN}$). Also, similar to the candidate strategy (b) of the previous section, if $t_N$ (respectively, $t_{NoN}$) increases, $p_{NoN}$ (respectively, $p_N$) increases with a rate $\frac{2}{3}$rd the rate of the growth of this transport cost. In addition, counter-intuitively, $p_N$ (respectively, $p_{NoN}$) also increases with a rate $\frac{1}{3}$rd of the rate of growth of $t_N$ (respectively, $t_{NoN}$).

Note that in this case, Internet access fees are independent of the quality provided for EUs, i.e. $\tilde{q}_f$. Recall that in contrast, in a non-neutral regime, the Internet access fee quoted by ISP NoN is dependent on the quality she provides ($\tilde{q}_p$).  The reason for that is the product differentiation in the latter. The non-neutral ISP can charge for the quality she provides for EUs through differentiating her product from the neutral ISP. While in a neutral regime, no ISP can charge for the quality they provide because of competition. It is noteworthy that if $t_{NoN}\& t_N\rightarrow 0$,  $p^{eq}_{NoN,B}\&p^{eq}_{N,B}\rightarrow c$. In other words,  in the absence of inertias, since there is no differentiation between the quality offered by the ISPs in the neutral regime, price competition drives the access fees to the marginal cost. This implies that by removing the inertias ($t_N$ and $t_{NoN}$), the model would be similar to a Bertrand competition \cite{MWG}. Thus, considering the inertias brings a realistic dimension to the model.

The relationship between $n^{eq}_N$ and $n^{eq}_{NoN}$ and the transport costs are similar to that of candidate strategies (b) and (c) of the previous section, and is investigated with numerical examples in Section~\ref{section:simulation_findNE}.

\section{Numerical Results}\label{section:numericalresults}

First, in Section~\ref{section:simulation_findNE}, we find the NE strategies for the entire range of transport costs, for various parameters, using numerical analysis. In Section~\ref{section:numerical_more_results},  we complement the discussions in Section~\ref{section:summaryof resutls},  by providing more intuitions about $n^{eq}_{NoN}$, $\tilde{p}^{eq}$, and the payoff of ISPs, based on the numerical results.	 We assess the benefits of non-neutrality by comparing the results of the model  with the benchmark case in Section~\ref{section:compare_results}.  In Section~\ref{section:regulation}, we provide regulatory comments based on the results.

\subsection{NE Strategies}\label{section:simulation_findNE}

Recall that if SPNE exists, it would of the form of strategies (a)-(e) (Section~\ref{section:summaryof resutls}).  Now, we check whether these strategies are indeed  SPNE. In order to do so, we check for any profitable deviation  for each of the ISPs. Note that to check for unilateral deviations, we consider different regions of $\Delta p$ (regions are characterized in Theorem \ref{l-theorem:p_tilde_new}-I). In each region, we can identify the potential profitable deviations (using first order condition for some regions, and the fact that payoff of ISPs are monotonic in other regions). Thus, the search for the best deviation for each ISP is equivalent to comparing the payoff of finite number of candidate deviations with the payoff of the candidate equilibrium. We also check additional conditions listed in Theorems \ref{l-theorem:NE_stage1_new_q>}-I, \ref{theorem:NE_stage1_new_q<}, and \ref{theorem:neutralz=0_q<}.

\begin{figure}[t]
	\centering
	\includegraphics[width=0.35\textwidth]{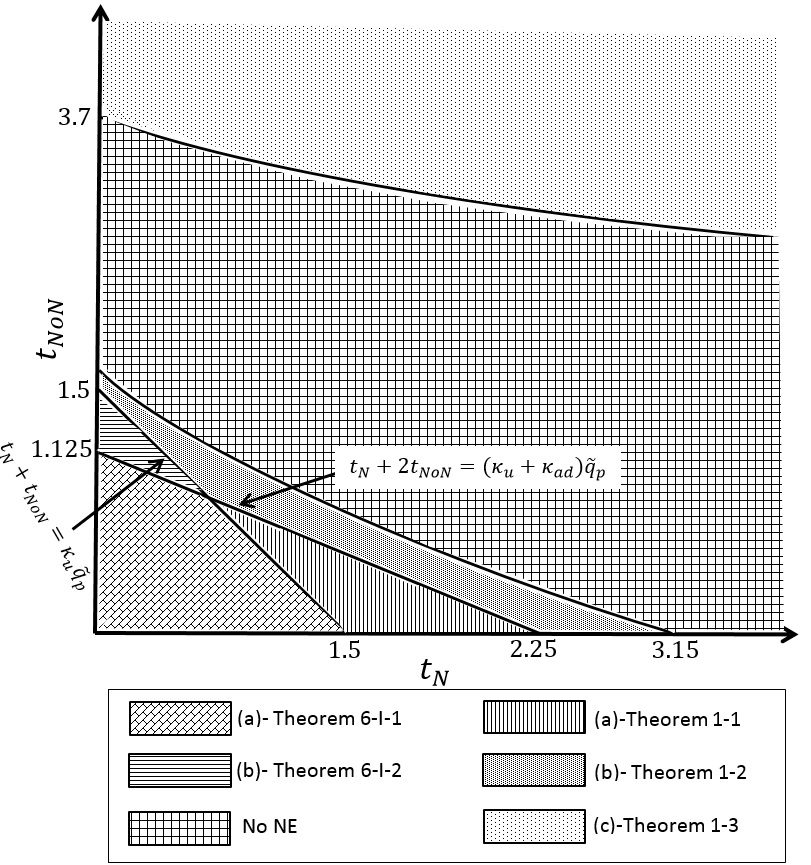}
	\caption{NE strategies of Stage 1 for various $t_N$ and $t_{NoN}$ when $\kappa_u=1$ and $\kappa_{ad}=0.5$}
	\label{figure:NE_general_ku>kad}
\end{figure}

\begin{figure}[t]
	\centering
	\includegraphics[width=0.35\textwidth]{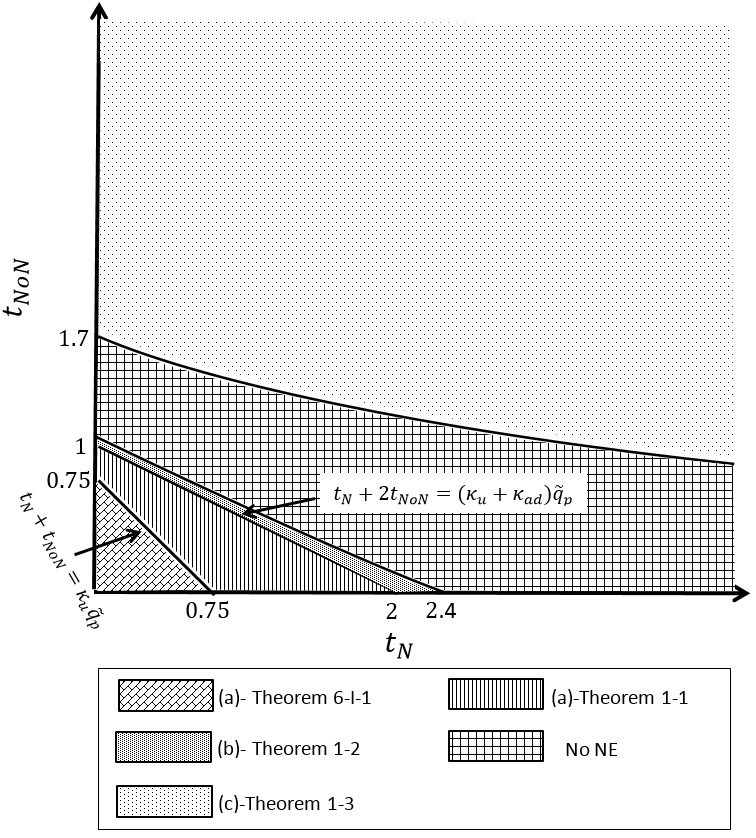}
	\caption{NE strategies of Stage 1 for various $t_N$ and $t_{NoN}$ when $\kappa_u=0.5$ and $\kappa_{ad}=1$}
	\label{figure:NE_general_ku<kad}
\end{figure}

%
%

We now present two illustrative examples. In Figure~\ref{figure:NE_general_ku>kad},  we identify the NE strategies of stage 1 for different regions of $t_N$ and $t_{NoN}$ when $\kappa_u=1$ and $\kappa_{ad}=0.5$. In Figure~\ref{figure:NE_general_ku<kad}, we identify the NE strategies when $\kappa_u=0.5$ and $\kappa_{ad}=1$. In the figures, we consider $\tilde{q}_f=1$, $\tilde{q}_p=1.5$, and $c=1$. Numerical results for a large set of parameters reveal that the pattern of NE strategies for different values of parameters is similar to one of the two pattern presented in Figures  \ref{figure:NE_general_ku>kad} and \ref{figure:NE_general_ku<kad}. Overall, the outcome in which the neutral ISP is driven out of the market occurs when  $t_N$ and $t_{NoN}$ are sufficiently small. As $t_N$ and $t_{NoN}$ increases, we expect to have equilibrium outcomes in which both ISPs are active. Next, we discuss about the trends we observe in the results.

Note that in Theorem~\ref{l-theorem:NE_stage1_new_q>}-I, we proved that,  for $\tilde{q}_{p}\geq \frac{t_N+2t_{NoN}}{\kappa_u+\kappa_{ad}}$ and $\tilde{q}_p\geq \frac{t_N+t_{NoN}}{\kappa_u} $, candidate strategy (Theorem~\ref{l-theorem:NE_stage1_new_q>}-I-1) is an NE. Numerical results for a large set of parameters also reveal that for   $\tilde{q}_{p}\geq  \frac{t_N+2t_{NoN}}{\kappa_u+\kappa_{ad}}$ and $\tilde{q}_p <  \frac{t_N+t_{NoN}}{\kappa_u} $, candidate strategy (Theorem~\ref{theorem:NE_stage1_new_q<}-1) is also an NE strategy. Note that these two strategies are the same and are listed as candidate strategy (a). Therefore, when  $  \frac{t_N+2t_{NoN}}{\kappa_u+\kappa_{ad}}\leq \tilde{q}_{p}$, (a) is an NE strategy.
Thus, the outcome in which the neutral ISP is  driven out of the market (candidate strategy (a)) occurs when $t_N$ and $t_{NoN}$ are sufficiently small. In this case, since the transport costs are low, EUs can easily switch ISPs. Thus, ISP NoN is able to attract all EUs by using some of the side payment fee she receives from the CP, and 
discounting the Internet access fee for EUs.

With increase in  $t_N$ or $t_{NoN}$, EUs have more inertia. Thus, one of the ISPs should provide a low Internet access fee for EUs to attract them all. However, in this case, ISPs prefer to maintain a high Internet access fee for EUs\footnote{As we discussed, when both  ISPs are active, the Internet connection fees are increasing with the transport costs. In other words, each ISP lock in some EUs and charge high Internet access fees to them.}, and split the EUs. Thus, as $t_N$ and $t_{NoN}$ increases, we expect to have equilibrium outcomes in which both ISPs are active. Numerical results reveal that if $\tilde{q}_{p}< \frac{t_N+2t_{NoN}}{\kappa_u+\kappa_{ad}}$ and $\tilde{q}_p\geq \frac{t_N+t_{NoN}}{\kappa_u} $, candidate strategy (Theorem~\ref{l-theorem:NE_stage1_new_q>}-I-2) is an NE. In addition, consider the lines $t_N+2t_{NoN}=\tilde{q}_p(\kappa_u+\kappa_{ad})$ and $t_N+t_{NoN}=\kappa_u \tilde{q}_p$. Results reveal that when the point $(t_N ,t_{NoN})$ is just above these lines, the candidate strategy (Theorem~\ref{theorem:NE_stage1_new_q<}-2) is an NE strategy.  When $(t_N,t_{NoN})$ is substantially above these lines, then Candidate strategy (Theorem~\ref{theorem:NE_stage1_new_q<}-3) is an NE strategy. This result have been proved in Theorem~\ref{theorem:bigt}, i.e. when either of $t_N$ and $t_{NoN}$ are high, then the unique NE outcome of the game is (Theorem~\ref{theorem:NE_stage1_new_q<}-3). In addition, when $(t_N,t_{NoN})$ is above these lines, but is in an intermediate range,  then no NE exists.
  

 Numerical results for large set of parameters also reveal that  the NE is unique, if it were to exist (in Figures there exists only one NE in each region). In addition, extensive numerical results reveal that candidate strategies (d) and (e) are never SPNE. Thus, henceforth we do not include (d) and (e) in our discussions about the results.
  


\subsection{Dependencies of $n^{eq}_{NoN}$, $\tilde{p}^{eq}$, and Payoffs of ISPs to $t_N$ and $t_{NoN}$}\label{section:numerical_more_results}

In this section, using numerical results, we discuss about the dependencies of $n^{eq}_{NoN}$, $\tilde{p}^{eq}$, and Payoffs of ISPs to $t_N$ and $t_{NoN}$. Note that in Section~\ref{section:summaryof resutls}, we explained that the relationship between $n^{eq}_{NoN}$ and the transport costs is not obvious from the expressions.  Thus, in this section, we provide intuitions for the behavior of $n^{eq}_{NoN}$, and subsequently $\tilde{p}^{eq}$ and the payoffs of ISPs with respect to $t_N$ and $t_{NoN}$.

\textbf{Numerical Results on $n^{eq}_{NoN}$:} Numerical results reveal that $n^{eq}_{NoN}$ is non-increasing  with respect to both transport costs. For instance, in Figure \ref{figure:nNoN_double}, we plot the value of $n^{eq}_{NoN}$ with respect  to $t_{NoN}$ and $t_N$, when $\tilde{q}_f=1$ and $\tilde{q}_p=1.5$.\footnote{Recall that $n^{eq}_N=1-n^{eq}_{NoN}$. Thus, we only plot $n^{eq}_{NoN}$.} 

%
%

\begin{figure}
	\begin{subfigure}{.25\textwidth}
		\centering
		\includegraphics[width=\linewidth]{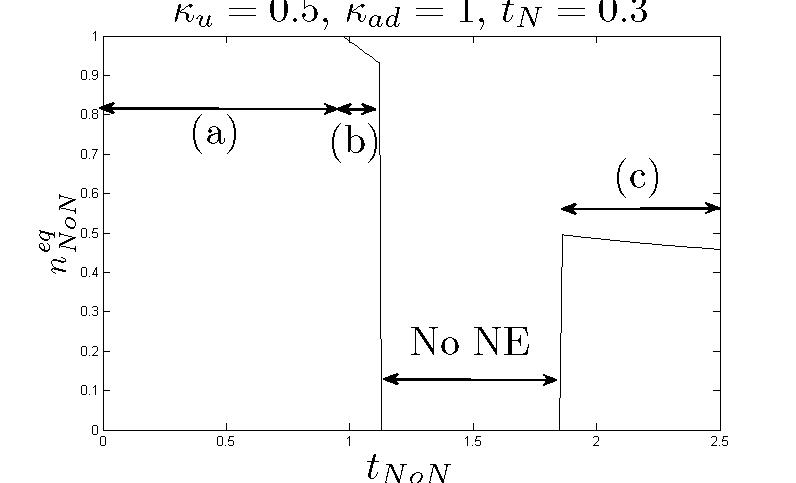}
		\label{fig:sfig1}
	\end{subfigure}%
	\begin{subfigure}{.25\textwidth}
		\centering
		\includegraphics[width=\linewidth]{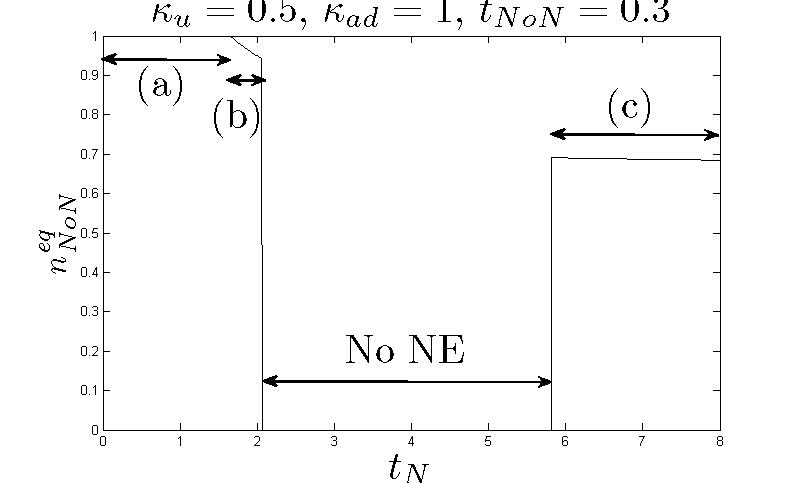}
		\label{fig:sfig2}
	\end{subfigure}
	\caption{$n^{eq}_{NoN}$ with respect to $t_N$ and $t_{NoN}$}\label{figure:nNoN_double}
\end{figure}

Note that for candidate strategy (a), as we know from the results, $n^{eq}_{NoN}=1$. To understand the results for candidate strategies (b) and (c), note that from \eqref{l-equ:EUs_linear}-I the number of EUs with each ISP has a decreasing relation with (i) the transport costs of the ISP, and (ii) the Internet access fee of the ISP which itself is increasing with both transport costs. In addition, the number of EUs with the ISP has an increasing relation with respect to the same parameters for the other ISP. Thus, different factors, some decreasing and some increasing with respect to the transport cost of an ISP, play a role in determining the number of EUs with each ISP. Overall, it turns out that the effect of increasing either of the transport costs decreases the incentive of EUs to  join ISP NoN. Thus, in candidate strategies (b) and (c), $n^{eq}_{NoN}$ is decreasing with respect to both transport costs.

 \textbf{Numerical Results on $\tilde{p}^{eq}$:} Note that the higher the number of EUs with ISP NoN, the higher would be the benefit of the CP from the premium quality. Thus, we expect the side-payment, i.e. $\tilde{p}^{eq}$ to be increasing with respect to number of EUs with ISP NoN. Results in Section \ref{section:summaryof resutls} also confirms this fact. Thus, the relationship between $\tilde{p}^{eq}$ and the transport costs is similar to the relationship between $n^{eq}_{NoN}$ and the transport costs. Therefore, in candidate strategy (b) and (c), the higher one of the transport costs, the lower would be the side payments. For instance, in Figures \ref{figure:tildep_kadlessku}, we plot the value of $\tilde{p}^{eq}$ with respect  to $t_{NoN}$ and $t_N$, respectively, when $\tilde{q}_f=1$ and $\tilde{q}_p=1.5$. 
 
 Note that as we discussed in Section~\ref{section:summaryof resutls}, in candidate strategy (b), $\tilde{p}^{eq}$ can be positive or negative. However, numerical results for  a large set of parameters reveal that $\tilde{p}^{eq}$ is positive, whenever this candidate strategy is NE.

%

\begin{figure}
	\begin{subfigure}{.25\textwidth}
		\centering
		\includegraphics[width=\linewidth]{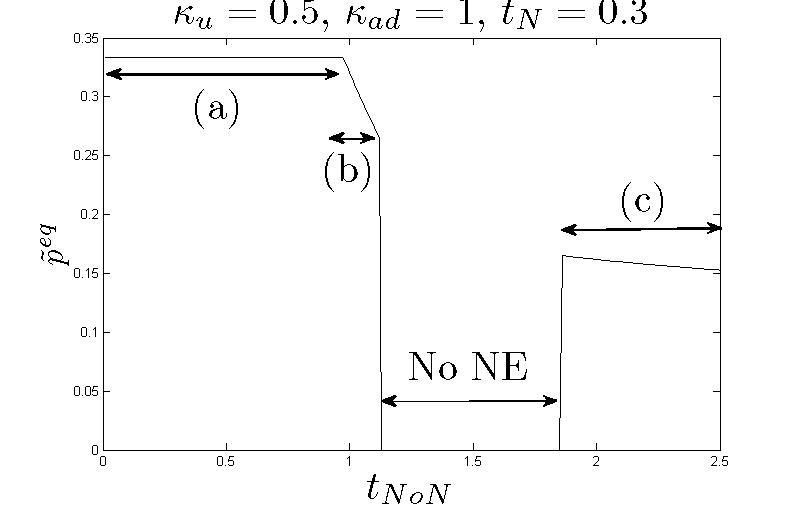}
		\label{figure:tildep_kadlessku_tNoN}
	\end{subfigure}%
	\begin{subfigure}{.25\textwidth}
		\centering
		\includegraphics[width=\linewidth]{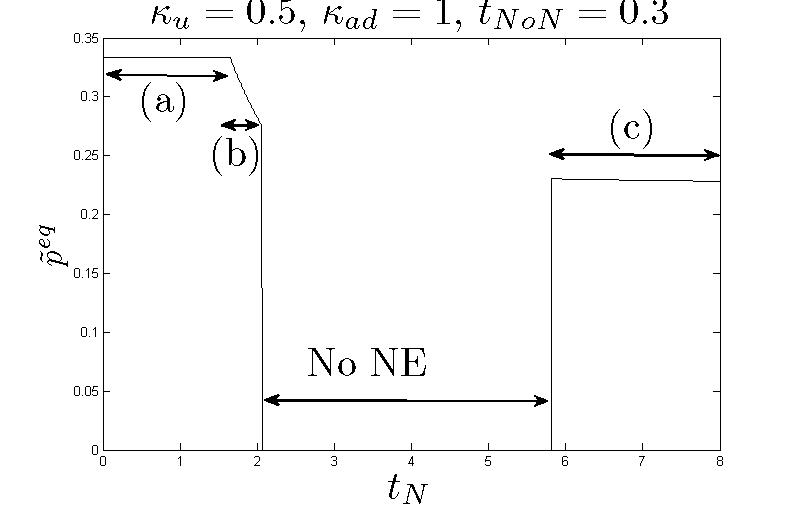}
		\label{figure:tildep_kadlessku_tN}
	\end{subfigure}
	\caption{$\tilde{p}^{eq}$ with respect to $t_N$ and $t_{NoN}$}\label{figure:tildep_kadlessku}
\end{figure}

 \textbf{Numerical Results on the Payoffs of ISPs:} Numerical results for the case  $\tilde{q}_f=1$ and $\tilde{q}_p=1.5$  are plotted in Figure \ref{figure:payoffs_double}. If  there is no NE strategy,  we plot the payoff of ISPs in the benchmark case, i.e. when both ISPs are neutral.
 
 \begin{figure}
 	\begin{subfigure}{.25\textwidth}
 		\centering
 		\includegraphics[width=\linewidth]{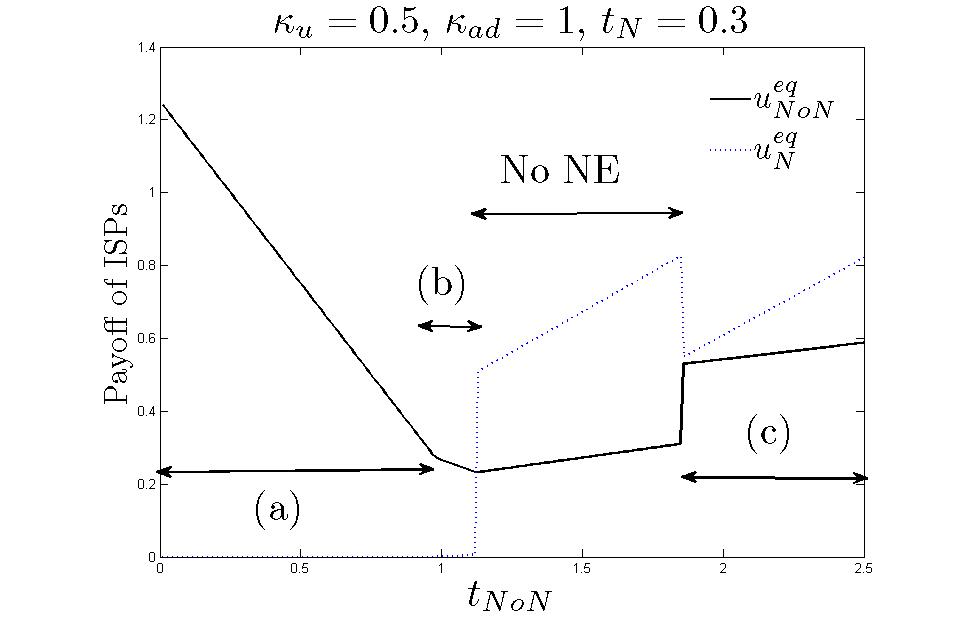}
 	\end{subfigure}%
 	\begin{subfigure}{.25\textwidth}
 		\centering
 		\includegraphics[width=\linewidth]{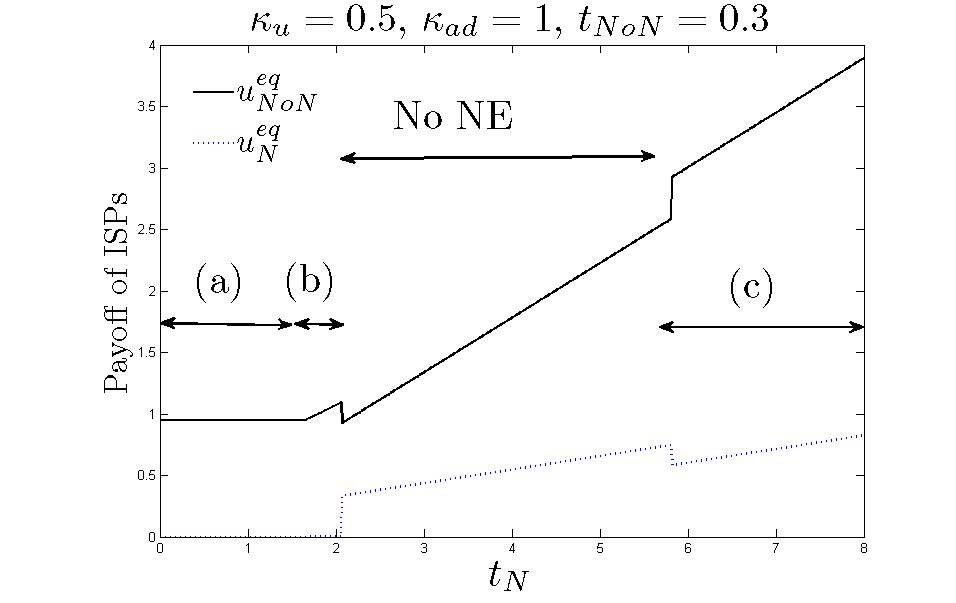}
 	\end{subfigure}
 	\caption{Payoff of ISPs with respect to $t_N$ and $t_{NoN}$}\label{figure:payoffs_double}
 \end{figure}

Note that when the market of power ISP NoN is small, i.e. the fraction $\frac{t_{N}}{t_N+t_{NoN}}$ is small, then the payoff of ISP NoN would be lower than the payoff of ISP N (Figure~\ref{figure:payoffs_double}-left). 

For candidate strategy (a), the payoff of ISP N is zero (since the number of EUs with this ISP is zero), and the payoff of ISP NoN is independent of $t_N$ (since ISP N is out of the market), but decreasing with respect to $t_{NoN}$ (since $p^{eq}_{NoN}$ is decreasing with $t_{NoN}$). 
   Intuitively, we expect  the utility of an ISP to be decreasing with respect to the transport cost of that ISP, and increasing with respect to the transport cost of the other ISP.  
    However, for some parameters and some of the candidate strategies, results reveal that the  payoff of an ISP is increasing with the transport cost of the ISP.  Next, we explain the underlying reasons for this counter-intuitive behavior.

   Note that the payoff of an ISP is increasing with (i) the number of EUs with the ISP and also (ii) the Internet access fee charged to the EUs. Recall that for the neutral ISP, in candidate strategies (b), (c), and the benchmark case, both of (i) and (ii) are increasing  with respect to  both transport costs. Thus, the payoff of ISP N is increasing with respect to both transport costs. On the other hand, for ISP NoN, the number of EUs is decreasing and the Internet access fee is increasing with the transport costs. Thus, depending on which of these factors overweights the other one, the payoff of ISP NoN can be decreasing or increasing with respect to the transport costs.

  

\subsection{Profits of Entities Due to Non-neutrality}\label{section:compare_results}

We compare  the results of the model and the benchmark case in which both ISPs are neutral. We compare   Internet access fees,  payoff of ISPs, the welfare of EUs, and  the payoff of the CP in Sections   \ref{section:comp_access}, \ref{section:comp_payoffs}, \ref{section:EUW}, and \ref{section:comp_CP}, respectively.

\subsubsection{Internet Access Fees}\label{section:comp_access} In a non-neutral case, the neutral ISP would always decrease her Internet access fee, while the Internet access fee of the non-neutral ISP could be higher or lower depending on the parameters of the market. We now provide insights on when each of these scenarios happens. 

First, note that the discount that ISP N provides for EUs in a non-neutral case, i.e. $p^{eq}_{N,B}-p^{eq}_{N}$, is always positive for candidate strategies (a), (b), and (c) (using previous results). Thus, the neutral ISP would always decrease her Internet access fee in a non-neutral scenario in order to compete with the non-neutral ISP which is now offering a better quality. 

  In a non-neutral regime, if (a) occurs, then  the discount that ISP NoN provides for EUs in a non-neutral case is $p^{eq}_{NoN,B}-p^{eq}_{NoN}=\frac{1}{3}(5t_{NoN}+t_N)-\kappa_u\tilde{q}_p$ (using the previous results). This discount can be negative or positive, and is decreasing with  $\kappa_u$ and $\tilde{q}_p$, and increasing with  $t_{NoN}$ and $t_N$. Thus, if (i) EUs are not sensitive to the quality, i.e. small $\kappa_u$, (ii) ISP NoN
does not provide a high premium quality, i.e. small $\tilde{q}_p$, (iii) end-users cannot switch between ISPs easily, i.e. $t_N$ and $t_{NoN}$  large enough, or a combination of these factors, then ISP NoN provides a cheaper Internet access fee for EUs in comparison to the neutral scenario.

For candidate strategies (b) (respectively, (c)), using the results in Sections~\ref{section:summaryof resutls} and \ref{section:summary_benchmark}, the amount of discount is $p^{eq}_{NoN,B}-p^{eq}_{NoN}=\frac{1}{3}\tilde{q}_p(2\kappa_{ad}-\kappa_u)$ (respectively, $p^{eq}_{NoN,B}-p^{eq}_{NoN}=\frac{1}{3}(\tilde{q}_p-\tilde{q}_f)(2\kappa_{ad}-\kappa_u)$). Thus, if $2\kappa_{ad}>\kappa_u$, i.e. the sensitivity of the CP is high enough, then the discount is positive and is increasing with the premium quality (respectively, the difference  between the premium and free quality). On the other hand, if the sensitivity of the CP is low, then the discount is negative, i.e. ISP NoN charges higher access fees to the EUs. The reason is that  if the CP is sensitive to the quality, ISP NoN can charge higher side-payments to the CP. Thus, she can provide some of these new revenue to EUs as a discount even though they receive a premium quality. This is not possible when the CP is not sensitive to the quality of her content. In this case, ISP NoN charges  the premium quality to the EUs directly, i.e. higher Internet access fees for EUs.

\subsubsection{Payoff of ISPs}\label{section:comp_payoffs} 
Consider the payoffs of the neutral and non-neutral ISPs under both neutral and non-neutral scenarios.  The difference in the payoffs for the case  $\kappa_u=0.5$, $\kappa_{ad}=1$, $\tilde{q}_f=1$, $\tilde{q}_p=1.5$, and $t_N=0.3$ are plotted in Figure \ref{figure:diff_payoffs_NoN_double}.\footnote{Using different parameters values yields the same intuitions.}

\emph{Results reveal that the neutral ISP will lose payoff in all of the non-neutral NE strategies}, i.e. those that yield $z^{eq}=1$ (Figure \ref{figure:diff_payoffs_NoN_double}-right). Note that in case (a), ISP N would be driven out of the market. Thus, $\pi^{eq}_{N}=0$, while $\pi^{eq}_{N,B}>0$. In cases (b) and (c), although ISP N is active, she has to subsidize the Internet connection fee for EUs to be able to compete with ISP NoN, while possibly can attract lower number of EUs. This yields a loss in the payoff under a non-neutral scenario. 

\emph{Results also reveal that for a wide range of parameters, ISP NoN receives a better payoff under a non-neutral scenario} (Figure \ref{figure:diff_payoffs_NoN_double}-left). We discussed that ISP NoN extracts the additional profit of the CP (from the premium quality her EUs receive) in a non-neutral scenario. In addition, we also explained that for some parameters ($\kappa_u>2\kappa_{ad}$), ISP NoN  charges higher prices to EUs. Even when ISP NoN subsidizes the Internet access fee for EUs ($2\kappa_{ad}>\kappa_{ad}$), she would compensate through the side payment charged to the CP (high $\kappa_{ad}$ yields a high side payment). Moreover, ISP NoN can potentially attract more EUs by providing a cheaper fee or a premium quality (or both). Thus, overall we expect the non-neutral ISP to receives a better payoff under a non-neutral regime. 

However, we can find scenarios in which the non-neutral ISP loses payoff by switching to non-neutrality. For example, with $\kappa_u=\kappa_{ad}=0.85$, $\tilde{q}_f=1$, $\tilde{q}_p=1.03$, $t_N=0.05$, and $t_{NoN}=0.8$, then $\pi^{eq}_{NoN}<\pi^{eq}_{NoN,B}$. In particular, the payoff of ISP NoN decreases in a non-neutral regime if the outcome of the market is (a), and  $\kappa_u$, $\kappa_{ad}$, $\tilde{q}_p-\tilde{q}_f$, and $\frac{t_{N}}{t_N+t_{NoN}}$ (the market power of ISP NoN) are small.

We now explain the underlying reason for this counter-intuitive result. Note that knowing that the other ISP has switched to non-neutrality, the neutral ISP would decrease her Internet access fee for EUs to compensate for the superior quality that her competitor offers. On the other hand, the non-neutral ISP also has to significantly decrease her Internet access fee for EUs (because of her low market power, competition, and low sensitivity of EUs to the quality), while not generating enough revenue from the side-payments received from the CP (because of low sensitivity of the CP to quality or a premium quality that is not significantly better than a free quality)\footnote{Note that the non-neutral ISP still extracts the additional profit she creates for EUs.}. This makes both ISPs, lose revenue in a non-neutral setting under the specified conditions.

\begin{figure}
 	\begin{subfigure}{.25\textwidth}
 		\centering
 		\includegraphics[width=\linewidth]{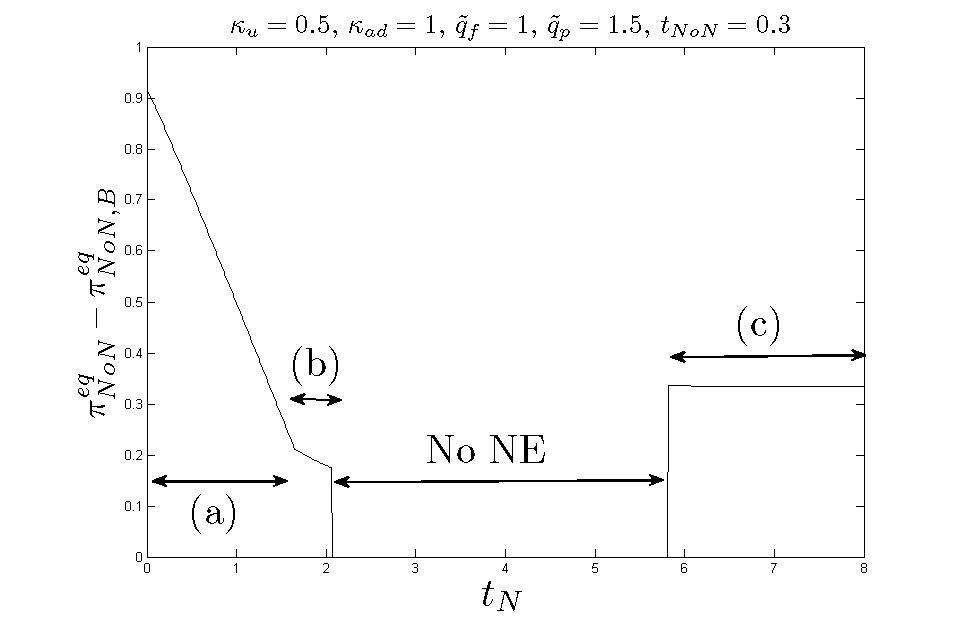}
 	\end{subfigure}%
 	\begin{subfigure}{.25\textwidth}
 		\centering
 		\includegraphics[width=\linewidth]{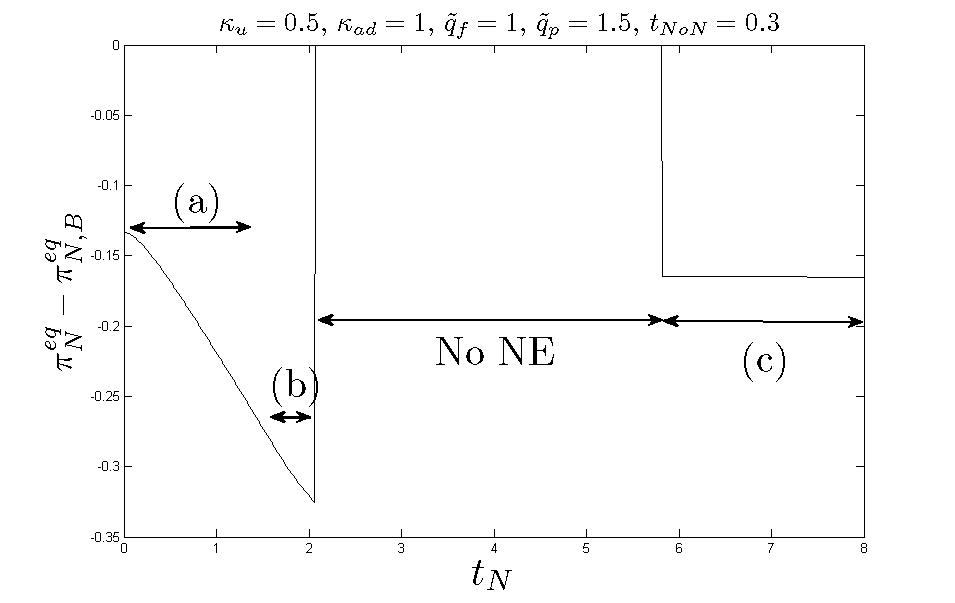}
 	\end{subfigure}
 	\caption{The difference between the payoff of ISPs for two scenarios with respect to $t_N$ and $t_{NoN}$}\label{figure:diff_payoffs_NoN_double}
 \end{figure}

\subsubsection{EU's Welfare}\label{section:EUW}
 \begin{figure}
 	\begin{subfigure}{.25\textwidth}
 		\centering
 		\includegraphics[width=\linewidth]{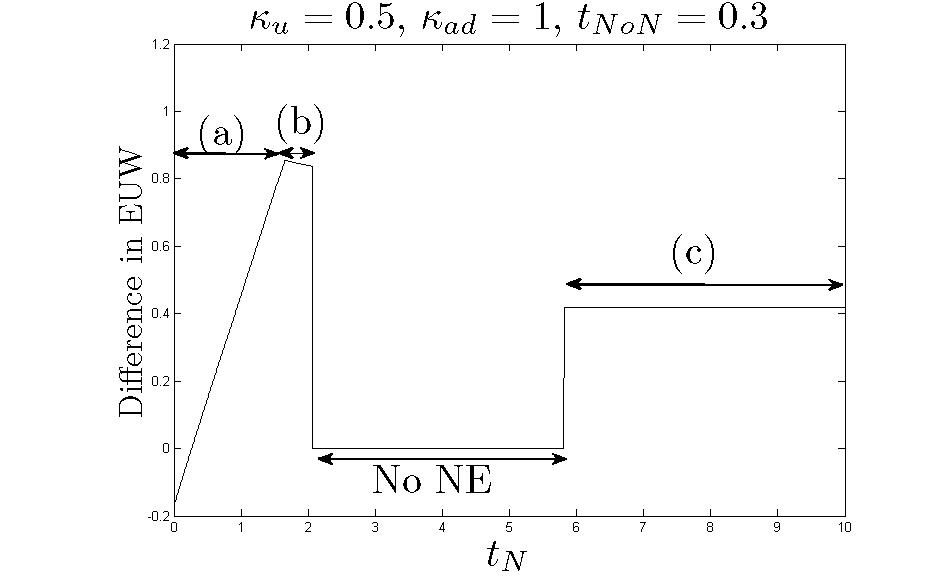}
 	\end{subfigure}%
 	\begin{subfigure}{.25\textwidth}
 		\centering
 		\includegraphics[width=\linewidth]{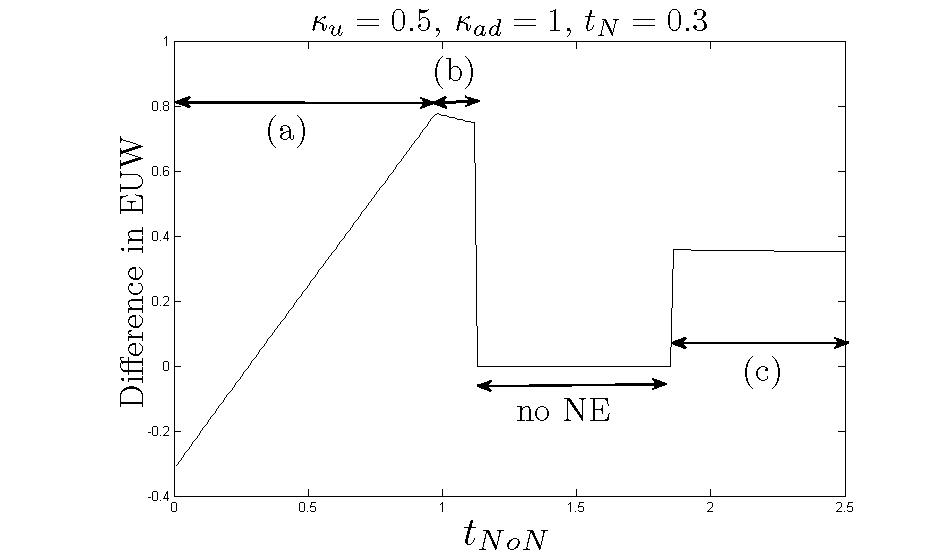}
 	\end{subfigure}
 	\caption{EUW with respect to $t_N$ and $t_{NoN}$}\label{figure:EUW_double}
 \end{figure}

 \begin{figure}
 	\begin{subfigure}{.25\textwidth}
 		\centering
 		\includegraphics[width=\linewidth]{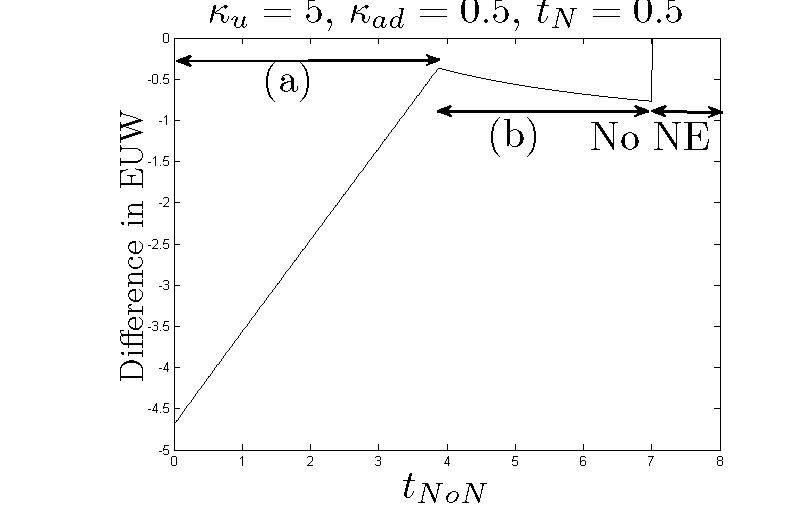}
 	\end{subfigure}%
 	\begin{subfigure}{.25\textwidth}
 		\centering
 		\includegraphics[width=\linewidth]{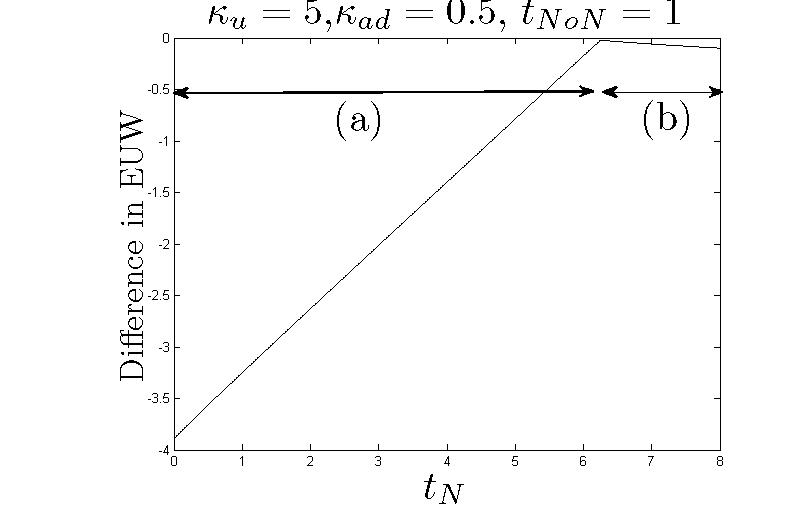}
 	\end{subfigure}
 	\caption{EUW with respect to $t_N$ and $t_{NoN}$}\label{figure:EUW_double_kularge}
 \end{figure}

Recall that from \eqref{l-equation:CP_2}-I, the utility of an EU who connects to the ISP $j\in \{\text{N},\text{NoN}\}$ located at distance $x_j$ of the ISP, and is receiving the content with quality $q_j$, is $u_{EU,j}=v^*+\kappa_u q_j-t_j x_j -p_j$. Now, let us define  the  Welfare of EUs (EUW) for an EU connected to ISP $j$ located at distance $x_j$ from this ISP to be $u_{EU,j}(x_j)-v^*=\kappa_u q_j-p_j-t_j x_j$. Note that we dropped the common valuation $v^*$ since it is equal for all EUs in all scenarios, and is only used to guarantee the full coverage of the market, i.e. to prevent negative utility for EUs. Thus, the total welfare of EUs is:

{\small
	\be\label{equ:SW}
	\ba 
	EUW&=\int_0^{n_N} \big{(}\kappa_u q_N-p_N-t_N x\big{)}dx \\
	&\qquad \qquad +\int_{n_N}^1\big{(}\kappa_u q_{NoN}-p_{NoN}-t_{NoN}(1-x)\big{)}dx\\
	&= (\kappa_u q_N-p_N)n_N-\frac{t_N}{2}n^2_N+(\kappa_u q_{NoN}-p_{NoN})n_{NoN}\\
	&\qquad \qquad -\frac{t_{NoN}}{2}n^2_{NoN}
	\ea 
	\ee}
\normalsize 
Note that since we dropped $v^*$, EUW could be negative. In Figures \ref{figure:EUW_double} and \ref{figure:EUW_double_kularge}, we plot the difference in the EUW of the non-neutral case with the benchmark case for various parameters of the market, when $\tilde{q}_f=1$ and $\tilde{q}_p=1.5$.

Results reveal that in general, EUW would be higher in a non-neutral setting if (i) the market power of ISP NoN is low, (ii) the sensitivity of the CP to the quality is high, or (iii) EUs are not very sensitive to the quality, or a combination of these conditions. However, when both transport costs are sufficiently small, or the sensitivity of EUs (respectively, the CP) to the quality is high (respectively, low), then  the benchmark case yields a better EUW in comparison to the non-neutral case. We next explain the reasons behind these results.

Consider the benchmark case. In this case, the welfare of EUs is dependent on the transport costs and the Internet access fees determined by ISPs N and NoN. Recall that both access fees are increasing with $t_N$ and $t_{NoN}$. Thus, intuitively, EUW of the  benchmark case is decreasing with $t_N$ and $t_{NoN}$.\footnote{Note that $n_N$ and $n_{NoN}$ are sum up to one. Thus, the effect of access fees on EUW is more than the effect of number of EUs with each ISP.} 

In case (a), in which only the non-neutral ISP is active, EUW is dependent on the Internet access fee of ISP NoN, i.e. $p^{eq}_{NoN}=c+\kappa_u \tilde{q}_p-t_{NoN}$.  Thus, EUW of the non-neutral scenario with outcome (a) is increasing with respect to $t_{NoN}$.  In other words, if $t_{NoN}$ is large, ISP NoN should provide a cheaper Internet access fee (subsidizing the access fee), to attract EUs and keep the neutral ISP out of the market. Thus, EUW would be high. In addition, the EUW is independent of $t_N$. Thus, as Figures~\ref{figure:EUW_double} and \ref{figure:EUW_double_kularge} confirm , the difference between the EUW of the non-neutral scenario in case (a) and the EUW of the benchmark case is increasing with respect to $t_N$ and $t_{NoN}$. 

We observe that when both transport costs are sufficiently small, the benchmark case yields a higher payoff than the non-neutral scenario. Note that if $t_{NoN}$ is small, i.e. EUs can join (switch to) ISP NoN without incurring high transport costs, ISP NoN attracts all EUs even when quoting a high Internet access fee for EUs (since it offers a premium quality). Thus, ISP NoN charges a high Internet access fee, and the EUW would be small. On the other hand, if $t_N$ is also small, the EUW of the benchmark case would be high (as discussed previously). Thus, when both transport costs are sufficiently small, we expect the benchmark case to yield a better EUW in comparison to the non-neutral case. Negative differences in Figures \ref{figure:EUW_double} and \ref{figure:EUW_double_kularge} confirm this intuition. Note that in Figure~\ref{figure:EUW_double_kularge}, because of high sensitivity of EUs to the quality, EUW of the neutral scenario is higher than the non-neutral scenario even when $t_N$ or $t_{NoN}$ are not small. Finally, observe that the maximum difference in the EUWs is achieved for the highest $t_N$ and $t_{NoN}$ by which the outcome of the game is (a), i.e. when only the non-neutral is active.

For candidate strategies (b) and (c), similar to the benchmark case, the Internet access fees are increasing with respect to $t_N$ and $t_{NoN}$. Thus, EUW is expected to be decreasing with respect to these transport costs. Results in the figures reveal that  the difference in EUWs is decreasing with respect to  $t_{NoN}$ and $t_N$.\footnote{From figures, it seems that the difference in (c) is constant. However, a closer look on the results reveals a slow decreasing behaviour.} This means that EUW of the non-neutral case decreases more than EUW of the benchmark case. This difference is positive when the sensitivity of the EUs to the quality is low, i.e. small $\kappa_u$ (Figure \ref{figure:EUW_double}), and negative when $\kappa_u$ is large (Figure \ref{figure:EUW_double_kularge}). Recall that the non-neutral ISP provides discount to EUs when the sensitivity of the CP to the quality is high enough. If not, ISP NoN charges  higher prices to EUs in comparison to the benchmark case.  This is the reason that EUW of the non-neutral case is lower than the benchmark case when EUs are highly sensitive to the quality they receive.

Thus,  the transport costs and the sensitivity of EUs and the CP to the quality  are the important factors  in comparing the EUW of the neutral and non-neutral scenario. Note that (as explained) the higher the sensitivity of the CP (respectively, EUs) to the quality, the higher (respectively, lower) would be EUW in the non-neutral case.

\subsubsection{Payoff of the CP}\label{section:comp_CP}  Using \eqref{l-equ:payoffCP_new}-I, the candidate strategies and their outcomes listed in Section~\ref{section:summaryof resutls}, and the outcome of the benchmark case in Section~\ref{section:summary_benchmark}, we can calculate the payoff of the CP in different outcomes. Results yield that  the equilibrium payoff of the CP in  all of the possible outcomes of the non-neutral scenario  and also in the benchmark scenario are equal and are $\pi^{eq}_{CP}=\pi^{eq}_{CP,B}=\kappa_{ad}\tilde{q}_f$. The reason is that the non-neutral ISP  is the leader in the this leader-follower game. Thus, knowing the parameters of the game and the tie-breaking assumption 2 of the CP, it can extract all the profits of the CP and make it indifferent between taking the non-neutral option and not taking it. 

\subsection{Does the Market Need to be Regulated?}\label{section:regulation}

We showed that in the presence  of a ``big" monopolistic CP and when EUs can switch between ISPs, if a non-neutral regime emerges, then neutral ISPs are likely to lose their market share, and are expected to be forced out of the market. In addition, in any NE outcome, the neutral ISP would lose payoff. Thus, if, the regulator is interested in keeping some of the neutral ISPs in the market\footnote{For example, the reason could be to prevent non-neutral ISPs from becoming monopoly or it could be the social pressure to preserve some neutrality in the market.}, she should provide incentives for them. These incentives could be in the form of monetary subsidies or tax deductions.

Although for many parameters, the payoff of the non-neutral ISP would be higher by adopting a non-neutral regime, as explained before, with certain conditions on the parameters, an ISP is likely to receive a lower payoff by switching to non-neutral regime. 
These conditions are when (i) EUs are not sensitive to the quality, i.e. small $\kappa_u$, (ii)  the CP is not sensitive to the quality her EUs receive, i.e. small $\kappa_{ad}$,  (iii) ISP NoN does not offer enough differentiation in the quality, i.e. small $\tilde{q}_p-\tilde{q}_f$,  (iv) the market power of the non-neutral ISP is low, or a combination of these factors. Thus, with these conditions a non-neutral regime is unlikely to emerge, and there is no need for a government intervention.	

\section{Discussions on Generalization of the Model}\label{section:implicationAssum}

Note that we assumed $q_N\in \{0,\tilde{q}_f\}$ and $q_{NoN}\in \{0,\tilde{q}_f,\tilde{q}_p\}$. This assumption can be generalized to selecting quality strategies from continuous sets, i.e. $q_N\in [0,\tilde{q}_f]$ and $q_{NoN}\in[0,\tilde{q}_p]$. In this case, the CP pays a side payment of $\tilde{p}q_{NoN}$ if she chooses $q_{NoN}\in (\tilde{q}_f,\tilde{q}_p]$. In Appendix~\ref{section:general}, we prove that our results herein would continue to hold under this generalization.

The result that over some parameters, an ISP can lose payoff by switching to a non-neutral regime is dependent on the assumption that the neutral and non-neutral ISPs first decide on the Internet access fees, and then the non-neutral ISP decides on the side-payment in the second stage. If we swap the order of these two stages, then the non-neutral ISP would not lose payoff by switching to non-neutrality since in this case, she is the leader of the game. Thus, ISP NoN, in the worst case, obtains the payoff of the neutral scenario. Recall that the reason for our choice  of the orders of the stages of the game is that Internet access fees are expected to be kept constant for a longer time horizons in comparison to  the side-payment.

Recall that in the hotelling model, we considered EUs to be distributed uniformly between zero and one. We now provide insights on a generalization of the uniform distribution to a non-uniform one. In that case, depending on the skewness of the probability measure, results would be similar to small $t_N$ or  $t_{NoN}$. For example, if the probability measure is skewed toward zero, i.e. EUs are distributed close to the neutral ISP, results would be similar to uniform distribution and $t_N$ small.
\bibliographystyle{IEEEtran}
\bibliography{bmc_article}

\appendices

\section{Proof of Theorem~\ref{theorem:NE_stage1_new_q<} and Corollary~\ref{corollary:outcome_q<}}\label{appendix:theorem:NE_stage1_new_q<}

First, we prove Theorem~\ref{theorem:NE_stage1_new_q<}. Then, using the results of this theorem, we prove Corollary~\ref{corollary:outcome_q<}.

\begin{proof}
In this case, note that $\tilde{q}_f<\tilde{q}_{p}<\frac{t_N+t_{NoN}}{\kappa_u}$. Thus, we characterize the optimum strategies for the CP using items 1, 2, and 4 of Theorem~\ref{l-theorem:p_tilde_new}-I. 

First, note that by Lemma~\ref{l-lemma:deltap_t}-I, since $\tilde{q}_{p}< \frac{t_N+t_{NoN}}{\kappa_u}$,  $\kappa_u \tilde{q}_{p}-t_{NoN} <  \Delta p_t < t_N+\kappa_u(\tilde{q}_{p}-\tilde{q}_f)$, where $\Delta p_{t}=\kappa_{u}(2\tilde{q}_{p}-\tilde{q}_f)-t_{NoN}$ characterized in Lemma~\ref{l-lemma:thresh_on_delta_p_z1_new}-I. Thus, using this result, we characterize  the regions characterized in items 1, 2, and 4 of Theorem~\ref{l-theorem:p_tilde_new}-I. We denote 
 $\Delta p\leq \kappa_u \tilde{q}_{p}-t_{NoN}$ by region A, $\kappa_u \tilde{q}_{p}-t_{NoN}<\Delta p<\kappa_u (2\tilde{q}_{p}-\tilde{q}_f)-t_{NoN}$ by region $B_1$, $\kappa_u (2\tilde{q}_{p}-\tilde{q}_f)-t_{NoN}\leq \Delta p<t_N+\kappa_u(\tilde{q}_{p}-\tilde{q}_f)$ by region C, $t_N+\kappa_u(\tilde{q}_{p}-\tilde{q}_f)\leq \Delta p<t_N+\kappa_u \tilde{q}_{p}$ by set $B_2$, and $\Delta p\geq t_N+\kappa_u \tilde{q}_{p}$ by D. Using Theorem~\ref{l-theorem:p_tilde_new}-I, if $z^{eq}=1$, then $\Delta p<t_N+\kappa_u \tilde{q}_{p}$. Thus, we characterize any possible NE strategies by which $z^{eq}=1$, in regions A and $B_1$, $C$, and $B_2$:

\textbf{Case A:} First, we consider $\Delta p\leq \kappa_u \tilde{q}_{p}-t_{NoN}$. In this case, we show that the payoff of ISP NoN is an increasing function of $\Delta p$. Then, we characterize the NE as  $p^{eq}_{NoN}=c+\kappa_u \tilde{q}_{p}-t_{NoN}$ and $p^{eq}_N=c$, using the fact that when choosing an NE, no player can increase her payoff by unilaterally changing her strategy. 

Note that by Theorem~\ref{l-theorem:p_tilde_new}-I, for region A, $(q^{eq}_N,q^{eq}_{NoN})=(0,\tilde{q}_{p})\in F^L_1$ if and only if $\tilde{p}\leq \tilde{p}_{t,1}=\kappa_{ad}(1-\frac{\tilde{q}_f}{\tilde{q}_{p}})$. In addition, by Theorem~\ref{l-theorem:NE_stage2_new}-I, if $z^{eq}=1$ then $\tilde{p}^{eq}=\tilde{p}_{t,1}=\kappa_{ad}(1-\frac{\tilde{q}_f}{\tilde{q}_{p}})$ (Definition~\ref{l-def:pt1,pt2}-I).  Thus, in  this region, if $z^{eq}=1$, the payoff of ISP NoN is equal to $p_{NoN}-c+\tilde{q}_{p}\tilde{p}_{t,1}$ (by \eqref{l-equ:payoffISPsGeneral_new}-I) since $n_{NoN}=1$. Therefore, the payoff is an increasing function of $p_{NoN}$. In addition, note that in region A, $n_N=0$ and regardless of $p_N$, the neutral ISP receives a payoff of zero (by \eqref{l-equ:payoffISPsGeneral_new}-I). Thus, $p^{eq}_{NoN}$, i.e. the equilibrium Internet access fee, should be such that the neutral ISP cannot get a positive payoff  by increasing or decreasing $p_N$, and changing the region of $\Delta p$ to $B_1$, $B_2$, or $C$. Using this condition, we find the equilibrium strategy.

First consider a unilateral deviation by ISP N. Note that increasing $p_N$ decreases $\Delta p$, and cannot change the region of $\Delta p$.\footnote{Recall that in this region, $n_N=0$, and ISP N fetch a payoff of zero.} Thus, a deviation of this kind would not be profitable. We claim that by decreasing $p_N$ to $p'_N$ such that $p_{NoN}-p'_N>\kappa_u \tilde{q}_{p}-t_{NoN}$, the ISP N can fetch a positive payoff as long as $p'_N>c$ (the claim is proved in the next paragraph). Therefore, in the equilibrium, $p^{eq}_{NoN}$ is such that even with $p'_N=c$ (the minimum plausible price), $\Delta p \leq \kappa_u \tilde{q}_{p}-t_{NoN}$. Thus, $p^{eq}_{NoN}\leq c+ \kappa_u \tilde{q}_{p}-t_{NoN}$.\footnote{Otherwise, there exists a $p'_N>c$ by which  $\Delta p > \kappa_u \tilde{q}_{p}-t_{NoN}$.} Given that the payoff of ISP NoN is an increasing function of $p_{NoN}$, we get  $p^{eq}_{NoN}=c+\kappa_u \tilde{q}_{p}-t_{NoN}$. In addition, we claim that $p^{eq}_N=c$. If not, then $p^{eq}_N>c$. In this case, $\Delta p=p^{eq}_N-p^{eq}_{NoN}<\kappa_u \tilde{q}_{p}-t_{NoN}$.  We argued that the payoff of ISP NoN is an increasing function of $p_{NoN}$. Thus, by increasing $p_{NoN}$ such that $\Delta p=\kappa_u \tilde{q}_{p}-t_{NoN}$, ISP NoN can increase her payoff, which is a contradiction with $p^{eq}_N$ and $p^{eq}_{NoN}$ being NE strategies.


To prove the claim, note that if  $p_{NoN}-p'_N>\kappa_u \tilde{q}_{p}-t_{NoN}$, then either (i) $z^{eq}=0$ or (ii) $z^{eq}=1$. For case (i), since  $\kappa_u \tilde{q}_{p}-t_{NoN}>-t_{NoN}$, when $\Delta p>\kappa_u \tilde{q}_{p}-t_{NoN}$, then $(q^{eq}_N,q^{eq}_{NoN})$ is of the form of items 1 or  2 of Theorem~\ref{l-lemma:CP_z=0_new}-I.  Thus, $n_N>0$. Therefore ISP N can fetch a positive payoff as long as $p_N>c$ (by \eqref{l-equ:payoffISPsGeneral_new}-I). Now consider case (ii), i.e. $z^{eq}=1$. In this case, if $z^{eq}=1$, then by using item 2 of Thoerem~\ref{l-theorem:p_tilde_new}-I, $n_N>0$ (since solutions that yield $z^{eq}=1$ are in $F^I$.). Thus, ISP N can fetch a positive payoff as long as $p_N>c$ (by \eqref{l-equ:payoffISPsGeneral_new}-I). This completes the proof of the claim that by decreasing $p_N$ to $p'_N$ such that $p_{NoN}-p'_N>\kappa_u \tilde{q}_{p}-t_{NoN}$, the ISP N can fetch a positive payoff as long as $p'_N>c$.

Therefore, the NE strategies are $p^{eq}_{NoN}=c+\kappa_u \tilde{q}_{p}-t_{NoN}$ and $p^{eq}_N=c$, and the payoff of the ISP NoN at this price by \eqref{l-equ:payoffISPsGeneral_new}-I and $\tilde{p}_{t,1}=\kappa_{ad}(1-\frac{\tilde{q}_f}{\tilde{q}_{p}})$ is equal to  (note that $n_{NoN}=1$), and
\be \label{equ:payoff_NoN_eq_1_new}
\pi^{eq}_{NoN}=\kappa_u \tilde{q}_{p}-t_{NoN}+\tilde{q}_{p} \tilde{p}_{t,1}=\kappa_u \tilde{q}_p-t_{NoN}+\kappa_{ad}(\tilde{q}_p-\tilde{q}_f)
\ee
which is strictly positive since $\tilde{q}_{p}>\frac{t_N+t_{NoN}}{\kappa_u}$ and $\tilde{q}_p>\tilde{q}_f$. Note that Lemma~\ref{l-lemma:appendix_caseA}-I yields that with $p^{eq}_N$ and $p^{eq}_{NoN}$, $z^{eq}=1$. The first item of the theorem follows. 


\textbf{Case $B_1$ and $B_2$:} Now, consider regions $B_1$ and $B_2$, i.e. $\kappa_u \tilde{q}_{p}-t_{NoN}<\Delta p<\kappa_u (2\tilde{q}_{p}-\tilde{q}_f)-t_{NoN}$ and $t_N+\kappa_u(\tilde{q}_{p}-\tilde{q}_f)\leq \Delta p<t_N+\kappa_u \tilde{q}_{p}$, respectively.

Note that in these regions,  by items 2-a-ii and 2-b of Theorem~\ref{l-theorem:p_tilde_new}-I, if $z^{eq}=1$, then $(q^{eq}_N,q^{eq}_{NoN})=(0,\tilde{q}_{p})\in F^I_1$. In addition, by Theorem~\ref{l-theorem:NE_stage2_new}-I, $\tilde{p}^{eq}=\tilde{p}_{t,2}=\kappa_{ad} (n_{NoN}-\frac{\tilde{q}_f}{\tilde{q}_{p}})$ and $n_{NoN}=\frac{t_N+\kappa_u \tilde{q}_{p}-\Delta p}{t_N+t_{NoN}}$ (Definition~\ref{l-def:pt1,pt2}-I). Thus, by \eqref{l-equ:payoffISPsGeneral_new}-I, the payoff of ISP NoN in this region is $\pi_{NoN,B}(p_{NoN},\tilde{p}_{t,2})=(p_{NoN}-c)n_{NoN}+\tilde{p}_{t,2} \tilde{q}_{p}$, and the payoff of ISP N is  $\pi_{N,B}=(p_{N}-c)(1-n_{NoN})$. Note that $\tilde{p}_{t,2} \tilde{q}_{p}=\kappa_{ad}(\tilde{q}_{p} n_{NoN}-\tilde{q}_f)$. Thus, using the expression of $n_{NoN}$, the payoffs are:

\footnotesize
\be\label{equ:equ:Theorem7_help_2}
\ba 
\pi_{NoN,B}(p_{NoN},\tilde{p}_{t,2})&=(p_{NoN}-c+\kappa_{ad}\tilde{q}_{p})(\frac{t_N+\kappa_u \tilde{q}_{p}+p_N-p_{NoN}}{t_N+t_{NoN}})\\
&\qquad \qquad \qquad -\kappa_{ad}\tilde{q}_f
\ea
\ee 
\be
\pi_{N,B}(p_N)=(p_N-c)(\frac{t_{NoN}-\kappa_u \tilde{q}_{p}+p_{NoN}-p_N}{t_N+t_{NoN}}) 
\ee
\normalsize

First, we rule out any NE such that $\Delta p^{eq}=t_N+\kappa_u(\tilde{q}_{p}-\tilde{q}_f)$. Suppose that $\Delta p^{eq}=p^{eq}_{NoN}-p^{eq}_{N}=t_N+\kappa_u(\tilde{q}_{p}-\tilde{q}_f)$. Consider a deviation by ISP N such that $p'_N=p^{eq}_{N}+\epsilon>c$ for $\epsilon>0$ such that $\Delta p'=p^{eq}_{NoN}-p'_N$ to be in region C. Note that by item 2-a-i of Theorem~\ref{l-theorem:p_tilde_new}-I, in region C, $(q^{eq}_N,q^{eq}_{NoN})=(\tilde{q}_f,\tilde{q}_p)\in F_1^I$. Thus, the payoff of this ISP with this deviation is (by \eqref{equ:UN_new}):
\be \nonumber
\pi_{N}(p'_N)=(p^{eq}_N+\epsilon-c)(\frac{t_{NoN}-\kappa_u (\tilde{q}_{p}-\tilde{q}_f)+p^{eq}_{NoN}-p^{eq}_N-\epsilon}{t_N+t_{NoN}}) 
\ee 

Note that $\lim_{\epsilon\downarrow 0} \pi_N(p'_N)>\pi_{N,B}(p^{eq}_N)$. Thus, for $\epsilon>0$ small enough, this deviation is profitable. Thus, the strategies by which $\Delta p^{eq}=t_N+\kappa_u(\tilde{q}_{p}-\tilde{q}_f)$ cannot be NE. 

Now, we characterize any NE in    $\kappa_u \tilde{q}_{p}-t_{NoN}<\Delta p<\kappa_u (2\tilde{q}_{p}-\tilde{q}_f)-t_{NoN}$ and $t_N+\kappa_u(\tilde{q}_{p}-\tilde{q}_f)<\Delta p<t_N+\kappa_u \tilde{q}_{p}$. Note that any NE inside this region should satisfy the first order optimality condition (note that the payoffs are concave). Thus,
\be
\ba
 \frac{d \pi_N}{d p_{N}}=0& \Rightarrow t_{NoN}-\kappa_u \tilde{q}_{p} + p_{NoN}-2 p_N+c=0\\
\frac{d \pi_{NoN,B}}{d p_{NoN}}=0& \Rightarrow t_N+\tilde{q}_{p}(\kappa_u-\kappa_{ad})+p_N-2 p_{NoN}+c=0
\ea 
\ee

Solving the equation yields:
\be \label{equ:equ:B>_eq_pNoN_new}
p^{eq}_{NoN}=c+\frac{t_{NoN}+2t_N+  \tilde{q}_{p}(\kappa_u -2\kappa_{ad})}{3}
\ee 
\be \label{equ:equ:B>_eq_pN_new}
p^{eq}_{N}=c+\frac{2t_{NoN}+t_N-\tilde{q}_{p}(\kappa_u +\kappa_{ad})}{3}
\ee


First, note that if $\tilde{q}_{p}>\frac{2t_{NoN}+t_N}{\kappa_u+\kappa_{ad}}$, then $p^{eq}_N<c$, and $p^{eq}_N$ cannot be an NE. Thus, the first necessary condition for these strategies to be NE is  $\tilde{q}_{p}\leq \frac{2t_{NoN}+t_N}{\kappa_u+\kappa_{ad}}$.
 In addition, by Theorem~\ref{l-theorem:NE_stage2_new_suff}-I,\footnote{Note that in Regions $B_1$ and $B_2$, $\Delta p<t_N+\kappa_u \tilde{q}_p$.} $\pi^{eq}_{NoN}>\pi_{NoN,z=0}(\tilde{p}^{eq}_{NoN},\tilde{p})$ (for these strategies to yield $z^{eq}=1$).  The second item of the theorem follows.

\textbf{Case C: } Now, consider region C, i.e. $\Delta p_t=\kappa_u (2\tilde{q}_{p}-\tilde{q}_f)-t_{NoN}\leq \Delta p<t_N+\kappa_u (\tilde{q}_{p}-\tilde{q}_f)$. Note that in this regions,  by items 2-a-i of Theorem~\ref{l-theorem:p_tilde_new}-I, if $z^{eq}=1$, then $(q^{eq}_N,q^{eq}_{NoN})=(\tilde{q}_f,\tilde{q}_{p})\in F^I_1$. In addition, by Theorem~\ref{l-theorem:NE_stage2_new}-I and Definition \ref{l-def:pt}-I, $\tilde{p}^{eq}=\tilde{p}_{t,3}=\kappa_{ad}n_{NoN} (1-\frac{\tilde{q}_f}{\tilde{q}_{p}})$ and $n_{NoN}=\frac{t_N+\kappa_u (\tilde{q}_{p}-\tilde{q}_f)-\Delta p}{t_N+t_{NoN}}$ (Definition~\ref{l-def:pt1,pt2}-I). Thus, by \eqref{l-equ:payoffISPsGeneral_new}-I, the payoff of ISP NoN in this region is $\pi_{NoN,C}(p_{NoN},\tilde{p}_{t,3})=(p_{NoN}-c)n_{NoN}+\tilde{p}_{t,3} \tilde{q}_{p}$, and the payoff of ISP N is  $\pi_{N,B}=(p_{N}-c)(1-n_{NoN})$. Note that $\tilde{p}_{t,3} \tilde{q}_{p}=\kappa_{ad}n_{NoN}(\tilde{q}_{p} -\tilde{q}_f)$. Thus, using the expression of $n_{NoN}$, the payoffs are:

\footnotesize
\be\label{equ:equ:Theorem7_help}
\ba 
\pi_{NoN,C}&(p_{NoN},\tilde{p}_{t,3})=\\
&(p_{NoN}-c+\kappa_{ad}(\tilde{q}_{p}-\tilde{q}_f))(\frac{t_N+\kappa_u (\tilde{q}_{p}-\tilde{q}_f)+p_N-p_{NoN}}{t_N+t_{NoN}})
\ea 
\ee 
\be\label{equ:equ:Theorem7_help_5}
\pi_{N,C}(p_N)=(p_N-c)(\frac{t_{NoN}-\kappa_u (\tilde{q}_{p}-\tilde{q}_f)+p_{NoN}-p_N}{t_N+t_{NoN}}) 
\ee
\normalsize

First, in Part C-1, we characterize any NE in region  $\kappa_u (2\tilde{q}_{p}-\tilde{q}_f)-t_{NoN}< \Delta p<t_N+\kappa_u (\tilde{q}_{p}-\tilde{q}_f)$. Later, in Part C-2, we consider the case that $\Delta p^{eq}=\kappa_u(2\tilde{q}_{p}-\tilde{q}_f)-t_{NoN}$. 

\textbf{Part C-1:}
 Note that any NE in region  $\kappa_u (2\tilde{q}_{p}-\tilde{q}_f)-t_{NoN}< \Delta p<t_N+\kappa_u (\tilde{q}_{p}-\tilde{q}_f)$ should satisfy the first order optimality condition (note that the payoffs are concave). Thus,
 
 \small
\be
\ba \label{equ:equ:Theorem7_help_4}
 \frac{d \pi_{N,C}}{d p_{N}}=0& \Rightarrow t_{NoN}-\kappa_u (\tilde{q}_{p}-\tilde{q}_f) + p_{NoN}-2 p_N+c=0\\
\frac{d \pi_{NoN,C}}{d p_{NoN}}=0& \Rightarrow t_N+(\tilde{q}_{p}-\tilde{q}_f)(\kappa_u-\kappa_{ad})+p_N-2 p_{NoN}+c=0
\ea 
\ee  
\normalsize

Solving the equation yields:

\small
\be \label{equ:equ:B>_eq_pNoN_new}
p^{eq}_{NoN}=c+\frac{t_{NoN}+2t_N+  (\tilde{q}_{p}-\tilde{q}_f)(\kappa_u -2\kappa_{ad})}{3}
\ee 
\be \label{equ:equ:B>_eq_pN_new}
p^{eq}_{N}=c+\frac{2t_{NoN}+t_N-(\tilde{q}_{p}-\tilde{q}_f)(\kappa_u +\kappa_{ad})}{3}
\ee 
\normalsize

First, note that if $\tilde{q}_{p}-\tilde{q}_f>\frac{2t_{NoN}+t_N}{\kappa_u+\kappa_{ad}}$, then $p^{eq}_N<c$, and $p^{eq}_N$ cannot be an NE. Thus, the first necessary condition for these strategies to be NE is $\tilde{q}_{p}-\tilde{q}_f\leq \frac{2t_{NoN}+t_N}{\kappa_u+\kappa_{ad}}$. 
In addition, by Theorem~\ref{l-theorem:NE_stage2_new_suff}-I,   $\pi^{eq}_{NoN}>\pi_{NoN,z=0}(\tilde{p}^{eq}_{NoN},\tilde{p})$ (in order for these strategies to yield $z^{eq}=1$).  The third item of the theorem follows.


\textbf{Part C-2:} Now, consider $p^{eq}_N$ and $p^{eq}_{NoN}$ such that $\Delta p^{eq}=p^{eq}_{NoN}-p^{eq}_{N}=\kappa_u(2\tilde{q}_{p}-\tilde{q}_f)-t_{NoN}$. These strategies are not NE if ISP NoN can strictly increase her payoff by decreasing her price such that $\Delta p$ in region $B_1$. Note that using \eqref{equ:equ:Theorem7_help} and the expression for $
\Delta p^{eq}$, the payoff of ISP NoN in this case is:

\small
\be\label{equ:equ:Theorem7_help_3}
\pi_{NoN}(p^{eq}_{NoN},\tilde{p}_{t,3})=(p_{NoN}-c+\kappa_{ad}(\tilde{q}_{p}-\tilde{q}_f))(\frac{t_N-\kappa_u \tilde{q}_{p}+t_{NoN}}{t_N+t_{NoN}})
\ee
\normalsize

By choosing $p'_{NoN}=p^{eq}_{NoN}-\epsilon$ such that $\epsilon\downarrow 0$, ISP NoN can get a limit payoff of (since $\Delta p=\Delta p^{eq}$ when $\epsilon\rightarrow 0$, and it is in  region $B_1$, and using \eqref{equ:equ:Theorem7_help_2}):

\footnotesize
 $$
 \ba 
\pi'_{NoN}&=\lim_{\epsilon\downarrow 0} \pi^{eq}_{NoN}(p_{NoN}-\epsilon,\tilde{p}_{t,3})\\
&=(p_{NoN}^{eq}-c+\kappa_{ad}\tilde{q}_{p})(\frac{t_N-\kappa_u (\tilde{q}_{p}-\tilde{q}_f)+t_{NoN}}{t_N+t_{NoN}})-\kappa_{ad}\tilde{q}_f
\ea
$$
\normalsize

Thus, $p^{eq}_N$ and $p^{eq}_{NoN}$ such that $\Delta p^{eq}=p^{eq}_{NoN}-p^{eq}_{N}=\kappa_u(2\tilde{q}_{p}-\tilde{q}_f)-t_{NoN}$ are not NE if:
\footnotesize
$$
\ba
\pi'_{NoN}&>\pi_{NoN}(p^{eq}_{NoN},\tilde{p}_{t,3})\\
&\iff (p^{eq}_{NoN}-c+\kappa_{ad}\tilde{q}_{p})\frac{\kappa_u \tilde{q}_f}{t_N+t_{NoN}}-\frac{\kappa_{ad}\kappa_u\tilde{q}_f\tilde{q}_{p}}{t_N+t_{NoN}}>0\\
&\iff p^{eq}_{NoN}>c
\ea
$$
\normalsize
Thus, the necessary condition for these strategy to be NE is $p^{eq}_{NoN}\leq c$. 
 Note that from \eqref{equ:equ:Theorem7_help} and  \eqref{equ:equ:Theorem7_help_5}, since $\Delta p$ is fixed, the payoffs of ISP NoN and N are an increasing function of $p_{NoN}$ and $p_{N}$, respectively. Thus, $p^{eq}_{NoN}=c$, and $p^{eq}_N=c-\kappa_u(2\tilde{q}_{p}-\tilde{q}_f)+t_{NoN}$. Note that a necessary condition for $p^{eq}_N$ to be an NE is that $p^{eq}_N\geq c$. Thus, one necessary condition is that  $2\tilde{q}_{p}-\tilde{q}_f\leq \frac{t_{NoN}}{\kappa_{u}}$. In addition,  $\pi_{NoN}(\tilde{p}^{eq}_{NoN},\tilde{p}_{t,3})>\pi_{NoN,z=0}(\tilde{p}^{eq}_{NoN},\tilde{p})$ (using Theorem~\ref{l-theorem:NE_stage2_new_suff}-I, in order for these strategies to yield $z^{eq}=1$). The forth item of the theorem follows. 
\end{proof}

We now prove Corollary~\ref{corollary:outcome_q<}:

\begin{proof}
	First, consider Strategy 1 of Theorem~\ref{theorem:NE_stage1_new_q<}. Item 1 of Theorem~\ref{l-theorem:p_tilde_new}-I yields that  $(q^{eq}_N,q^{eq}_{NoN})=(0,\tilde{q}_p)\in F^L_1$. Thus,  $n^{eq}_N=0$, and $n^{eq}_{NoN}=1$. In addition, by Theorem~\ref{l-theorem:NE_stage2_new}-I,   $\tilde{p}^{eq}=\tilde{p}_{t,1}=\kappa_{ad}(1-\frac{\tilde{q}_f}{\tilde{q}_p})$. 
	
	Now, consider Strategy 2 of Theorem~\ref{theorem:NE_stage1_new_q<}. Note that we constructed this strategy such that $\Delta p$ satisfies items 2-a-ii or 2-b  of Theorem~\ref{l-theorem:p_tilde_new}-I. Thus, $(q^{eq}_N,q^{eq}_{NoN})=(0,\tilde{q}_p)\in F^I_1$. In addition, by Theorem~\ref{l-theorem:NE_stage2_new}-I, $\tilde{p}^{eq}=\tilde{p}_{t,2}=\kappa_{ad}(n^{eq}_{NoN}-\frac{\tilde{q}_f}{\tilde{q}_p})$. Using the expression for $\Delta p=p^{eq}_{NoN}-p^{eq}_N$, and \eqref{l-equ:EUs_linear}-I, the expressions for $n^{eq}_N$ and $n^{eq}_{NoN}$ follow.
	
	Consider Strategies 3 and 4 of Theorem \ref{theorem:NE_stage1_new_q<}. In this case, $\Delta p$ satisfies item 2-a-i of Theorem~\ref{l-theorem:p_tilde_new}-I (by construction of these strategies). Thus,   $(q^{eq}_N,q^{eq}_{NoN})=(\tilde{q}_f,\tilde{q}_p)\in F^I_1$. In addition, by Theorem~\ref{l-theorem:NE_stage2_new}-I, $\tilde{p}^{eq}=\tilde{p}_{t,3}=\kappa_{ad}n^{eq}_{NoN}(1-\frac{\tilde{q}_f}{\tilde{q}_p})$. Using the expression of $\Delta p^{eq}$ for each of the strategies, $n^{eq}_N$ and $n^{eq}_{NoN}$ follow.
\end{proof}

\section{Proof of Theorem~\ref{theorem:bigt}}\label{appendix:theorem:bigt}
\begin{proof}
	We use Theorem~\ref{theorem:NE_stage1_new_q<} to prove the result. First, in Part 1, we prove that when one of $t_N$ or $t_{NoN}$ is large, then  strategies 1), 2), and 4) listed in Theorem~\ref{theorem:NE_stage1_new_q<} are not NE. In Part 2, we prove that when one of $t_N$ or $t_{NoN}$ is high, then strategy 3) of Theorem~\ref{theorem:NE_stage1_new_q<} is an NE. This completes the proof of the theorem.
	
	\textbf{Part 1:} We prove that strategies 1), 2), and 4) listed in Theorem~\ref{theorem:NE_stage1_new_q<} are not NE in Parts 1-i, 1-ii, and 1-iii, respectively.
	
\textbf{Part 1-i: } In this part, we prove that, item 1 of Theorem~\ref{theorem:NE_stage1_new_q<}, i.e. $p^{eq}_{NoN}=c+\kappa_u \tilde{q}_p-t_{NoN}$ and $p^{eq}_N=c$ is not an NE. We do so in Parts 1-i-a and 1-i-b, by introducing a unilateral profitable deviation for ISP NoN for the cases that $t_{NoN}$ is large and  $t_{N}$ is large, respectively. Note that in this case, by item 1 of Theorem~\ref{l-theorem:p_tilde_new}-I, $(q^{eq}_N,q^{eq}_{NoN})\in (0,\tilde{q}_f)\in F^L_1$. Thus, $n_{NoN}=1$, and the payoff of ISP NoN  is (by \eqref{l-equ:payoffISPsGeneral_new}-I, Theorem~\ref{l-theorem:NE_stage2_new}-I, and Definition \ref{l-def:pt1,pt2}-I):
\be \label{equ:theorem:larget1}
\pi_{NoN}=\kappa_u \tilde{q}_p-t_{NoN}+\kappa_{ad}(\tilde{q}_p-\tilde{q}_f)
\ee
\textbf{Part 1-i-a:} If $t_{NoN}$ is large, then \eqref{equ:theorem:larget1} would be less than zero. A deviation to price $p'_{NoN}=c$ yields a payoff of at least zero for the ISP NoN (by \eqref{l-equ:payoffISPsGeneral_new}-I). Thus, this is a profitable deviation. \\
\textbf{Part 1-i-b: } Now, consider $t_N$ to be large, and a deviation by ISP NoN such that $p'_{NoN}=\frac{1}{2}t_N$ (Thus, $\Delta p=p'_{NoN}-p^{eq}_N=\frac{1}{2}t_N-c$). Note that in this case, $\Delta p_t=\kappa_u(2\tilde{q}_p-\tilde{q}_f)-t_{NoN}<\Delta p<t_N+\kappa_u(\tilde{q}_p-\tilde{q}_f)$. Thus, by item 2-a-i of
Theorem~\ref{l-theorem:p_tilde_new}-I, $(q^{eq}_n,q^{eq}_{NoN})=(\tilde{q}_f,\tilde{q}_p)\in F^I_1$. 
Thus, by \eqref{l-equ:payoffISPsGeneral_new}-I, the payoff of ISP NoN after deviation is  at least\footnote{Note that the payoff of NoN is equal to the maximum of the payoff when $\tilde{p}^{eq}=\tilde{p}_t$ and when $\tilde{p}^{eq}>\tilde{p}_t$, i.e. when $z^{eq}=0$.} (by the definition of $\tilde{p}_{t,3}$ in Definition~\ref{l-def:pt1,pt2}-I and Theorem~\ref{l-theorem:NE_stage2_new}-I):
\be \label{equ:thoerem:larget2}
\pi'_{NoN}=\frac{1}{2}t_Nn_{NoN}+\kappa_{ad}n_{NoN}(\tilde{q}_p-\tilde{q}_f) 
\ee  
, where $n_{NoN}=\frac{\frac{1}{2}t_N+\kappa_u(\tilde{q}_p-\tilde{q}_f)+c}{t_N+t_{NoN}}$. Thus, for $t_N$ large, $n_{NoN}\rightarrow \frac{1}{2}$. Thus, comparing \eqref{equ:thoerem:larget2} with \eqref{equ:theorem:larget1} yields:
$$
\pi'_{NoN}=\frac{1}{4}t_N+\frac{1}{2}\kappa_{ad}(\tilde{q}_p-\tilde{q}_f)>\pi_{NoN} \qquad \text{since $t_N$ is large}
$$
Thus, this deviation is  profitable . 

\textbf{Part 1-ii:} In this part, we prove that item 2 of Theorem~\ref{theorem:NE_stage1_new_q<}, i.e. $p^{eq}_{NoN}=c+\frac{t_{NoN}+2t_N+\tilde{q}_{p}(\kappa_u -2\kappa_{ad})}{3}$ and $p^{eq}_{N}=c+\frac{2t_{NoN}+t_N-\tilde{q}_{p}(\kappa_u +\kappa_{ad})}{3}$ is not an NE. We do so by proving that $\Delta p^{eq}$ does not satisfy   $\kappa_u \tilde{q}_{p}-t_{NoN}<\Delta p^{eq}<\kappa_u (2\tilde{q}_{p}-\tilde{q}_f)-t_{NoN}$ and $t_N+\kappa_u(\tilde{q}_{p}-\tilde{q}_f)<\Delta p^{eq}<t_N+\kappa_u \tilde{q}_{p}$, 
in the cases that $t_{NoN}$ or $t_{N}$ is large.

First, note that:
\begin{equation}
\Delta p^{eq}=p^{eq}_{NoN}-p^{eq}_N=\frac{1}{3}(t_N-t_{NoN}+\tilde{q}_p(2\kappa_u -\kappa_{ad}))
\end{equation}
If $\Delta p^{eq}<\kappa_u(2\tilde{q}_p-\tilde{q}_f)-t_{NoN}$, then 
$t_N+2t_{NoN}<3\kappa_u (2\tilde{q}_p-\tilde{q}_f)-\tilde{q}_p(2\kappa_u -\kappa_{ad})$, which is not correct when $t_{NoN}$  or $t_N$ is large. Thus, (a) $\Delta p^{eq}\geq \kappa_u(2\tilde{q}_p-\tilde{q}_f)-t_{NoN}$. In addition, if $t_N+\kappa_u(\tilde{q}_p-\tilde{q}_f)<\Delta p^{eq}$, then $2t_N+t_{NoN}<\tilde{q}_p(2\kappa_u -\kappa_{ad})-3\kappa_u(\tilde{q}_p-\tilde{q}_f)$,  which is not correct when $t_{NoN}$  or $t_N$ is large. Thus,   (b) $\Delta p^{eq}\leq t_N+\kappa_u(\tilde{q}_p-\tilde{q}_f)$. Therefore, (a) and (b) yields that $\Delta p^{eq}$ is not in the regions specified.  Thus, item 2 cannot be an NE.

\textbf{Part 1-iii:} In this part, we prove that item 4 of Theorem~\ref{theorem:NE_stage1_new_q<}, i.e. $p^{eq}_{NoN}=c$ and $p^{eq}_N=c-\kappa_u(2\tilde{q}_{p}-\tilde{q}_f)+t_{NoN}$ is not an NE. To do so, we prove that there exists a profitable unilateral deviation for ISP NoN. Note that, in this  case, $\Delta p^{eq}=\Delta p_t$. By item 2-a-i of Theorem~\ref{l-theorem:p_tilde_new}-I, when $\Delta p_t\leq \Delta p <t_N+\kappa_u(\tilde{q}_p-\tilde{q}_f)$, then   $(q^{eq}_N,q^{eq}_{NoN})=(\tilde{q}_f,\tilde{q}_p)\in F^I_1$. Thus, the expression of the payoff of ISP NoN is (by  $\tilde{p}_t=\tilde{p}_{t,3}$, Definition \ref{l-def:pt1,pt2}-I, Theorem~\ref{l-theorem:NE_stage2_new}-I, and \eqref{l-equ:UNoN_new}-I):

\footnotesize
$$
\ba 
\pi_{NoN,C}&(p_{NoN},\tilde{p}_{t,3})\\
&=(p_{NoN}-c+\kappa_{ad}(\tilde{q}_{p}-\tilde{q}_f))(\frac{t_N+\kappa_u (\tilde{q}_{p}-\tilde{q}_f)+p_N-p_{NoN}}{t_N+t_{NoN}})
\ea 
$$  
\normalsize

Note that:

\footnotesize
$$
\frac{d \pi_{NoN,C}}{d p_{NoN}}= \frac{t_N+(\tilde{q}_{p}-\tilde{q}_f)(\kappa_u-\kappa_{ad})+p_N-2 p_{NoN}+c}{t_N+t_{NoN}}
$$
\normalsize
Thus,
\footnotesize
$$
\ba 
\frac{d \pi_{NoN,C}}{d p_{NoN}}&|_{p^{eq}_N,p^{eq}_{NoN}}\\
&=\frac{t_N+t_{NoN}+(\tilde{q}_{p}-\tilde{q}_f)(\kappa_u-\kappa_{ad})-\kappa_u(2\tilde{q}_p-\tilde{q}_f)}{t_N+t_{NoN}}
\ea
$$
\normalsize
Note that $\frac{d \pi_{NoN,C}}{d p_{NoN}}|_{p^{eq}_N,p^{eq}_{NoN}}>0$, when either $t_N$ or $t_{NoN}$ are large enough. Thus, in this case, the payoff is increasing with respect to $p_{NoN}$\footnote{Note that after this deviation, $\Delta p$ remains in the same region.}. Thus, $p'_{NoN}=p^{eq}_{NoN}+\epsilon$ for $\epsilon>0$ small, is a unilateral profitable deviation.  

\textbf{Part 2:} We now prove that when one of $t_N$ or $t_{NoN}$ is large, then strategy 3) of Theorem~\ref{theorem:NE_stage1_new_q<} is an NE. To do so, we check conditions (i), (ii), and (iii) of strategy 3) of Theorem~\ref{theorem:NE_stage1_new_q<}, in Parts 2-i, 2-ii, and 2-iii, respectively. Later, in Part 2-iv, we prove that there is no unilateral profitable deviation for ISPs.  This completes the proof. \\
\textbf{Part 2-i:} In this part, we check the condition, i.e. $\kappa_u (2\tilde{q}_{p}-\tilde{q}_f)-t_{NoN}< \Delta p^{eq}<t_N+\kappa_u (\tilde{q}_{p}-\tilde{q}_f)$. Note that in this case:
\be 
\Delta p^{eq}=\frac{1}{3}(t_N-t_{NoN}+(\tilde{q}_p-\tilde{q}_f)(2\kappa_u-\kappa_{ad}))
\ee 
Comparing the lower boundary yields that:
\footnotesize
$$
\ba 
\kappa_u &(2\tilde{q}_{p}-\tilde{q}_f)-t_{NoN}< \Delta p^{eq}\\
&\Rightarrow 2t_{NoN}+t_N+(\tilde{q}_p-\tilde{q}_f)(2\kappa_u-\kappa_{ad})-3\kappa_u (2\tilde{q}_{p}-\tilde{q}_f)>0
\ea 
$$
\normalsize
which is true when one of $t_N$ or $t_{NoN}$ is large. Now, consider the upper boundary:
\footnotesize
$$
\ba 
\Delta p^{eq}&<t_N+\kappa_u (\tilde{q}_{p}-\tilde{q}_f)\\
&\Rightarrow 2t_N+t_{NoN}+\kappa_u (\tilde{q}_{p}-\tilde{q}_f) -(\tilde{q}_p-\tilde{q}_f)(2\kappa_u-\kappa_{ad})>0
\ea 
$$
\normalsize
which is true when one of $t_N$ or $t_{NoN}$ is large. Thus, condition (i)  of strategy 3) of Theorem~\ref{theorem:NE_stage1_new_q<} is true.\\
\textbf{Part 2-ii:} Condition (ii) of this strategy is $\tilde{q}_p-\tilde{q}_f\leq \frac{2t_{NoN}+t_N}{\kappa_u+\kappa_{ad}}$. This condition holds when one of $t_N$ or $t_{NoN}$ is large. \\
\textbf{Part 2-iii:} Now, we check the third condition, i.e. $\pi^{eq}_{NoN}=\pi_{NoN}(\tilde{p}^{eq}_{NoN},\tilde{p}_{t,3})>\pi_{NoN,z=0}(\tilde{p}^{eq}_{NoN},\tilde{p})$. We use \eqref{l-equ:payoffISPsGeneral_new}-I to find $\pi^{eq}_{NoN}=\pi_{NoN}(\tilde{p}^{eq}_{NoN},\tilde{p}_{t,3})$. Note that by using item 2-a-i of Theorem~\ref{l-theorem:p_tilde_new}-I (since $z^{eq}=1$), $(q^{eq}_N,q^{eq}_{NoN})=(\tilde{q}_f,\tilde{q}_p)$. Thus, by the definition of 
$p^{eq}_{NoN}$, $\Delta p^{eq}$, $\tilde{p}_{t,3}$, and using Definition \ref{l-def:pt1,pt2}-I, Theorem~\ref{l-theorem:NE_stage2_new}-I:

\footnotesize
\be \label{equ:equ:B>_eq_piNoN_new}
\pi^{eq}_{NoN}=\frac{\big{(}t_{NoN}+2t_N+  (\tilde{q}_{p}-\tilde{q}_f)(\kappa_u +\kappa_{ad})\big{)}^2}{9(t_N+t_{NoN})}
\ee 
\normalsize

	Now, we obtain $\pi_{NoN,z=0}(\tilde{p}^{eq}_{NoN},\tilde{p})$. Consider the case that $\tilde{p}$ is such that $z^{eq}=0$. Note that since  $\kappa_u (2\tilde{q}_{p}-\tilde{q}_f)-t_{NoN}< \Delta p^{eq}<t_N+\kappa_u (\tilde{q}_{p}-\tilde{q}_f)$, then  $-t_{NoN}<\Delta p^{eq}<t_N$ or $\Delta p^{eq}\geq t_N$. Using item 2 of Theorem~\ref{l-lemma:CP_z=0_new}-I, if $\Delta p^{eq}\geq t_N$, then $n_{NoN}=0$, and by \eqref{l-equ:payoffISPsGeneral_new}-I, $\pi_{NoN,z=0}(\tilde{p}^{eq}_{NoN},\tilde{p})=0$. Thus, $\pi^{eq}_{NoN}>\pi_{NoN,z=0}(\tilde{p}^{eq}_{NoN},\tilde{p})$, and this part follows. Now, consider the case that $-t_{NoN}<\Delta p^{eq}<t_N$. Using item 1 of Theorem~\ref{l-lemma:CP_z=0_new}-I, if $-t_{NoN}<\Delta p^{eq}<t_N$, then $(q^{eq}_N,q^{eq}_{NoN})=(\tilde{q}_f,\tilde{q}_f)\in F^I_0$. Since  $(q^{eq}_N,q^{eq}_{NoN})\in F^I_0$, we can use \eqref{equ:UNoN_new}. Thus,   by using  $p^{eq}_{NoN}$, $\Delta p^{eq}$,  and , $\pi_{NoN,z=0}(\tilde{p}^{eq}_{NoN},\tilde{p})$ is:
	
	\small
	\be 
	\ba 
	&\pi_{NoN,z=0}(\tilde{p}^{eq}_{NoN},\tilde{p})\\
	&=\frac{1}{9(t_N+t_{NoN})}\big{(}2t_N+t_{NoN}+(\tilde{q}_p-\tilde{q}_f)(\kappa_u-2\kappa_{ad})\big{)}\\
	&\qquad \qquad \qquad \times\big{(}2t_N+t_{NoN}-(\tilde{q}_p-\tilde{q}_f)(2\kappa_u -\kappa_{ad})\big{)}
	\ea
	\ee 
	\normalsize
	
	 Next, we prove that $t_{NoN}+2t_N+  (\tilde{q}_{p}-\tilde{q}_f)(\kappa_u +\kappa_{ad})>2t_N+t_{NoN}+(\tilde{q}_p-\tilde{q}_f)(\kappa_u-2\kappa_{ad})$ and $t_{NoN}+2t_N+  (\tilde{q}_{p}-\tilde{q}_f)(\kappa_u +\kappa_{ad})>2t_N+t_{NoN}-(\tilde{q}_p-\tilde{q}_f)(2\kappa_u -\kappa_{ad})$. This yields  $\pi^{eq}_{NoN}>\pi_{NoN,z=0}(\tilde{p}^{eq}_{NoN},\tilde{p})$. To prove the inequalities, note that:
	 \footnotesize
	 $$
	 \ba 
	 &t_{NoN}+2t_N+  (\tilde{q}_{p}-\tilde{q}_f)(\kappa_u +\kappa_{ad})\\
	 &>2t_N+t_{NoN}+(\tilde{q}_p-\tilde{q}_f)(\kappa_u-2\kappa_{ad})\iff 3\kappa_{ad}(\tilde{q}_p-\tilde{q}_f)>0\\
	 &t_{NoN}+2t_N+  (\tilde{q}_{p}-\tilde{q}_f)(\kappa_u +\kappa_{ad})\\
	 &>2t_N+t_{NoN}-(\tilde{q}_p-\tilde{q}_f)(2\kappa_u -\kappa_{ad})\iff 3\kappa_u(\tilde{q}_p-\tilde{q}_f)>0
	 \ea 
	 $$
\normalsize

	 Since $\tilde{q}_p>\tilde{q}_f$, both inequalities hold. This completes the proof of this part. \\
	 \textbf{Part 2-iv:} In this part, we prove that there is no profitable unilateral deviation by ISPs when one of $t_N$ or $t_{NoN}$ is large. To do so, first, in Part 2-iv-NoN, we rule out the possibility of a profitable deviation by the non-neutral ISP. Then, in Part 2-iv-N, we rule out profitable deviations by the neutral ISP. 
	 
	 Note that,  by \eqref{equ:equ:B>_eq_piNoN_new}, the equilibrium payoff of ISP NoN, $\pi^{eq}_{NoN}=\pi_{NoN}(\tilde{p}^{eq}_{NoN},\tilde{p}_{t,3})$ is: 
	 $$
	  \pi^{eq}_{NoN}=\frac{\big{(}t_{NoN}+2t_N+  (\tilde{q}_{p}-\tilde{q}_f)(\kappa_u +\kappa_{ad})\big{)}^2}{9(t_N+t_{NoN})}
	 $$
	 In addition, using $(q^{eq}_N,q^{eq}_{NoN})=(\tilde{q}_f,\tilde{q}_p)$, $p^{eq}_{N}$, $\Delta p^{eq}$, and \eqref{equ:UN_new}, we can find $\pi^{eq}_{N}=\pi_{N}(\tilde{p}^{eq}_{N})$,:
\be \label{equ:equ:B>_eq_piN_new}
\pi^{eq}_{N}=\frac{\big{(}2t_{NoN}+t_N-(\tilde{q}_{p}-\tilde{q}_f)(\kappa_u +\kappa_{ad})\big{)}^2}{9(t_N+t_{NoN})}
\ee 
	Note that when $t_N$ and $t_{NoN}$ are large, $\pi^{eq}_N$ and $\pi^{eq}_{NoN}$ would be large.  
	 
	 Consider different regions in Theorem~\ref{l-theorem:p_tilde_new}-I.  We denote 
	 $\Delta p\leq \kappa_u \tilde{q}_{p}-t_{NoN}$ by region A, $\kappa_u \tilde{q}_{p}-t_{NoN}<\Delta p<\Delta p_t=\kappa_u (2\tilde{q}_{p}-\tilde{q}_f)-t_{NoN}$ by region $B_1$, $\Delta p_t=\kappa_u (2\tilde{q}_{p}-\tilde{q}_f)-t_{NoN}\leq \Delta p<t_N+\kappa_u(\tilde{q}_{p}-\tilde{q}_f)$ by region C, $t_N+\kappa_u(\tilde{q}_{p}-\tilde{q}_f)\leq \Delta p<t_N+\kappa_u \tilde{q}_{p}$ by  $B_2$, and $\Delta p\geq t_N+\kappa_u \tilde{q}_{p}$ by D. Recall that $\Delta p^{eq}=p^{eq}_{NoN}-p^{eq}_N$ is in region C. Note that the payoffs are concave in C, and we found the strategies by solving the first order condition. Thus, there is no unilateral profitable deviation in C.\\
	 \textbf{Part 2-iv-NoN:}  Now, we consider unilateral deviations by ISP NoN. We prove that any deviation to regions A, $B_1$, $B_2$, and D is not profitable in Cases 2-iv-NoN-A, 2-iv-NoN-$B_1$, 2-iv-NoN-$B_2$, and 2-iv-NoN-D, respectively.  This yields that no deviation is profitable for ISP NoN.\\
	 \textbf{Case 2-iv-NoN-A:} First, we  prove  that in Region A, $z^{eq}=1$. Note that in this case, by Definition \ref{l-def:pt}-I, $\tilde{p}_t=\tilde{p}_{t,1}$. Thus,  $\pi_{NoN}(p_{NoN},\tilde{p}_{t})=p_{NoN}-c+\tilde{q}_{p}\tilde{p}_{t,1}=p_{NoN}-c+\kappa_{ad}(\tilde{q}_p-\tilde{q}_f)$ (by \eqref{l-equ:payoffISPsGeneral_new}-I, Definition~\ref{l-def:pt1,pt2}-I, and since $n_{NoN}=1$ by item 1 of Theorem~\ref{l-theorem:p_tilde_new}-I). On the other hand, $\pi_{NoN,z=0}(p_{NoN},\tilde{p})=(p_{NoN}-c)n_{NoN}$. Thus, $\pi_{NoN}(p_{NoN},\tilde{p}_t)>\pi_{NoN,z=0}(p_{NoN},\tilde{p})$ (since $\tilde{q}_p>\tilde{q}_f$ and $0\leq n_{NoN}\leq 1$). Thus, by Theorem~\ref{l-theorem:NE_stage2_new_suff}-I, in region A, $z^{eq}=1$.

	 Now, consider  $p^{eq}_N$ fixed and decreasing $p_{NoN}$ such that $\Delta p$  in region A, i.e. $\Delta p\leq \kappa_u \tilde{q}_{p}-t_{NoN}$.  Since in Region A, $z^{eq}=1$, and by Theorem~\ref{l-theorem:NE_stage2_new}-I, the payoff after deviation is  $\pi_{NoN}'=p_{NoN}-c+\tilde{q}_p\tilde{p}_{t,1}$ (by \eqref{l-equ:payoffISPsGeneral_new}-I, Definition~\ref{l-def:pt1,pt2}-I, and since $n_{NoN}=1$ by item 1 of Theorem~\ref{l-theorem:p_tilde_new}-I).   
	  Thus, the payoff of the ISP NoN is an increasing function of $p_{NoN}$. Therefore, all other prices are dominated by $p'_{NoN}={p}^{eq}_N+\kappa_u \tilde{q}_{p}-t_{NoN}$. The payoff in this case is $\pi'_{NoN}={p}^{eq}_N+\kappa_u \tilde{q}_{p}-t_{NoN}-c+\tilde{q}_{p}\tilde{p}_{t,1}$ (by \eqref{l-equ:payoffISPsGeneral_new}-I). Therefore:
	 \be \label{equ:deviation:NoN}
	 \pi'_{NoN}=\frac{1}{3}(t_N-t_{NoN})+\alpha
	 \ee 
	where $\alpha$ is a constant independent of $t_N$ and $t_{NoN}$. Now, in Cases (i), (ii), and (iii),  
	 we prove that $\pi^{eq}_{NoN}>\pi'_{NoN}$ when  (i) $t_N$ 
	  is sufficiently larger than other parameters, (ii) $t_{NoN}$ is sufficiently larger than other parameters, and (iii) $t_N$ and $t_{NoN}$ are of the same order of magnitude and
	  both are sufficiently larger than other parameters, respectively.\\
	  \textbf{Case (i): }  If $t_N$ 
	  is sufficiently larger than other parameters, then:
	$$
	\pi^{eq}_{NoN}\approx \frac{4 t_N}{9}>\pi'_{NoN}\approx\frac{1}{3}t_N
	$$
	Thus, this deviation is not profitable. \\
	\textbf{Case (ii): } If $t_{NoN}$ is sufficiently larger than other parameters, then:
		$$
		\pi^{eq}_{NoN}\approx \frac{ t_{NoN}}{9}>\pi'_{NoN}\approx-\frac{1}{3}t_{NoN}
		$$
 Thus, this deviation is also not profitable.\\
 \textbf{Case (iii): } If $t_N$ and $t_{NoN}$ are of the same order of magnitude ($t_N\approx t_{NoN}$) and both are sufficiently larger than other parameters, then: 
 $$
 \pi^{eq}_{NoN}=\frac{t_N}{2}>\frac{t_N}{3}>\pi'_{NoN}
 $$
Thus, this deviation is not profitable.

Thus, any deviation to region A by ISP NoN is not profitable. This completes the proof of this case. \\
\textbf{Case 2-iv-NoN-$B_1$:} Now, consider a deviation by ISP NoN to region $B_1$, i.e. $\kappa_u \tilde{q}_{p}-t_{NoN}<\Delta p<\Delta p_t=\kappa_u (2\tilde{q}_{p}-\tilde{q}_f)-t_{NoN}$. Note that with this deviation, $p'_{NoN}=\frac{1}{3}(t_N-t_{NoN})+\alpha$, where $\alpha_l<\alpha<\alpha_u$, in which $\alpha_l$ and $\alpha_u$  are bounded. In addition, by \eqref{l-equ:EUs_linear}-I, after the deviation, $n'_{NoN}=\frac{t_N+t_{NoN}-\beta}{t_N+t_{NoN}}$, where $\beta>0$ is bounded ($\beta_l<\beta<\beta_u$, and $\beta_l$ and $\beta_u$  bounded ). Therefore, for large $t_N$ and $t_{NoN}$, $n'_{NoN}\rightarrow 1$. Thus, by \eqref{equ:UNoN_new}, the payoff of ISP NoN after deviation is:
$$
\pi'_{NoN}=\frac{1}{3}(t_N-t_{NoN})+\gamma
$$ 
where $\gamma$ is bounded (Note that $\tilde{p}$ is independent of $t_N$ and $t_{NoN}$). This expression is similar to  \eqref{equ:deviation:NoN}. Thus, we can exactly repeat the arguments in Cases i, ii, and iii to prove that  any deviation to region $B_1$ by ISP NoN is not profitable. This completes the proof of this case. \\
\textbf{Case 2-iv-NoN-$B_2$:} Now, consider a deviation by ISP NoN to region $B_2$, i.e. $t_N+\kappa_u(\tilde{q}_{p}-\tilde{q}_f)\leq \Delta p<t_N+\kappa_u \tilde{q}_{p}$. Note that with this deviation, $\Delta p'=t_N+\alpha$, and $p'_{NoN}=\frac{2t_{NoN}+4t_N}{3}+\gamma$
 where $\kappa_u(\tilde{q}_p-\tilde{q}_f)\leq \alpha\leq \kappa_u \tilde{q}_p$ and thus $\gamma$ is bounded. Thus,  by \eqref{l-equ:EUs_linear}-I, after this deviation, $n'_{NoN}=\frac{\beta}{t_N+t_{NoN}}$, where $\beta>0$ is a constant independent of $t_N$ and $t_{NoN}$, and the payoff of ISP NoN after deviation is $\pi'_{NoN}=\frac{2t_{NoN}+4t_N}{3(t_N+t_{NoN})}\beta +\eta$ (by \eqref{l-equ:payoffISPsGeneral_new}-I and considering that by Theorem~\ref{l-theorem:NE_stage2_new}-I, if $z^{eq}=1$, then $\tilde{p}=\tilde{p}_{t,2},$, and independent of $t_N$ and $t_{NoN}$), where $\eta$ is a constant independent of $t_N$ and $t_{NoN}$. Thus, when one of $t_N$ and $t_{NoN}$ is large, $\pi'_{NoN}\rightarrow constant$. Therefore, $\pi^{eq}_{NoN}>\pi'_{NoN}$. Thus, any deviation to region $B_2$ by ISP NoN is not profitable.\\
\textbf{Case 2-iv-NoN-D:} By item 4 of Theorem~\ref{l-theorem:p_tilde_new}-I, in region D, $n_{NoN}=0$. Thus, a deviation to this region, yields a payoff of zero, by \eqref{equ:UNoN_new} and $z^{eq}=0$. Thus, a deviation of this kind is not profitable for ISP NoN. \\
\textbf{Part 2-iv-N:} Now, consider unilateral deviations by the  neutral ISP. Similar to Part 2-iv-N, we prove that any deviation to regions $A$, $B_1$, $B_2$, and $D$ is not profitable. We do so in Cases 2-iv-N-A, 2-iv-N-$B_1$, 2-iv-N-$B_2$, and 2-iv-N-D, respectively.  \\
\textbf{Case 2-iv-N-A:} Consider a deviation by ISP N to region A. In this case, by item 1 of Theorem~\ref{l-theorem:p_tilde_new}-I, $n_N=0$. Thus, the payoff of ISP N after deviation is zero (by~\eqref{equ:UN_new}), and this deviation is  not profitable.\\
\textbf{Case 2-iv-N-$B_1$:}  Now, consider a deviation by ISP N to region $B_1$, i.e. $\kappa_u \tilde{q}_{p}-t_{NoN}<\Delta p<\Delta p_t=\kappa_u (2\tilde{q}_{p}-\tilde{q}_f)-t_{NoN}$. Note that with this deviation, $\Delta p=-t_{NoN}+\alpha$, and $p'_N=\frac{4t_{NoN}+2t_N}{3}+\gamma$, where $\kappa_u\tilde{q}_p<\alpha<\kappa_u(2\tilde{q}_p-\tilde{q}_f)$ and thus $\gamma$ is bounded. Thus, by \eqref{l-equ:EUs_linear}-I, $n'_{N}=\frac{\beta}{t_N+t_{NoN}}$, where $\beta>0$ is bounded. By \eqref{l-equ:payoffISPsGeneral_new}-I. The payoff of ISP N after deviation is $\pi_N=\frac{4t_{NoN}+2t_N}{3(t_N+t_{NoN})}\beta$ (by \eqref{l-equ:payoffISPsGeneral_new}-I). Thus, when one of $t_N$ and $t_{NoN}$ is large, $\pi'_{N}\rightarrow constant$. Thus, $\pi^{eq}_N>\pi'_N$. Therefore, any deviation to region $B_1$ by ISP N is not profitable. 

\textbf{case 2-iv-N-$B_2$:} Now, consider a deviation by ISP NoN to region $B_2$, i.e. $t_N+\kappa_u(\tilde{q}_{p}-\tilde{q}_f)\leq \Delta p<t_N+\kappa_u \tilde{q}_{p}$. Note that with this deviation, $\Delta p'=t_N+\alpha$, where $\kappa_u(\tilde{q}_p-\tilde{q}_f)\leq \alpha<\kappa_u\tilde{q}_p$. Thus, $p'_{N}=\frac{1}{3}(t_{NoN}-t_N)+\beta$, where $\beta$ is bounded.    In addition, by \eqref{l-equ:EUs_linear}-I, after the deviation, $n'_{N}=\frac{t_N+t_{NoN}-\gamma}{t_N+t_{NoN}}$, where $\gamma>0$ is bounded. Therefore, for large $t_N$ or $t_{NoN}$, $n'_{N}\rightarrow 1$. Thus, by \eqref{equ:UN_new}, the payoff of ISP N after deviation is:
\be \label{equ:deviation:N}
\pi'_{N}=\frac{1}{3}(t_{NoN}-t_{N})+\eta
\ee  
where $\eta$ is bounded.  Now, in Cases i, ii, and iii,  
we prove that $\pi^{eq}_{N}>\pi'_{N}$ when  (i) $t_N$ 
is sufficiently larger than other parameters, (ii) $t_{NoN}$ is sufficiently larger than other parameters, and (iii) $t_N$ and $t_{NoN}$ are of the same order of magnitude and
both are sufficiently larger than other parameters, respectively.\\
\textbf{Case i: }  If $t_N$ 
is sufficiently larger than other parameters, then:
$$
\pi^{eq}_{N}\approx \frac{ t_N}{9}>\pi'_{N}\approx-\frac{1}{3}t_N
$$
Thus, this deviation is not profitable. \\
\textbf{Case ii: } If $t_{NoN}$ is sufficiently larger than other parameters, then:
$$
\pi^{eq}_{N}\approx \frac{4 t_{NoN}}{9}>\pi'_{N}\approx\frac{1}{3}t_{NoN}
$$
Thus, this deviation is also not profitable.\\
\textbf{Case iii: } If $t_N$ and $t_{NoN}$ are of the same order of magnitude ($t_N\approx t_{NoN}$) and both are sufficiently larger than other parameters, then: 
$$
\pi^{eq}_{N}=\frac{t_{NoN}}{2}>\frac{t_{NoN}}{3}>\pi'_{N}
$$
Thus, this deviation is not profitable. 

Thus, any deviation to Region $B_2$ by ISP N is not profitable. This completes the proof of this case.\\
\textbf{Case 2-iv-N-D:}  Now, consider  decreasing $p_{NoN}$ such that $\Delta p$  in region D, i.e. $\Delta p\geq \kappa_u \tilde{q}_{p}+t_{N}$.  Note that by item 4 of Theorem~\ref{l-theorem:p_tilde_new}-I, $z^{eq}=0$, and $n_N=1$. Thus, the payoff of ISP N is equal to $p_{N}-c$ (by \eqref{l-equ:payoffISPsGeneral_new}-I). Thus, the payoff of the ISP N is an increasing function of $p_{N}$. Therefore, all other prices are dominated by $p'_{N}={p}^{eq}_{NoN}-(\kappa_u \tilde{q}_{p}+t_{N})$. Thus, the payoff in this case is $\pi'_{N}=\frac{1}{3}(t_{NoN}-t_N)+\alpha$, where $\alpha$ is a constant and is independent of $t_N$ and $t_{NoN}$. This expression is similar to \eqref{equ:deviation:N}. Thus, we can exactly repeat the arguments in Cases 2-iv-N-$B_2$-a, 2-iv-N-$B_2$-b, and 2-iv-N-$B_2$-c to prove that  any deviation to region $D$ by ISP NoN is not profitable. This completes the proof of this case. This completes the proof of this case, and the theorem.

\end{proof}

\section{Proof of Theorem~\ref{theorem:neutralz=0_q<} and Corollary~\ref{corollary:outcome_z=0}}\label{appendix:theorem:neutralz=0_q<}

First, we prove Theorem~\ref{theorem:neutralz=0_q<}:
\begin{proof}
	We consider different regions of $\Delta p$ in Theorem~\ref{l-lemma:CP_z=0_new}-I and Theorem~\ref{l-theorem:p_tilde_new}-I. For each region, we characterize all possible NE strategies. 
	
	First, consider $\Delta p\leq \kappa_u \tilde{q}_p-t_{NoN}$. Note that in this region,  the payoff of non-neutral ISP if $z^{eq}=0$ is at most $p^{eq}_{NoN}-c$ (by~\eqref{l-equ:payoffISPsGeneral_new}-I). On the other hand, by Theorem \ref{l-theorem:p_tilde_new}-I, by choosing $\tilde{p}'=\tilde{p}_{t,1}$, ISP NoN can ensure that the CP chooses $z^{eq}=1$. In this case, the payoff of non-neutral ISP (by~\eqref{l-equ:payoffISPsGeneral_new}-I) is $p'_{NoN}-c+\tilde{p}_{t,1}\tilde{q}_{NoN}=p^{eq}_{NoN}-c+\kappa_{ad}(\tilde{q}_{p}-\tilde{q}_f)>p^{eq}_{NoN}-c$. Thus, $\pi_{NoN}(p^{eq}_{NoN},\tilde{p}_{t,1})>\pi_{NoN,z=0}(p^{eq}_{NoN},\tilde{p})$. Therefore, in this case, there is no NE by which $z^{eq}=0$.
	
	Now, consider $\Delta p>\kappa_u \tilde{q}_p-t_{NoN}$. Note that $\tilde{q}_p<\frac{t_N+t_{NoN}}{\kappa_u}$. Thus, two possibility may arise: (i) $-t_{NoN}<\Delta p<t_N$, and (ii) $\Delta p\geq t_N$. We consider these two cases in Case 1 and 2, respectively. 
	
	\textbf{Case 1:} In this case, $-t_{NoN}<\Delta p<t_N$. By item 1 of Theorem \ref{l-lemma:CP_z=0_new}-I, $(q^{eq}_N,q^{eq}_{NoN})=(\tilde{q}_f,\tilde{q}_f)\in F^I_0$. Note that in this region, $0<x_N<1$, and an NE strategy for ISPs  should satisfy the first order optimality conditions. Thus, using \eqref{equ:UN_new} and \eqref{equ:UNoN_new}:
	
	$$
	\small
	\pi_N(p_N)=(p_N-c)\frac{t_{NoN}+p_{NoN}-p_N}{t_N+t_{NoN}} 
	$$
	
	$$
	\ba 
	\pi_{NoN}(p_{NoN},\tilde{p})&=(p_{NoN}-c)\frac{t_N+p_N-p_{NoN}}{t_N+t_{NoN}}
	\ea
	$$
	\normalsize
	Solving the first order optimality condition yields:

	\be\label{equ:ISPcandidateMultihome_z=0}
	\ba
	p^{eq}_N&= c+\frac{1}{3}(2t_{NoN}+t_N)\\
	p^{eq}_{NoN}&=c+\frac{1}{3}(2t_N+t_{NoN})
	\ea 
	\ee
	which is unique. Note that $p^{eq}_N\geq c$  and $p^{eq}_{NoN}\geq c$. First, note that $-t_{NoN}< \Delta p^{eq}=p^{eq}_{NoN}-p^{eq}_N=\frac{t_N-t_{NoN}}{3}< t_N$. 
	
	The necessary condition for this strategy to be an NE is $\pi_{NoN,z=0}(p^{eq}_{NoN})\geq \pi_{NoN}(p^{eq}_{NoN},\tilde{p}_t)$ (by Theorem~\ref{l-theorem:NE_stage2_new_suff}-I). The candidate strategies and this necessary  condition is listed in the statement of the theorem.  
	
	\textbf{Case 2:} Now, consider $\Delta p\geq t_{N}$. We consider two cases $\Delta p=t_N$ and $\Delta p>t_N$ in Cases 2-i and 2-ii, respectively. \\
	\textbf{Case 2-i:} Now consider strategies $p_{NoN}$ and $p_N$ such that $\Delta p=t_{N}$. In this case, using case 2 of Theorem~\ref{l-lemma:CP_z=0_new}-I, $(q^{eq}_N,q^{eq}_{NoN})=(\tilde{q}_f,0)\in F_0^U$. Thus, $n_{NoN}=0$ and $\pi_{NoN}(p_{NoN},z=0)=0$, i.e. the payoff of the non-neutral ISP is zero.  Consider $\epsilon>0$ such that $p'_{NoN}=p_{NoN}-\epsilon> c$. In this case, $p'_{NoN}-p_N<t_{N}$. Thus, by  Theorem~\ref{l-lemma:CP_z=0_new}-I,  $(q^{eq}_N,q^{eq}_{NoN})\in F^I_0$ or $(q^{eq}_N,q^{eq}_{NoN})\in F^L_0$. Thus, $n_{NoN}>0$, and $\pi_{NoN}(p'_{NoN},z=0)>0$. Thus, $p'_{NoN}$ is a profitable deviation for the non-neutral ISP. Therefore, as long as such a deviation exist  $p_{NoN}$ and $p_N$ such that $\Delta p=t_{N}$ cannot be NE.\\
\textbf{Case 2-ii:}  Now, consider $\Delta p> t_{N}$. Thus, by item 2 of Theorem~\ref{l-lemma:CP_z=0_new}-I,  $n^{eq}_N=1$.   Consider a unilateral deviation by neutral ISP such that $p'_N=p^{eq}_N+\epsilon$ in which $\epsilon>0$ such that $p^{eq}_{NoN}-p'_N>t_N$. Note that the values of $q^{eq}_N$ and $q^{eq}_{NoN}$ is the same as before, since still $\Delta p'=p^{eq}_{NoN}-p'_N>t_N$. Thus, again $n^{eq}_N=1$, and by \eqref{l-equ:payoffISPsGeneral_new}-I, the payoff of neutral ISP is an increasing function of $p_N$. Thus, $p'_N$ is a profitable unilateral deviation. This contradicts the assumption that $p^{eq}_N$ and $p^{eq}_{NoN}$ is NE. Thus, the result of the theorem follows.
\end{proof}

Now, we prove Corollary~\ref{corollary:outcome_z=0}:

\begin{proof}
	Note that we constructed this strategy such that $\Delta p$ satisfies item  1 of Theorem~\ref{l-lemma:CP_z=0_new}-I. Thus, $(q^{eq}_N,q^{eq}_{NoN})=(\tilde{q}_f,\tilde{q}_f)\in F^I_0$. Using the expression for $\Delta p=p^{eq}_{NoN}-p^{eq}_N$, and \eqref{l-equ:EUs_linear}-I, the expressions for $n^{eq}_N$ and $n^{eq}_{NoN}$ follow.
\end{proof}

\section{Proof of Theorem~\ref{lemma:NEz=0}}\label{appendix:lemma:NEz=0}

 The following lemmas allow us to characterize the NE when $(q^{eq}_N,q^{eq}_{NoN})\in F_0$, i.e. $z^{eq}=0$. 
 
  Lemmas \ref{lemma:pnonDominated} and \ref{lemma:pnDominated} are useful in proving Theorem~\ref{lemma:NEz=0}.
 \begin{lemma}\label{lemma:pnonDominated}
 	No $p_{NoN}$ and $p_N$ such that $\Delta p=p_{NoN}-p_N\leq -t_{NoN}$ can be equilibrium strategies.
 \end{lemma}
 
 \begin{proof}
 	First, we rule out the existence of an NE when  $\Delta p<-t_{NoN}$, and then  when  $\Delta p=-t_{NoN}$.
 	
 	First, consider $p_{NoN}$ and $p_N$ such that $\Delta p<-t_{NoN}$. In this case, $p_{NoN}<p_N-t_{NoN}$. Note that the payoff of the non-neutral ISP when $\Delta p\leq -t_{NoN}$ is $p_{NoN}-c$ (by \eqref{l-equ:payoffISPsGeneral_new}-I and $n_{NoN}=1$, using case 3 of Theorem~\ref{l-lemma:CP_z=0_new}-I), and is strictly increasing with respect to $p_{NoN}$. Thus, every price $p_{NoN}<p_N-t_{NoN}$ yields a strictly lower payoff for the non-neutral ISP in comparison with the payoff of the this ISP when $p_{NoN}=p_N-t_{NoN}$. Thus, there exist a profitable deviation for the non-neutral ISP for strategies such that $p_{NoN}-p_N<-t_{NoN}$. Therefore, no $p_{NoN}$ and $p_N$ such that $p_{NoN}-p_N<-t_{NoN}$ can be NE strategies. 
 	
 	Now consider strategies $p_{NoN}$ and $p_N$ such that $\Delta p=-t_{NoN}$. In this case, using case 3 of Theorem~\ref{l-lemma:CP_z=0_new}-I, $(q^{eq}_N,q^{eq}_{NoN})=(0,\tilde{q}_f)\in F_0^L$. Thus, $n_N=0$ and $\pi_N(p_N)=0$, i.e. the payoff of the neutral ISP is zero.  Consider $\epsilon>0$ such that $p'_N=p_N-\epsilon> c$. In this case, $p_{NoN}-p'_N>-t_{NoN}$. Thus, by Theorem \ref{l-lemma:CP_z=0_new}-I, $(q^{eq}_N,q^{eq}_{NoN})\in F^I_0$ or $(q^{eq}_N,q^{eq}_{NoN})\in F_0^U$. Thus, $n_N>0$, and $\pi_N(p'_N)>0$. Thus, $p'_N$ is a profitable deviation for the neutral ISP. Therefore, as long as such a deviation exist  $p_{NoN}$ and $p_N$ such that $\Delta p=-t_{NoN}$ cannot be NE. Now, we prove that such deviation always exist. This complete the proof. Note that this deviation does not exist if and only if $p_N-\epsilon\leq c$ for all $\epsilon>0$. Therefore, this deviation does not exist if only if $p_N\leq c$. Thus, $p_{NoN}\leq c-t_{NoN}<c$, which contradicts the fact that if $z=0$, $p^{eq}_{NoN}\geq c$ (as mentioned in the beginning of the section). The lemma follows. 
 \end{proof}
 
 \begin{lemma}\label{lemma:pnDominated}
 	No $p_{NoN}$ and $p_{N}$ such that $\Delta p\geq t_{N}$ can be equilibrium strategies.
 \end{lemma}
 
 \begin{proof}
 	First, we rule out the existence of an NE when  $\Delta p>t_{N}$, and then  when  $\Delta p=t_{N}$.
 	
 	First, consider $p_{NoN}$ and $p_N$ such that $\Delta p>t_N$. In this case, $p_{N}<p_{NoN}-t_N$. Note that the payoff of the neutral ISP when $\Delta p\geq t_{N}$ is $p_{N}-c$ (by \eqref{l-equ:payoffISPsGeneral_new}-I and $n_N=1$, using case 2 of Theorem~\ref{l-lemma:CP_z=0_new}-I), and is strictly increasing with respect to $p_{N}$. Thus, every price $p_{N}<p_{NoN}-t_{N}$ yields a strictly lower payoff for the neutral ISP in comparison with the payoff of the this ISP when $p_{N}=p_{NoN}-t_{N}$. Thus, no $p_{NoN}$ and $p_N$ such that $p_{NoN}-p_N>t_N$ can be Ne strategies.
 	
 	Now consider strategies $p_{NoN}$ and $p_N$ such that $\Delta p=t_{N}$. In this case, using case 2 of Theorem~\ref{l-lemma:CP_z=0_new}-I, $(q^{eq}_N,q^{eq}_{NoN})=(\tilde{q}_f,0)\in F_0^U$. Thus, $n_{NoN}=0$ and $\pi_{NoN}(p_{NoN},z=0)=0$, i.e. the payoff of the non-neutral ISP is zero.  Consider $\epsilon>0$ such that $p'_{NoN}=p_{NoN}-\epsilon> c$. In this case, $p'_{NoN}-p_N<t_{N}$. Thus, by  Theorem~\ref{l-lemma:CP_z=0_new}-I,  $(q^{eq}_N,q^{eq}_{NoN})\in F^I_0$ or $(q^{eq}_N,q^{eq}_{NoN})\in F^L_0$. Thus, $n_{NoN}>0$, and $\pi_{NoN}(p'_{NoN},z=0)>0$. Thus, $p'_{NoN}$ is a profitable deviation for the non-neutral ISP. Therefore, as long as such a deviation exist  $p_{NoN}$ and $p_N$ such that $\Delta p=t_{N}$ cannot be NE. Now, we prove that such deviation always exist. This complete the proof. Note that this deviation does not exist if and only if $p_{NoN}-\epsilon\leq c$ for all $\epsilon>0$. Therefore, this deviation does not exist if only if $p_{NoN}\leq c$. Therefore, $p_{N}\leq c-t_{N}<c$, which contradicts the fact that $p^{eq}_{N}\geq c$ (as mentioned after at the beginning of the section.). The lemma follows. 
 \end{proof}

 Now, we proceed to prove Theorem~\ref{lemma:NEz=0}:
 
  \begin{proof}
 	First, in Part 1, we characterize the candidate equilibrium strategies by applying the first order condition on the payoffs. Then, in Part 2, we prove that no unilateral deviation is profitable for ISPs. Thus, the strategies characterized in Part 1 are NE. 
 	
 	\textbf{Part 1:} Note that $z^{eq}=0$. First note that by Lemmas~\ref{lemma:pnonDominated} and \ref{lemma:pnDominated}, no $p_N$ and $p_{NoN}$ such that $\Delta p\leq -t_{NoN}$ or $\Delta p\geq t_N$ can be Nash equilibrium. Thus, we consider $-t_{NoN}< \Delta p< t_N$. Note that in this region, $0<x_N<1$, and an NE strategy for ISPs  should satisfy the first order optimality conditions. Thus, using \eqref{equ:UN_new} and \eqref{equ:UNoN_new}, and item 1 of Theorem~\ref{l-lemma:CP_z=0_new}-I:
 	\be\label{equ:ISPcandidateMultihome}
 	\ba
 	p^{eq}_N&= c+\frac{1}{3}(2t_{NoN}+t_N)\\
 	p^{eq}_{NoN}&=c+\frac{1}{3}(2t_N+t_{NoN})
 	\ea 
 	\ee
 	which is unique. Note that $p^{eq}_N\geq c$  and $p^{eq}_{NoN}\geq c$. In order to prove that this is an NE, it is enough to prove that (i) $-t_{NoN}< \Delta p^{eq}=p^{eq}_{NoN}-p^{eq}_N< t_N$, (ii) a deviation of one of the ISPs by which $\Delta p$ is shifted to the region $\Delta p\leq -t_{NoN}$ or $\Delta p\geq t_N$ is not profitable for that ISP.
 	
 	The condition (i) can be proved by \eqref{equ:ISPcandidateMultihome}. From this equation, $\Delta p^{eq}=\frac{t_N-t_{NoN}}{3}$. Thus, $\Delta p^{eq}>-t_{NoN}$ and $\Delta p^{eq}< t_N$. Therefore, (i) is true for this case. 
 	
 	\textbf{Part 2:} Now, we should prove that condition (ii) holds, i.e. no unilateral deviation is profitable. First, in Case 2-a, we rule out the possibility of a unilateral deviation when  $-t_{NoN}<\Delta p<t_N$ for both neutral and non-neutral ISPs. Then, we consider  $\Delta p\leq -t_{NoN}$ and $\Delta p\geq t_N$, and in Cases  2-NoN and 2-N, we rule out the possibility of a unilateral deviation in these regions for ISP N and NoN, respectively.
 	
 	\textbf{Case 2-a:} First, note that  by concavity of the payoffs (using \eqref{equ:UN_new} and \eqref{equ:UNoN_new}) as long as $-t_{NoN}<\Delta p<t_N$, i.e. $0<x_N<1$, a unilateral deviation by one of the ISPs from $p^{eq}_N$ or $p^{eq}_{NoN}$  decreases this ISP's payoff. Thus, we should consider the deviation by ISPs by which $\Delta p\leq -t_{NoN}$ or $\Delta p\geq t_N$.
 	
 	\textbf{Case 2-NoN:} Now, consider the deviations by the non-neutral ISP. Fix $p_N=p^{eq}_N$, and consider two cases. In Case 2-NoN-i (respectively, Case 2-NoN-ii), we consider deviation by ISP NoN such that $\Delta p\geq t_N$ (respectively, $\Delta p\leq -t_{NoN}$). 
 	
 	\textbf{Case 2-NoN-i:} Suppose the non-neutral ISP increases her price from ${p}^{eq}_{NoN}$ to make $\Delta p\geq t_N$. In this case, $n_{NoN}=0$, and the payoff of the ISP is zero (by \eqref{l-equ:payoffISPsGeneral_new}-I). Since in the candidate equilibrium strategy this payoff is non-negative\footnote{Note that $p^{eq}_{NoN}>c$ and $0\leq n_{NoN}\leq 1$.}, this deviation is not profitable.

 	\textbf{Case 2-NoN-ii:}  Now, consider the case in which the non-neutral ISP decreases her price to make $\Delta p\leq -t_{NoN}$. In this case, $n_{NoN}=1$ and $\pi_{NoN}(p'_{NoN},z=0)=p'_{NoN}-c$ (by \eqref{l-equ:payoffISPsGeneral_new}-I). Thus, the payoff is a strictly  increasing function of $p'_{NoN}$, and is  maximized at $p'_{NoN}=p^{eq}_{N}-t_{NoN}$. We show that $\pi_{NoN}(p'_{NoN},z=0)<\pi_{NoN}(p^{eq}_{NoN},z=0)$ \footnote{note that $p_N=p^{eq}_N$ is fixed.}. Note that $\pi_{NoN}(p'_{NoN},z=0)=\frac{1}{3}(t_N-t_{NoN})$. In addition, using \eqref{equ:ISPcandidateMultihome}, \eqref{l-equ:payoffISPsGeneral_new}-I, $0\leq x_{N}\leq 1$, \eqref{l-equ:EUs_linear}-I, and the fact that with $p^{eq}_N$ and $p^{eq}_{NoN}$, $q^{eq}_{NoN}-q^{eq}_{N}=0$:
 	$$
 	\pi_{NoN}(p^{eq}_{NoN},z=0)=\frac{1}{9}\frac{(2t_N+t_{NoN})^2}{t_{NoN}+t_N}
 	$$ 
 	Thus:
 	 $$
 	\ba
 	&\pi_{NoN}(p'_{NoN},z=0)< \pi_{NoN}(p^{eq}_{NoN},z=0)\\ &\qquad \iff 3(t^2_N-t^2_{NoN})<4t^2_N+t^2_{NoN}+4 t_N t_{NoN}\\
 	& \qquad \iff t^2_N+4 t^2_{NoN}+4t_Nt_{NoN}>0
 	\ea
 	$$
 	where the last inequality is always true. Thus, this deviation is not profitable for ISP NoN.

 	These cases prove that no deviation form \eqref{equ:ISPcandidateMultihome} is profitable for ISP NoN.
 	
 	\textbf{Case 2-N:}  Now, consider a deviation by the neutral ISP from \eqref{equ:ISPcandidateMultihome}. 
 	Similar argument can be done for the neutral ISP. Fix, $p_{NoN}=p^{eq}_{NoN}$, and consider two cases.  In Case 2-N-i (respectively, Case 2-N-ii), we consider deviation by ISP N such that $\Delta p\leq -t_{NoN}$ (respectively, $\Delta p\geq t_{N}$). 
 	
 	\textbf{Case 2-N-i:} Suppose the neutral ISP increases her price from ${p}^{eq}_{N}$ to get $\Delta p\leq -t_{NoN}$. In this case, $n_{N}=0$, and the payoff of this ISP is zero. Since in the candidate equilibrium strategy the payoff is non-negative\footnote{Note that $p^{eq}_{N}>c$ and $0\leq n_{N}\leq 1$.}, this deviation is not profitable.
 	
 	\textbf{Case 2-N-ii:} Now, consider the case in which  the non-neutral ISP decreases her price such that $\Delta p\geq t_{N}$. In this case, $n_{N}=1$ and $\pi_{N}(p'_{N})=p'_{N}-c$. Thus, the payoff is a strictly  increasing function of $p'_{N}$, and is  maximized at $p'_{N}=p^{eq}_{NoN}-t_{N}$. We show that $\pi_{N}(p'_{N})<\pi_{N}(p^{eq}_{N})$. Note that $\pi_{N}(p'_{N})=\frac{1}{3}(t_{NoN}-t_N)$ (by \eqref{l-equ:payoffISPsGeneral_new}-I). In addition, using \eqref{equ:ISPcandidateMultihome}, \eqref{l-equ:payoffISPsGeneral_new}-I, $0\leq x_{N}\leq 1$, \eqref{l-equ:EUs_linear}-I, and the fact that with $p^{eq}_N$ and $p^{eq}_{NoN}$, $q^{eq}_{NoN}-q^{eq}_{N}=0$:
 	$$
 	\pi_{N}(p^{eq}_{N})=\frac{1}{9}\frac{(2t_{NoN}+t_{N})^2}{t_{NoN}+t_N}
 	$$ 
 	Thus:
 	$$
 	\ba
 	\pi_{N}(p'_{N})&< \pi_{N}(p^{eq}_{N}) \\
 	&\iff 3(t^2_{NoN}-t^2_{N})<4t^2_{NoN}+t^2_{N}+4 t_N t_{NoN}\\
 	& \iff t^2_{NoN}+4 t^2_{N}+4t_Nt_{NoN}>0
 	\ea
 	$$
 	where the last inequality is always true. Thus, this deviation is not profitable for ISP N.
 	
 	Thus, there is no profitable deviation for ISP N. This completes the proof, and the lemma follows.
 \end{proof}

\section{Continuous Strategy Set for the CP}\label{section:general}

In this section, we consider $q_N\in [0,\tilde{q}_f]$ and $q_{NoN}\in[0,\tilde{q}_p]$. In this case, the CP pays a side payment of $\tilde{p}q_{NoN}$ if she chooses $q_{NoN}\in (\tilde{q}_f,\tilde{q}_p]$. The rest of the model is the same as before. Note that in this case, the optimum strategies in Stage 4 of the game, in which end-users decide on the ISP, is the same as before. \emph{We prove that the optimum decisions made by the CP is similar to the decisions of the CP when she has a discrete set of strategies (explained in the paper)}. This yields that the results of the model would the same as before when the CP chooses her strategy from a continuous set.     

Therefore, we focus on characterizing the optimum strategies of the CP when she chooses her strategy from  continuous sets, i.e.  $q_N\in [0,\tilde{q}_f]$ and $q_{NoN}\in[0,\tilde{q}_p]$. The following lemma is useful in defining the maximization and to characterize the optimum answers.

\begin{lemma}\label{lemma:qfree}
	$\pi_{CP}(q_N,\tilde{q}_{f,NoN},z=0)\geq \pi_{CP}(q_N,\tilde{q}_{f,NoN},z=1)$.
\end{lemma}

\begin{remark}
	Note that although we considered $z$ to be a dummy variable, in this lemma and for the purpose of analysis, we treat it as an independent variable.    
\end{remark}

\begin{proof}
	The lemma follows by \eqref{l-equ:payoffCP_new}-I, and comparing the expressions in these two cases:
	
	\footnotesize
	$$
	\ba 
	\pi_{CP}(q_N,\tilde{q}_{f,NoN},z=0)-\pi_{CP}(q_N,\tilde{q}_{f,NoN},z=1)&= \tilde{q}_{f,NoN} \tilde{p} \geq 0
	\ea
	$$
	
	\normalsize
	Note that we used the fact that  from \eqref{l-equ:EUs_linear}-I, since the qualities are the same in  both cases, $n_N$ and $n_{NoN}$ are equal for both cases.
\end{proof}

Lemma~\ref{lemma:qfree} provides the ground to formally define the maximization for the CP as:
\begin{equation}\label{equ:CPopt_initial2}
\begin{aligned}
\max_{z,q_N,q_{NoN}}&\pi_{CP}(q_N,q_{NoN},z)=\\
&\max_{z,q_N,q_{NoN}} \l n_{N}\kappa_{ad}q_N+n_{NoN}\kappa_{ad}q_{NoN}-z\tilde{p}q_{NoN}\r\\
&\text{s.t:}\\
&\qquad q_N\leq \tilde{q}_f\\
&\qquad \mbox{if } z=1 \quad   \tilde{q}_{f} < q_{NoN} \leq \tilde{q}_{p}\\
&\qquad \mbox{if } z=0 \quad   q_{NoN} \leq \tilde{q}_{f}
\end{aligned}
\end{equation}

\textit{Existence of the maximum:} Note that the mixed integer programming \eqref{equ:CPopt_initial2} can be written as two convex maximizations, one for $z=0$ and one for $z=1$. In addition, note that for the case $z=1$, the feasible set is not closed (since  $ \tilde{q}_{f} < q_{NoN} \leq \tilde{q}_{p}$). Thus, in this case, we should use the ``supremum" instead of ``maximum". However, using Lemma~\ref{lemma:qfree}, we prove that the maximum of  \eqref{equ:CPopt_initial2} exists, and therefore the term maximum can be used safely. To prove this, consider the closure of the feasible set when $z=1$ formed by adding $\tilde{q}_{f}$ to the set, i.e. $\tilde{F}_1$. Since the feasible set associated with $z=0$ ($F_0$) and the closure of the feasible set associated to $z=1$ ($\tilde{F}_1$) are closed and bounded (compact) and the objective function is continuous for each $z\in \{0,1\}$, using Weierstrass Extreme Value Theorem, we can say that a maximum exists in each of these two  sets and for the overall optimization~\eqref{equ:CPopt_initial2}. If the maxima in $\tilde{F}_1$ is not $\tilde{q}_{f}$, then the maxima is in the original feasible set ($F_1$). Therefore the maximum of \eqref{equ:CPopt_initial2} exists. If not and $\tilde{q}_{f}$ is the maxima in the set $\tilde{F}_1$, then by Lemma~\ref{lemma:qfree}, the maximum in the set $F_0$ dominates the maximum of the set $\tilde{F}_1$. Thus, the  maxima of \eqref{equ:CPopt_initial2} is in $F_0$. Therefore, the maximum of \eqref{equ:CPopt_initial2} exists, and we can use the term maximum safely.  

Henceforth, the solution $(\tilde{q}^*_N,\tilde{q}^*_{NoN},z^*)$ of the maximization~\eqref{equ:CPopt_initial2} would be called the optimum strategies of the CP. This solution yields $x^*_N$ and subsequently $n^*_N$ and $n^*_{NoN}$  by \eqref{l-equ:EUs_linear}-I.  In addition, we denote the feasible set of \eqref{equ:CPopt_initial2} by $\mathcal{F}$. 

\textit{Finding the optimum strategies of the CP:}  To characterize the optimum strategies, we use the partition the feasible set in Table~\ref{l-table:subsets}-I, and characterize the candidate optimum strategies, i.e. the strategies that yield a higher payoff than the rest of the feasible solutions, in each subset. The overall optimum, which is chosen by the CP, is the one that yields the highest payoff among candidate strategies.

Note that although the maximum of the overall optimization exist, a maximum may not necessarily exist in each of the subsets. We will show in the next set of lemmas that the optimization in each subset of the feasible set can be reduced to a convex maximization over linear constraints. Thus, only the extreme points of the feasible sets may constitute the optimum solution. This means that the CP chooses her strategy among the discrete strategies, $q_N\in \{0,\tilde{q}_f\}$ and $q_{NoN}\in \{0,\tilde{q}_f,\tilde{q}_p\}$.

We now characterize optimum strategies of the CP, by considering each of the sub-feasible sets and characterizing the optimum solutions in each of them. In Lemma~\ref{lemma:optimumCPqs}, we prove that if $(q^*_N,q^*_{NoN})\in F^I$, then $q^*_N\in \{0,\tilde{q}_f\}$, $q^*_{NoN}\in\{0,\tilde{q}_{f},\tilde{q}_{p}\}$, and $(q^*_N,q^*_{NoN})\neq (0,0)$. In Lemma~\ref{lemma:xn<0}, we prove that if $(q^*_N,q^*_{NoN})\in F^L$, then ${q}^*_{NoN}=\tilde{q}_{f}$, if $q^*_{NoN}\in F^L_0$, and  ${q}^*_{NoN}=\tilde{q}_{p}$, if $q^*_{NoN}\in F^L_1$.  Moreover, $0\leq q^*_N\leq \frac{1}{\kappa_u}(\kappa_u {q}^*_{NoN}-t_{NoN}-\Delta p)$, and $\Delta p \leq \kappa_u {q}^*_{NoN}-t_{NoN}$. In Lemma~\ref{lemma:xn>1}, we prove that if $(q^*_N,q^*_{NoN})\in F^U$, then ${q}^*_{N}=\tilde{q}_f$ and $0\leq {q}^*_{NoN}\leq \frac{1}{\kappa_u}(\kappa_u \tilde{q}_f-t_{N}+\Delta p)$ and $\Delta p \geq t_N -\kappa_u \tilde{q}_f$. In addition, Lemmas \ref{lemma:tildep_condition} and  \ref{lemma:pos} provide some results that are useful in proving Lemmas~\ref{lemma:optimumCPqs}-\ref{lemma:xn>1}.  

\begin{lemma}\label{lemma:tildep_condition}
	In an optimum solution of \eqref{equ:CPopt_initial2}, $n_{NoN}\kappa_{ad}-z\tilde{p}\geq 0$. 
\end{lemma}
\begin{proof}
	Suppose there exists an optimum answer such that $n_{NoN}\kappa_{ad}-z\tilde{p}< 0$. Note that $0\leq n_N,n_{NoN}\leq 1$ and qualities are non-negative. Thus, in this case, $\pi_{CP} <\kappa_{ad} q_N$. However,  choosing $z=0$ and $q_{NoN}=q_N$, yields a profit equal to $\kappa_{ad} q_N$. This contradicts the solution with $n_{NoN}\kappa_{ad}-z\tilde{p}< 0$ to be optimum. Thus,  the Lemma follows.
\end{proof}


\begin{lemma}\label{lemma:pos}
	In an optimum solution, the CP offers the content quality equal to one of the threshold at least on one ISP, i.e. $q^*_{N}=\tilde{q}_f$ OR $(q^*_{NoN}=\tilde{q}_{p} \text{ XOR } q^*_{NoN}=\tilde{q}_{f} )$, where XOR means only one the qualities is chosen.
\end{lemma}

\begin{proof}
	Suppose not. Let the optimum qualities to be $\hat{q}_{NoN}<\tilde{q}_{f}$ if $z=0$, or $\tilde{q}_{f} <\hat{q}_{NoN}<\tilde{q}_{p}$ if $z=1$, and  $\hat{q}_N<\tilde{q}_f$. The difference between the qualities offered in two platforms is $\Delta q=\hat{q}_{NoN}-\hat{q}_N$. Consider $q'_{NoN}=\hat{q}_{NoN}+\epsilon$ and $q'_{N}=\hat{q}_{N}+\epsilon$ in which $\epsilon>0$ and is such that  ${q}'_{NoN}\leq \tilde{q}_{f}$ if $z=0$, or $\tilde{q}_{f} \leq {q}'_{NoN}\leq \tilde{q}_{p}$ if $z=1$, and  ${q}'_N\leq \tilde{q}_f$. Note that $z$ remains fixed and ${q}'_{NoN}-{q}'_N=\hat{q}_{NoN}-\hat{q}_N=\Delta q$. Since $\Delta q$ is the same for two cases, the number of subscriber to each ISP is the same for both cases by \eqref{l-equ:EUs_linear}-I. Lemma~\ref{lemma:tildep_condition}, \eqref{l-equ:payoffCP_new}-I, and the fact that $n_N,n_{NoN}\geq 0$ yield that $\pi'_{CP}\geq \hat{\pi}_{CP}$, where $\hat{\pi}_{CP}$ (, respectively  $\pi'_{CP}$) is the payoff of the CP when the vector of qualities is $(\hat{q}_N,\hat{q}_{NoN})$ (, respectively, $({q}'_N,{q}'_{NoN})$). 
	
	We now prove if $(\hat{q}_N,\hat{q}_{NoN})$ is the optimum solution, then the inequality is strict, i.e. $\pi'_{CP}> \hat{\pi}_{CP}$. Suppose not, and  $\pi'_{CP}= \hat{\pi}_{CP}$. This only happens if $n_{NoN}\kappa_{ad}-z\tilde{p}=0$ and $n_N=0$.  Note that in this case, $\pi'_{CP}= \hat{\pi}_{CP}=0$. However, in the previous paragraph, we argued that with $q_N=\tilde{q}_f$ and $q_{NoN}=\tilde{q}_f$, the CP can get a payoff of $\kappa_{ad}\tilde{q}_f>0$. This contradicts the assumption that $(\hat{q}_N,\hat{q}_{NoN})$ is the optimum solution. Thus, $\pi'_{CP}> \hat{\pi}_{CP}$.
	
	This inequality contradicts the assumption that $(\hat{q}_N,\hat{q}_{NoN})$ is the optimum solution. Thus, the result follows.
\end{proof}

Clearly, the decision of the CP about the vector of qualities depends on the parameter $x_N$ \eqref{l-equ:xn}-I, and subsequently on $n_N$. First, we characterize the candidate strategies of the CP when $0\leq x_N\leq 1$, i.e. $(q^*_N,q^*_{NoN})\in F^I$ and therefore $n_N=x_N$. Then, we consider the case of $x_N< 0$ ($n_N=0$ and $(q^*_N,q^*_{NoN})\in F^L$) and $x_N> 1$ ($n_N=1$ and $(q^*_N,q^*_{NoN})\in F^U$). Finally, we combine both cases to determine the optimum strategies of the CP.  In the following lemma, we characterize the candidate optimum qualities in $F^I$, i.e. the strategies by which $0\leq x_N\leq 1$. 
\begin{lemma}\label{lemma:optimumCPqs}
	If $(q^*_{N},q^*_{NoN})\in F^I$, i.e. optimum strategies are such that $0< x^*_N< 1$, then $q^*_N\in\{0,\tilde{q}_f\}$, $q^*_{NoN}\in\{0,\tilde{q}_{f},\tilde{q}_{p}\}$, $(q^*_N,q^*_{NoN})\neq (0,0)$. 
\end{lemma}
\begin{remark}
	Note that to be in $F^I$ and from \eqref{l-equ:EUs_linear}-I, $(q^*_N,q^*_{NoN})$ should be such that $\frac{\Delta p-t_N}{\kappa_u}< \Delta q^*=q^*_{NoN}-q^*_N< \frac{\Delta p + t_{NoN}}{\kappa_u}$. In Lemma~\ref{lemma:pos}, we have proved that the quality on at least one of the ISPs is equal to a threshold. In this lemma, we prove that the qualities offered on both ISPs are equal to thresholds or one of them is zero. 
\end{remark}
\begin{proof}
	We would like to characterize the optimum qualities in $F^I=F^I_0\bigcup F^I_1$, i.e. optimum strategies for which $0< x_N< 1$. 
	First note that by Lemma~\ref{lemma:pos}, either (a) $q^*_{N}=c$ and $q^*_{NoN}=c+\Delta q$ where $c=\tilde{q}_f$, or (b) $q^*_{NoN}=c$ and $q^*_{N}=c-\Delta q$ where $c\in\{\tilde{q}_{f},\tilde{q}_{p}\}$. Note that the feasible sets for each case can be rewritten as a function of $\Delta q$. We characterize the candidate solutions for each case:
	
	\begin{itemize}
		\item Case (a):  The feasible set for the case (a) is $\Delta q\in G_0=[-c, \tilde{q}_{f}-c]$ (for $z=0$) and $\Delta q\in G_1=(\tilde{q}_{f}-c, \tilde{q}_{p}-c]$  (for $z=1$), where $c=\tilde{q}_f$. Let $G=G_0\cup G_1$. Note that if $0\leq x_N\leq 1$, then $n_N=x_N$ and $n_{NoN}=1-x_N$. Thus, \eqref{equ:CPopt_initial2} can be written as,
		
		\be \label{equ:obj1}
		\ba 
	&	\max_{z,\Delta q \in G=G_0\cup G_1}  \pi_{CP}(c,c+\Delta q,z)=\\
	&\max_{z,\Delta q \in G}  \big{(} t_{NoN}-\kappa_{u}\Delta q+p_{NoN}-p_{N}\big{)} \kappa_{ad}c+\\
		& +\big{(} t_{N}+\kappa_{u}\Delta q+p_{N}-p_{NoN}\big{)} \kappa_{ad}(c+\Delta q)-z\tilde{p}(c+\Delta q)
		\ea 
		\ee
\normalsize		
		
		Note that although the feasible set  $G_1$ is not closed, we used maximum instead of supremum. We will show that the maximum of \eqref{equ:obj1} exists. Thus, the term maximum can be used safely. Note that the objective functions of \eqref{equ:obj1} is a strictly convex functions of $\Delta q$. Note that henceforth wherever we refer to maximum without further clarification, we refer to the solution of \eqref{equ:CPopt_initial2}. 
		
		Let $\tilde{G}_1$  be the closure of $G_1$, then $\tilde{G}_1 \backslash G_1=\{\tilde{q}_{f}-c\}$. First, we prove that the maximum of \eqref{equ:obj1} exists.  Note that $G_0$ and $\tilde{G}_1$ are closed and bounded (compact) and the objective function of \eqref{equ:obj1} is continuous with respect to $\Delta q$ for each $z\in \{0,1\}$. Using Weierstrass Extreme Value Theorem, we can say that a maxima for $\pi_{CP}(c,\Delta q+c,z=0)$ and $\pi_{CP}(c,\Delta q+c,z=1)$ exists in each of two sets  $G_0$ and $\tilde{G}_1$, respectively. Thus, the overall maximum for the objective function of~\eqref{equ:obj1} over $G_0$ and $\tilde{G}_1$ exists. Now, consider two cases:
		
		\begin{enumerate}
			\item If the maxima of  $\pi_{CP}(c,\Delta q+c,z=1)$ in $\tilde{G}_1$ is not $\Delta q=\tilde{q}_{f}-c$, then the maxima is in the original feasible set ($G_1$). Therefore the maximum of \eqref{equ:obj1} exists (since $G_0$ is closed).
			\item  If $\Delta q=\tilde{q}_{f}-c$ is the maxima  of  $\pi_{CP}(c,\Delta q+c,z=1)$ in the set $\tilde{G}_1$, then by Lemma~\ref{lemma:qfree}, the maximum of  $\pi_{CP}(c,\Delta q+c,z=0)$ in the set $G_0$ greater than or equal to the maximum of $\pi_{CP}(c,\Delta q+c,z=1)$  in $\tilde{G}_1$. Thus, the maxima of  \eqref{equ:obj1} over $G_0$ and $G_1$ exists and is in $G_0$. 
		\end{enumerate}
		
		Now, that we have proved the existence of the maximum for \eqref{equ:obj1}, we aim to find all the candidate optimum solutions. Note that  the set $G_0$ is closed. Thus, by the strict convexity of the objective function of \eqref{equ:obj1}, the candidate optimums in $G_0$ are the extreme points of $G_0$. Using the definition of this feasible set, the candidate answers are (i) $q^*_N=\tilde{q}_f$ and $q^*_{NoN}\in\{0,\tilde{q}_{f}\}$.   
		
		Now, consider the feasible set $\tilde{G}_1$, and consider two cases: 
		\begin{enumerate}
			\item If $\Delta q=\tilde{q}_{f}-c$ is not the unique maxima of \eqref{equ:obj1} in $\tilde{G}_1$, then the maxima is in $G_1$ or $G_0$. The candidate answers in the set $G_0$ are already characterized. In addition, by strict convexity of the objective function, the maxima can only be an extreme point of $\tilde{G}_1$. Since $\tilde{q}_{f}-c$ is not the unique maxima of \eqref{equ:CPopt_initial2} in $\tilde{G}_1$, $\tilde{q}_{p}$ is a maxima of \eqref{equ:CPopt_initial2} in $G_1$. Thus, by strong convexity, for all $\Delta q\in G_1$ $\pi_{CP}(c,\tilde{q}_{p},z=1)>\pi_{CP} (c,\Delta q+c,z=1)$, and the only candidate optimum solution over $G_1$ is at $\Delta q=\tilde{q}_{p}-c\in G_1$ which yields   (ii) $q^*_N=\tilde{q}_f$ and $q^*_{NoN}=\tilde{q}_{p}$.
			
			\item If $\tilde{q}_{f}-c$ is the unique maxima in  $\tilde{G}_1$, then  $\pi_{CP}(c,\tilde{q}_{f},z=1)>\pi_{CP} (c,\Delta q+c,z=1)$ for $\Delta q\in G_1$.  By Lemma~\ref{lemma:qfree}, $\pi_{CP}(c,\tilde{q}_{f},z=0)\geq \pi_{CP}(c,\tilde{q}_{f},z=1)$. Therefore, the overall maximum of \eqref{equ:obj1} is in the set $G_0$, and is as characterized previously.
		\end{enumerate}
		
		\item Case (b): The feasible set for the case (b) is $\Delta q\in \hat{G}_0= [c-\tilde{q}_f, c]$ where $c=\tilde{q}_{f}$ (for $z=0)$, and  $\Delta q\in \hat{G}_1=[c-\tilde{q}_f, c]$ where $c=\tilde{q}_{p}$ (for $z=1$). For this case,  \eqref{equ:CPopt_initial2} can be written as:
		\be \label{equ:obj2}
		\ba 
		&\max_{z,\Delta q \in \hat{G}=\hat{G}_0\cup \hat{G}_1} \pi_{CP}(c-\Delta q,c,z)=\\
		&\max_{z,\Delta q \in \hat{G}} \kappa_{ad} \big{(} t_{NoN}-\kappa_{u}\Delta q+p_{NoN}-p_{N}\big{)} (c-\Delta q)+\\
		& \qquad \qquad \quad  +\kappa_{ad}c \big{(} t_{N}+\kappa_{u}\Delta q+p_{N}-p_{NoN}\big{)}-z\tilde{p}c
		\ea 
		\ee
		Note that the feasible set is closed. Thus the term maximum is fine. In addition, the objective functions of \eqref{equ:obj2} are strictly convex functions of $\Delta q$. Thus, using the strict convexity and the definition of the feasible set, i.e. $c-\tilde{q}_f\leq \Delta q^*\leq c$ where $c$ is $\tilde{q}_{f}$ and $\tilde{q}_{p}$, respectively, we can get the other set of candidate answers, (iii) $q^*_{NoN}=\tilde{q}_{f}$ and $\tilde{q}^*_N\in \{0,\tilde{q}_f\}$, and (iv) $\tilde{q}^*_{NoN}=\tilde{q}_{p}$ and $\tilde{q}^*_N\in\{0,\tilde{q}_f\}$. 
	\end{itemize}

	From, (i), (ii), (iii), and (iv), the result follows.
	

	

\end{proof}

The following corollary follows immediately from Lemma  \ref{lemma:optimumCPqs}: 

\begin{corollary}\label{lemma:candidatesCP}
	The possible candidate optimum strategies by which $0< x^*_N< 1$, i.e.  $(q^*_N,q^*_{NoN})\in F^I$, are $(1)$  $(0,\tilde{q}_{f})$, $(2)$ $(\tilde{q}_f,0)$, and $(3)$ $(\tilde{q}_f,\tilde{q}_{f})$ when $z=0$, i.e. $(q^*_N,q^*_{NoN})\in F^I_0$, and $(1)$ $(0,\tilde{q}_{p})$ and $(2)$ $(\tilde{q}_f,\tilde{q}_{p})$ when $z=1$, i.e. $(q^*_N,q^*_{NoN})\in F^I_1$. Note that the necessary and sufficient condition for each of these candidate outcomes to be in $F^I$ is $\frac{\Delta p-t_N}{\kappa_u}< \Delta q^*< \frac{\Delta p + t_{NoN}}{\kappa_u}$.
\end{corollary}

Note that Corollary~\ref{lemma:candidatesCP} lists all the candidate answers by which $0< x_N< 1$. In the next three lemmas,  we focus on the candidate answers when $x_N\geq1$ or $x_N\leq 0$. 

\begin{lemma}\label{lemma:necessary}
	If $\Delta p> \kappa_u \tilde{q}_{f}-t_{NoN}$ then $x_N> 0$ for all choices of $q_{NoN}$ and $q_N$ in the feasible set $F_0$ \footnote{That is $F^L_0$ is an empty set.}. Similarly,  If $\Delta p> \kappa_u \tilde{q}_{p}-t_{NoN}$ then $x_N> 0$ for all choices of $q_{NoN}$ and $q_N$ in the feasible set $F_1$ \footnote{That is $F^L_1$ is an empty set.}.  In addition,  if $\Delta p<  t_N -\kappa_u \tilde{q}_f$ then $x_N< 1$ for all choices of $q_{NoN}$ and $q_N$ in the overall feasible set $\mathcal{F}$\footnote{That is $F^U$ is empty.}.
\end{lemma}

\begin{proof}
	First note that  from \eqref{l-equ:EUs_linear}-I, $x_N> 0$ is equivalent to:
	\begin{equation}\label{equ:equivalent}
	\Delta p > \kappa_u (q_{NoN}-q_N)-t_{NoN}
	\end{equation} 
	
	Consider $\Delta p> \kappa_u \tilde{q}_{f}-t_{NoN}$ (respectively, $\Delta p>\kappa_u \tilde{q}_{p}-t_{NoN}$), if $(q_N,q_{NoN})\in F_0$ (respectively, $(q_N,q_{NoN})\in F_1$) then  $\Delta p > \kappa_u \tilde{q}_{f}-t_{NoN}\geq  \kappa_u(q_{NoN}-q_N)-t_{NoN}$ (respectively, $\Delta p > \kappa_u \tilde{q}_{p}-t_{NoN}\geq \kappa_u(q_{NoN}-q_N)-t_{NoN}$) for every choice of $(q_N,q_{NoN})\in F_0$ (respectively, $(q_N,q_{NoN})\in F_1$). The inequality $\Delta p>\kappa_u(q_{NoN}-q_N)-t_{NoN}$ yields $x_N> 0$. The first result of the lemma follows. 
	
	Now, we prove the second statement. From \eqref{l-equ:EUs_linear}-I, $x_N< 1$ is equivalent to:
	\be\label{equ:equivalent2} 
	\Delta p < t_N+\kappa_u (q_{NoN}-q_N)
	\ee  
	
	Consider $\Delta p<  t_N-\kappa_u \tilde{q}_f$. Note that:
	
	$$
	\Delta p< t_N-\kappa_u \tilde{q}_f\leq t_N+ \kappa_u(q_{NoN}-q_N)
	$$
	for every choice of $0 \leq q_N\leq \tilde{q}_f$ and $0 \leq q_{NoN}\leq \tilde{q}_{p}$ which are all the possible choices in $\mathcal{F}$. The inequality $\Delta p< t_N+ \kappa_u(q_{NoN}-q_N)$ yields that $x_N< 1$. The result follows.
\end{proof}

The following lemma characterizes all the candidate answers when $x^*_N\leq 0$, and characterize the necessary condition on parameters for this solutions to be feasible.

\begin{lemma}\label{lemma:xn<0}
	Let $(q^*_N,q^*_{NoN})\in F^L$. If $(q^*_N,q^*_{NoN})\in F^L_0$ (respectively, if $(q^*_N,q^*_{NoN})\in F^L_1$), then ${q}^*_{NoN}=\tilde{q}_{f}$ (respectively, ${q}^*_{NoN}=\tilde{q}_{p}$). Moreover, for every $x\in [0, \frac{1}{\kappa_u}(\kappa_u \tilde{q}_{f}-t_{NoN}-\Delta p)]$ (respectively, $x\in [0, \frac{1}{\kappa_u}(\kappa_u \tilde{q}_{p}-t_{NoN}-\Delta p)]$) and $\Delta p \leq \kappa_u \tilde{q}_{f}-t_{NoN}$ (respectively, $\Delta p \leq \kappa_u \tilde{q}_{p}-t_{NoN}$), 
	$(x,\tilde{q}_{f})$ (respectively, $(x,\tilde{q}_{p})$) constitutes an optimum solution in $F^L_0$ (respectively, in $F^L_1$).
\end{lemma} 

\begin{proof}
	From \eqref{l-equ:EUs_linear}-I, $x_N\leq  0$ is equivalent to:
	\begin{equation}\label{equ:equivalent}
	\Delta p \leq \kappa_u (q_{NoN}-q_N)-t_{NoN}
	\end{equation} 
	
	Note that from \eqref{l-equ:EUs_linear}-I, if $x_N\leq 0$ then $n_N=0$ and $n_{NoN}=1$. In this case, the payoff of the CP is,
	\be \label{equ:increasingxn<0}
	\pi_G=\kappa_{ad} q_{NoN}-z\tilde{p}q_{NoN}
	\ee
	Note that the value of the payoff is independent of $q_N$ as long as $n_N=0$, and from \eqref{l-equ:EUs_linear}-I $n_N$ is a function of $q_N$ and $q_{NoN}$. In addition, note that if there exist a $q_{NoN}$ that satisfies the constraint  $\Delta p \leq \kappa_u (q_{NoN}-q_N)-t_{NoN}$ (and therefore $n_N=0$) then $q'_{NoN}\geq q_{NoN}$ also satisfies this constraint. Therefore for $q'_{NoN}\geq q_{NoN}$, $n_{N}=0$ and \eqref{equ:increasingxn<0} is true. Note  that from Lemma~\ref{lemma:tildep_condition}, \eqref{equ:increasingxn<0} is an increasing function of $q_{NoN}$. Thus,  if $x_N\leq 0$, then $q^*_{NoN}=\tilde{q}_{f}$ if  $(q^*_N,q^*_{NoN})\in F^L_0$ or $q^*_{NoN}=\tilde{q}_{p}$ if $(q^*_N,q^*_{NoN})\in F^L_1$ (using the feasible sets in Table~\ref{l-table:subsets}-I and their definitions).

	Using \eqref{equ:equivalent}, $(q^*_N,q^*_{NoN})\in F^L_0$ (respectively, $(q^*_N,q^*_{NoN})\in F^L_1$) if and only if,
	
	\footnotesize
	\be \label{equ:local1}
	{q}^*_N\leq \frac{1}{\kappa_u}(\kappa_u \tilde{q}_{f}-\Delta p-t_{NoN}) \quad \bigg{(}\text{respectively, } {q}^*_N\leq \frac{1}{\kappa_u}(\kappa_u \tilde{q}_{p}-\Delta p-t_{NoN})  \bigg{)}
	\ee
	\normalsize
	
	Note that every $q^*_N$ that satisfies \eqref{equ:local1} is an optimum answer since when $(q^*_N,q^*_{NoN})\in F^L$, $n_N=0$ and $q^*_N$ is of no importance. Also, note that $q_N\geq 0$. Thus,  \eqref{equ:local1} is true for at least one $q^*_N$ if $\Delta p\leq \kappa_u \tilde{q}_{f}-t_{NoN}$ (respectively, $\Delta p\leq \kappa_u \tilde{q}_{p}-t_{NoN}$). The result follows.
\end{proof}

The following lemma characterizes all the candidate answers when $x_N\geq 1$, and characterize the necessary condition on parameters for this solutions to be feasible. 

\begin{lemma}\label{lemma:xn>1}
	If $(q^*_N,q^*_{NoN})\in F^U$, i.e. optimum strategies  are such that $x^*_N\geq 1$. Then ${q}^*_{N}=\tilde{q}_f$. Moreover, for all $x\in [0, \frac{1}{\kappa_u}(\kappa_u \tilde{q}_f-t_{N}+\Delta p)]$ and  $\Delta p \geq t_N -\kappa_u \tilde{q}_f$, $(q^*_N,x)$ constitutes an optimum solution in $F^U$.
\end{lemma} 

\begin{proof}
	From \eqref{l-equ:EUs_linear}-I, $x_N\geq 1$ is equivalent to:
	\be\label{equ:equivalent2} 
	\Delta p\geq t_N+\kappa_u (q_{NoN}-q_N)
	\ee

	Now, we prove the first result of the lemma. Note that from \eqref{l-equ:EUs_linear}-I, if $x_N\geq 1$ then $n_N=1$ and $n_{NoN}=0$. In this case, the payoff of the CP is,
	\be \label{equ:increasingxn>1}
	\pi_G=\kappa_{ad} q_{N}
	\ee
	Note that the value of the payoff is independent of $q_{NoN}$ as long as $n_N=1$, and from \eqref{l-equ:EUs_linear}-I, $n_N$ is a function of $q_N$ and $q_{NoN}$. In addition, note that if there exist a $q_{N}$ that satisfies   $\Delta p \geq t_N+\kappa_u (q_{NoN}-q_N)$, then $q'_N\geq q_{N}$ also satisfies this constraint. Therefore, for $q'_N\geq q_N$, $n_N=1$ and \eqref{equ:increasingxn>1} is true. Note that \eqref{equ:increasingxn>1} is an increasing function of $q_{N}$. Thus, ${q}^*_N=\tilde{q}_f$ (using the feasible sets in Table~\ref{l-table:subsets}-I and their definitions). 
	
	Using \eqref{equ:equivalent2}, $(q^*_N,q^*_{NoN})\in F^U$ if and only if:
	\be\label{equ:local2} 
	{q}^*_{NoN}\leq \frac{1}{\kappa_u}(\kappa_u \tilde{q}_{f}-t_{N}+\Delta p) 
	\ee
	Note that every $q^*_{NoN}$ that satisfies \eqref{equ:local2} is an optimum answer since when $(q^*_N,q^*_{NoN})\in F^U$, $n^*_{NoN}=0$ and $q^*_{NoN}$ is of no importance. Also, note that $q_{NoN}\geq 0$. Thus, the condition \eqref{equ:local2} is true for at least one $q^*_{NoN}$ if $\kappa_u \tilde{q}_{f}-t_{N}+\Delta p \geq 0$. The result follows.
\end{proof}

\begin{corollary}\label{corollary:equi_coun}
	 If  $(q^{eq}_N,q^{eq}_{NoN})\in F_0^L$, then $(q^{eq}_N,q^{eq}_{NoN})=(0,\tilde{q}_f)$.   If  $(q^{eq}_N,q^{eq}_{NoN})\in F_1^L$, then $(q^{eq}_N,q^{eq}_{NoN})=(0,\tilde{q}_p)$.	If $(q^{eq}_N,q^{eq}_{NoN})\in F^U$, then $(q^{eq}_N,q^{eq}_{NoN})=(\tilde{q}_f,0)$.
\end{corollary}

\begin{proof}
	Note that when $(q^*_N,q^*_{NoN})\in F^L$ (, respectively $(q^*_N,q^*_{NoN})\in F^U$), then the payoff of the CP is independent of $q^*_N$ and $q^*_{NoN}$. Thus, result of the corollary follows from Tie-Breaking Assumption \ref{l-assumption:tie_n=0}-I.
\end{proof}

\begin{theorem} All possible equilibrium strategies are:
\be \label{equ:candidates_countinuos}  
\ba 
&(0,\tilde{q}_f)\in F^I_0\cup F^L_0\ ,\ (\tilde{q}_f,0) \in F^I_0\cup F^U_0\ ,\ (\tilde{q}_f,\tilde{q}_f) \in F^I_0\ ,\\
& (0,\tilde{q}_{p}) \in F^I_1\cup F^L_1\ ,\  (\tilde{q}_f,\tilde{q}_{p}) \in F^I_1
\ea
\ee 
\end{theorem}

\begin{proof}
	Results follow directly from Corollaries	\ref{lemma:candidatesCP} and \ref{corollary:equi_coun}.
\end{proof}

Note that  \eqref{equ:candidates_countinuos} and \eqref{l-equ:summarize_CP_candidate_new}-I are exactly similar. This implies that the strategies chosen by the CP when she chooses from continuous sets is exactly similar to the strategies when she chooses from the discrete set characterized in our model. This completes our proof. 

\end{document}